\documentclass[11pt,psfig,a4]{article}
\makeatletter
\def\@seccntformat#1{\@ifundefined{#1@cntformat}%
   {\csname the#1\endcsname\quad}  
   {\csname #1@cntformat\endcsname}
}
\let\oldappendix\appendix 
\renewcommand\appendix{%
    \oldappendix
    \newcommand{\section@cntformat}{\appendixname~\thesection:\,\,}
}
\makeatother

\usepackage[T1]{fontenc}
\usepackage{amsfonts}
\usepackage{geometry}
\usepackage{graphics}
\usepackage{lscape}
\usepackage{subfigure}
\usepackage{hyperref}
\usepackage{graphicx}
\usepackage{setspace}
\usepackage{amsthm}
\usepackage{amssymb}
\usepackage{amsmath}
\usepackage{color,soul}
\usepackage{multirow}
\usepackage{hhline}
\usepackage[table]{xcolor}
\usepackage{booktabs}
\usepackage{rotating}
\usepackage[skip=30pt]{caption}
\usepackage[capitalise]{cleveref}
\usepackage{epsfig}
\usepackage{dsfont}
\usepackage{imakeidx}
\usepackage{enumitem}
\usepackage{amsbsy}

\usepackage[round]{natbib}
\PassOptionsToPackage{
        natbib=true,
        style=authoryear-comp,
        hyperref=true,
        backend=biber,
        maxbibnames=99,
        firstinits=true,
        uniquename=init,
        maxcitenames=2,
        parentracker=true,
        url=false,
        doi=false,
        isbn=false,
        eprint=false,
        backref=true,
            }   {biblatex}

\newcommand{\Lyx}{L\kern-.1667em\lower.25em\hbox{y}\kern-.125emX\spacefactor1000}
\newtheorem{proposition}{Proposition}
\newtheorem{theorem}{Theorem}
\newtheorem{definition}{Definition}

\newcommand{\E}{\mathbb{E}}
\newcommand{\round}{\texttt{round}}

\DeclareMathOperator*{\argmin}{arg\,min}

\begin{document}

\title{
\textbf{A return-diversification approach to portfolio selection}}
\author{Francesco Cesarone$^{1*\dagger}$, Rosella Giacometti$^{2\dagger}$, Manuel L Martino$^{1\dagger}$, Fabio Tardella$^{3\dagger}$\\\\
$^1$ Department of Business Studies, Roma Tre University.\\
$^2$ Department of Management, Bergamo University.\\
$^3$ Department of Information Engineering, University of Florence.\\
$*$ Corresponding author(s). E-mail(s): francesco.cesarone@uniroma3.it;\\
Contributing authors: rosella.giacometti@unibg.it; manuelluis.martino@uniroma3.it;\\
fabio.tardella@unifi.it;\\
$\dagger$ These authors contributed equally to this work.}
\date{\today}
\maketitle

\begin{abstract}
In this paper, we propose a general bi-objective model for portfolio selection, aiming to maximize both a diversification measure and the portfolio expected return.
Within this general framework, we focus on maximizing a diversification measure recently proposed by Choueifaty and Coignard for the case of volatility as a risk measure.
We first show that the maximum diversification approach is actually equivalent to the Risk Parity approach using volatility under the assumption of equicorrelated assets.
Then, we extend the maximum diversification approach formulated for general risk measures.
Finally, we provide explicit formulations of our bi-objective model for different risk measures, such as volatility, Mean Absolute Deviation, Conditional Value-at-Risk, and Expectiles, and we present extensive out-of-sample performance results for the portfolios obtained with our model.
The empirical analysis, based on five real-world data sets, shows that the return-diversification approach provides portfolios that tend to outperform the strategies based only on a diversification method or on the classical risk-return approach.


\end{abstract}

\section{Introduction}\label{sec:Intro}

The importance of diversification in portfolio selection has been recently pointed out by several authors \citep[see, e.g.,][]{Meucci2009b,Fabozzi2010,Roncalli2014,lhabitant2017portfolio,cesarone2020optimization}.
A qualitative definition of diversification is very clear to portfolio managers:	
a portfolio is well-diversified if it is not heavily exposed to individual shocks \citep{Meucci2009b}.
In order to formalize a quantitative approach, suitable portfolio diversification measures should be introduced.
One of the first attempts dates back to the Babilonian Talmud \citep[see, e.g.,][]{duchin2009markowitz}, where the uniform (or naive) capital diversification is recommended.
Some early studies on the use of diversification in finance can be found in the works of \cite{leavens1945diversification} and \cite{Bernoulli1954}.
However, only recently more refined notions of portfolio diversification have been proposed and analyzed.
Among them, the most popular approaches, both from a theoretical and from a practical viewpoint, seem to be the Risk Parity model, introduced by \cite{Maillard2010} \citep[see also][]{Roncalli2014},
and the Most Diversified Portfolio, proposed in \cite{Choueifaty2008} and \cite{Choueifaty2013} for the case of the volatility risk measure.
These two approaches provide different diversified portfolio in general.
However, we show here that, in the case of volatility and under the assumption of equicorrelated assets, the Most Diversified Portfolios and the Risk Parity portfolios actually coincide.
This result establishes a first bridge between these hitherto uncorrelated notions of portfolio diversification.
%
%

\noindent
We recall that other diversification approaches have been proposed.
Among them we find entropy \citep[see, e.g.,][]{usta2011mean,carmichael2018rao,pola2016entropy,bera2008optimal,deng2022improved}, extensions of Risk Parity
with general risk measures \citep[see, e.g.,][]{Ararat2021,Bai2015,cesarone2017equal,cesarone2017minimum,Chaves2012,Qian2005,Qian2011,roncalli2014introducing,bellini2021risk}, the Herfindahl index \citep[see, e.g.,][]{coqueret2015diversified,dhingra2023norm}, and the Minimum-Torsion Bets \citep[][]{meucci2015risk}.
For a review of the main contributions on diversification strategies and classes of diversification measures, we refer to \cite{koumou2020diversification}.

\noindent
Typically, in the diversification approach the classical risk-return optimization model is replaced by the quest for a portfolio that maximizes some measure of diversification.
However, the maximization of portfolio diversification could be included as an additional criterion in a Risk-Gain framework as already done by some authors in specific cases \citep[see, e.g.,][]{roncalli2014introducing,usta2011mean,bera2008optimal,deng2022improved}.

\noindent
Other studies have been devoted to bi-objective portfolio selection models where risk is minimized and a diversification measure is maximized \citep[see, e.g.,][]{bera2008optimal,coqueret2015diversified}.
However, we observe that the goals of risk and diversification are strictly related and hence diversification should be used as a substitute of risk rather than a complement.
Indeed, very few attempts have been made to include gain maximization into
the proposed diversification models.
Here we propose a general bi-objective model for portfolio selection where we aim at maximizing both a diversification measure and the expected return.
We first analyze some general properties of this bi-objective model and then we specialize the analysis to the case of a general diversification measure that extends the one proposed in \cite{Choueifaty2013} for the case of volatility.
We provide explicit formulations of our model for the cases of volatility, Mean Absolute Deviation (MAD), Conditional Value-at-Risk (CVaR), and Expectiles risk measures and we present extensive out-of-sample performance results for the portfolios obtained with our models.
The empirical analysis, based on five real-world data sets, shows that the return-diversification approach provides portfolios that tend to outperform the strategies based only on a diversification method or on the classical risk-return approach.

The paper is organized as follows. 
In Section \ref{sec:TheoFram}, we recall some definitions and properties of risk measures and we propose a general return-diversification model.
In Section \ref{sec:DivRatio}, we discuss and extend the maximum diversification ratio approach to the case of a general risk measure.
In Section \ref{sec:DivRatioVsRP}, we recall the Risk Parity approach and we prove its equivalence with the maximum diversification ratio approach in the case of volatility under the assumption of constant correlation.
Section \ref{sec:DR_RiskMeasures} provides explicit linear or quadratic formulations for the problem of maximizing both return and diversification in the case of some specific risk measures typically used for asset allocation.
In Section \ref{sec:EmpAnalysis}, we provide an extensive empirical analysis based on five real-world data sets, which highlights promising out-of-sample performance results of return-diversification portfolio strategies. 
Some concluding remarks are given in Section \ref{sec:Conclusions} together with possible ideas for future research.

\section{A Return-Diversification Model}\label{sec:TheoFram}

We consider a general probability space $\{ \Omega, \mathcal{F}, \mathbb{P} \}$ where all random variables are defined.
We denote by $L^{p}:=L^{p}(\Omega, \mathcal{F}, \mathbb{P})$ the set of all random variables $X$ with $\mathbb{E}[ |X|^{p} ] < +\infty$,
where $p \in [1, +\infty)$.
For our purposes, we are concerned with an investment universe with $n$ tradable assets, and we denote by $R_{i}$ the random variable representing the return of the asset $i$.
For a nonnegative vector $x$, the random portfolio return is defined as $\displaystyle R_P(x)=\sum\limits_{i=1}^n x_i R_{i}$ where $\displaystyle\sum\limits_{i=1}^{n} x_{i}=1$ and $x_i$ is the percentage of capital invested in the asset $i$.
Note that $R_P(x)$ belongs to the convex cone $\mathcal{P}=\{\sum_{i=1}^{n}x_{i}R_{i} : x \in \mathbb{R}_{+}^{n}\}$.
To each random variable $R$ in $\mathcal{P}$ we associate real values $\phi(R),\,\psi(R)$, and $g(R)$, which may be interpreted as diversification, risk, and gain measures respectively.
The diversification measure $\phi(R)$ is often expressed as a function of the risk measure $\psi(R)$.
We recall some general properties of risk measures in terms of the mapping $\rho : \mathbb{R}_{+}^{n} \rightarrow \mathbb{R}$ given by $\rho (x)=\psi(R_{P}(x))$, which can also be viewed as a risk measure for the portfolio represented by $x$. The risk measure $\rho$ is called:
\begin{itemize}
\item[(i)] positive iff $\qquad \rho(x)>0, \qquad \forall x \in \mathbb{R}_{+}^{n} \backslash \{0\}$;
\item[(ii)] convex iff $\qquad\rho(\lambda x+(1-\lambda)y)\le\lambda\rho(x)+(1-\lambda)\rho(y) \qquad\forall x,\,y\in\mathbb{R}_{+}^{n},\,\, \lambda\in[0,1]$;
\item[(iii)] positively homogeneous (of degree 1) iff $\qquad\rho(\lambda x)=\lambda\rho (x), \qquad\forall x\in\mathbb{R}_{+}^{n},\,\,\forall \lambda\in\mathbb{R}_{+}$.
\end{itemize}
We recall the following classical definitions.
\begin{definition}[\cite{ArtDelEbeHea1999}] \label{def:coh}
    A \emph{coherent risk measure} is a risk measure $\rho$ that satisfies Conditions (ii) and (iii), which imply subadditivity, namely $\forall X, Y \in L^{p}$ $\rho(X+Y) \leq \rho(X) + \rho(Y)$. Furthermore, it satisfies monotonicity, namely, $\forall X, Y \in L^{p}, X \leq Y$ implies $\rho(X) \geq \rho(Y)$, and translativity, namely $\forall X \in L^{p}$ and $\alpha \in \mathbb{R}$ $\rho(X+\alpha) = \rho(X) - \alpha$.
\end{definition}
\begin{definition}[\cite{Rockafellar2006}]\label{def:Vol}
    A \emph{deviation risk measure} is a risk measure $\rho$ that satisfies Conditions (ii) and (iii), which imply subadditivity, (i), translational invariance, i.e., $\forall X \in L^{p}$ and $\alpha \in \mathbb{R}$ $\rho(X+\alpha) = \rho(X)$, and normalization, namely $\rho(0)=0$.
\end{definition}
%
%
%
%
%
%
On this basis, we propose to address the problem of maximizing some diversification measure $\phi(R_P(x))$ and, simultaneously, maximizing the portfolio expected return $g(x)=\mu_{P}(x)=\sum\limits_{i=1}^{n}\mu_{i}x_{i}$, where $\mu_{i}$ is the expected return of the asset $i$:
\begin{equation}\label{eq:biobjectiveGeneral}
		\begin{array}{lll}
            \max\limits_{x} & \phi(R_P(x)) & \\[5pt]
            \max\limits_{x} & \mu^{T} x & \\
            \mbox{s.t.} & & \\[5pt]
            & \displaystyle \sum_{i=1}^{n}x_{i}=1 &\\[5pt]
            & x_{i}\ge 0 & i=1,\dots,n\\[5pt]
		\end{array}
\end{equation}
%

\subsection{Maximum Diversification for General Risk Measures}\label{sec:DivRatio}

We now recall, extend, and compare two popular diversification approaches (namely the Maximum Diversification ratio and Risk Parity) associated with a risk measure that satisfies Conditions (i), (ii), and (iii).
\begin{definition}[Diversification Ratio]\label{def:DR}
    For a risk measure $\rho$ that satisfies conditions (i), (ii), (iii), we define the \emph{Diversification Ratio (DR)} as follows:
    \begin{equation*}
        \displaystyle DR(x)=\frac{\sum\limits_{i=1}^{n} x_i \rho(R_i)}{\rho(\sum\limits_{i=1}^{n} x_i R_i)}\,.
    \end{equation*}
\end{definition}
\noindent
Note that, since $\rho$ is convex, we have
\begin{equation}\label{eq:reldiff}
\rho(R_P(x))=\rho(\sum\limits_{i=1}^n x_i R_{i}) \leq \sum\limits_{i=1}^{n}x_{i} \rho(R_{i})=
\sum\limits_{i=1}^{n}x_{i} \rho_{i} \,
\end{equation}
where $\rho_{i}=\rho(R_{i})$. Hence, by construction, $DR(x)\ge 1$. Maximizing DR amounts to maximizing the \emph{relative} distance between the portfolio risk in the worst-case scenario, i.e., when $\rho$ is additive (right-hand side in \eqref{eq:reldiff}), and the generic portfolio risk (left-hand side in \eqref{eq:reldiff}). 

\noindent
Furthermore, we observe that for a risk measure $\rho$ satisfying (i), (ii) and (iii), the greatest Diversification Ratio based on at most $m$ assets is given by:
\begin{equation}\label{eq:DivRatio}
  DR_{m}(\mathcal{P}) = \sup \left\{\frac{\sum_{i=1}^{m} \rho(\alpha_i X_i)}{\rho\left( \sum_{i=1}^{m} \alpha_i X_i \right)}: \,
  X_1, \cdots, X_m \in \mathcal{P}, \, \sum_{i=1}^{m} \alpha_i = 1, \, \alpha_i \geq 0  \right\}\,.
\end{equation}
Here, $DR_{m}(\mathcal{P})$ represents the maximum \emph{relative} diversification risk benefit we can obtain at a random variable $X \in \mathcal{P}$ with respect to a decomposition of $X$ as a convex combination of at most $m$ random variables $X_i \in \mathcal{P}$.
By construction, $DR_{m}(\mathcal{P}) \geq 1$ and $DR_{m}(\mathcal{P}) \le DR_{m+1}(\mathcal{P})$.
We now show that the maximum value of $DR_{m}(\mathcal{P})$ for all $m$ is actually attained when $m=n$ by using the extreme rays $R_{i}$ of the cone $\mathcal{P}$.
%
%
%
\begin{theorem}\label{theo:DR}
  If $\rho$ satisfies (i)+(ii)+(iii), then
\begin{equation}\label{th:DivRatio}
 \sup_{m} DR_{m}(\mathcal{P}) = DR_{n}(\mathcal{P}) = \max \left\{\frac{\sum\limits_{i=1}^{n} \lambda_i \rho(R_i)}{\rho\left( \sum\limits_{i=1}^{n} \lambda_i R_i \right)}: \, \sum_{i=1}^{n} \lambda_i = 1, \, \lambda_i \geq 0  \right\}
\end{equation}
\end{theorem}
\begin{proof}
Given $X_{1},\dots,X_{m}\in\mathcal{P}$, there exist $\beta_{ij}\in\mathbb{R}_{+}$, with $i=1,\dots,m$ and $j=1,\dots,n$, such that $\sum_{j=1}^{n}\beta_{ij}=1$ and $X_{i}=\sum_{j=1}^{n}\beta_{ij}R_{j}$ for $i=1,\dots,m$. 
Furthermore, given $\alpha_{i},\dots,\alpha_{m}$ with $\sum_{i=1}^{m}\alpha_{i}=1$ and setting $\bar\lambda_{j}=\sum\limits_{i=1}^{m}\alpha_{i}\beta_{ij}$, we have
    \begin{equation*}
        \frac{\sum\limits_{i=1}^{m}\rho(\alpha_{i}X_{i})}{\rho(\sum\limits_{i=1}^{m}\alpha_{i}X_{i})}\le\frac{\sum\limits_{j=1}^{n}\sum\limits_{i=1}^{m}\alpha_{i}\beta_{ij}\rho(R_{j})}{\rho(\sum\limits_{j=1}^{n}\sum\limits_{i=1}^{m}\alpha_{i}\beta_{ij}R_{j})}=\frac{\sum\limits_{j=1}^{n}\bar\lambda_{j}\rho(R_{j})}{\rho(\sum\limits_{j=1}^{n}\bar\lambda_{j}R_{j})}\le\max\left\{\frac{\sum\limits_{i=1}^{n}\lambda_{i}\rho(R_{i})}{\rho(\sum\limits_{i=1}^{n}\lambda_{i}R_{i})}:\sum\limits_{i=1}^{n}\lambda_{i}=1,\,\lambda_{i}\ge 0\right\}\,,
    \end{equation*}
    where the first inequality follows from assumptions (ii) and (iii) which imply $\rho(\alpha_{i}X_{i})\le\sum_{j=1}^{n}\alpha_{i}\beta_{ij}\rho(R_{j})$ for all $i$.
\end{proof}

\noindent
In financial terms, Theorem \eqref{theo:DR} implies that the greatest diversification benefit that can be obtained by decomposing a portfolio $X$ as a convex combination of any number of portfolios $X_{1},\dots,X_{m}$, is actually attained by decomposing $X$ as a convex combination of the $n$ assets $R_{1},\dots,R_{n}$, and hence by maximizing the diversification ratio $DR(x)$ for $x\in\Delta$, where $\Delta=\{x\in\mathbb{R}^{n}:\sum\limits_{i=1}^{n}x_{i}=1,\,x\ge 0\}$.
Thus, one can consider the maximum diversification portfolio $x^{MD}$ as the one maximizing $DR(x)$ over $\Delta$.


\subsection{Maximum Diversification and Risk Parity}\label{sec:DivRatioVsRP}

We now recall the diversification approach based on the general notion of equal risk contribution from each asset which has been originally proposed by \cite{Qian2005,Qian2011,Maillard2010} for the case of the volatility risk measure.

\noindent
For a differentiable risk measure $\rho(x)=\psi(R_{P}(x))$ satisfying Conditions (i), (ii), (iii), Euler's theorem guarantees that the total risk can be decomposed into the sum of the total risk contributions of each asset
\begin{equation*}
    \rho({x})=\sum_{i=1}^{n}TRC_{i}(x) \,
\end{equation*}
where the total risk contribution of asset $i$ is
\begin{equation*}
 TRC_{i}(x)= x_i\frac{\partial\rho({x})}{\partial x_i}   
\end{equation*}
The Risk Parity portfolio is the (unique) portfolio $x^{RP}\in\Delta$ that satisfies
\begin{equation*}
TRC_{i}(x) = TRC_{j}(x) \quad \forall i, j\,.
\end{equation*}
Equivalently, the Risk Parity portfolio can be obtained as a solution of the system
\begin{eqnarray*}
     TRC_{i}(x)&=&\lambda^{RP}\,, \qquad \mbox{with} \, \, \, i=1,\dots,n  \\
     x\in&\Delta,&\lambda^{RP}\ge 0
\end{eqnarray*}
By Euler's theorem, $\lambda^{RP}=\displaystyle\frac{\rho(x)}{n}$ is the equal contribution to the risk from each asset.

\noindent
In the general case, the Risk Parity portfolio and the maximum diversification portfolio are different. However, we will now show that they coincide for the case of volatility under the assumption of constant correlation.
This hypothesis is supported by some scholars \citep[see, e.g.,][]{elton1973estimating,elton1976simple,elton2006improved,engle2012dynamic} who suggest to estimate the covariance matrix with the constant correlation model, i.e., $\sigma_{ij}=c \sigma_{i}\sigma_{j}$ and $\sigma_{ii}= \sigma_{i}^2$. 
In the case of the variance risk measure with constant correlation, \cite{Maillard2010} show that the Risk Parity portfolio is given by $\displaystyle x_{i}^{RP}=\frac{\sigma_{i}^{-1}}{\sum\limits_{j=1}^{n} \sigma_{j}^{-1}}$. 
The following result shows that this coincides with the maximum diversification portfolio under the same assumptions.
\begin{theorem}[]
    If $\rho(x)=\sum\limits_{i=1}^{n}\sum\limits_{j=1}^{n}x_{i}x_{j} c_{ij}
  \sigma_i \sigma_j$, with $c_{ii}=1$ and $c_{ij}=k$ for all $i\neq j$, then a maximum diversification portfolio is given by $x_{i}^{MD}=\displaystyle\frac{\sigma_{i}^{-1}}{\sum\limits_{j=1}^{n} \sigma_{j}^{-1}}$. 
\end{theorem}
\begin{proof}
Assume that $\sigma_{ij}=k \sigma_{i}\sigma_{j}$ and $\sigma_{ii}= \sigma_{i}^2$. Following \cite{Choueifaty2013} (see p. 18 - Appendix B, and also Section \ref{sec:DR_RiskMeasures} below for the general case), the maximum diversification portfolio in the case of volatility can be obtained by finding a solution $y^{*}$ of the following problem
\begin{equation}
\begin{array}{lll}
\displaystyle\min_{y} & y'\Sigma y &
\\
s.t. &  &  \\
& \sigma'y=1 &  \\
& y\geq 0 %
\end{array}
\label{mark0}
\end{equation}
and taking $x_{i}^{MD}=\displaystyle\frac{y_{i}^{*}}{\sum\limits_{j=1}^{n} y_{j}}$ with $i=1,\dots,n$.
Note that, under the assumption of the constant correlation model, the covariance matrix $\Sigma$ is positive definite.
Hence, Problem \eqref{mark0} has a unique solution $y^{*}$.
If $ y^{*}>0$, then it must satisfy the following KKT conditions:
\begin{equation}\label{mark}
\centering
\left\{
	\begin{array}{lll}
(\Sigma y)_i+\lambda \sigma_i&=0 \quad \forall i \\
\sigma'y=1
	\end{array}
\right.
\end{equation}
for some $\lambda \in \mathbb{R}$.

\noindent
Consider the points $\displaystyle y^k=k\left(\frac{1}{\sigma_1},\dots,\frac{1}{\sigma_n}\right)$ with $k \in \mathbb{R}$ and substitute them into the first equation of \eqref{mark}. 
We obtain
\begin{equation*}
 (\Sigma y^k)_i+\lambda \sigma_i= (\sigma_i(k(1+(n-1)c)))+\lambda \sigma_i=0 \quad \forall i\,.
\end{equation*}
Then, solving for $\lambda$ we have
\begin{equation*}
\lambda=-(k(1+(n-1)c))\,.  
\end{equation*}
Hence $y^k$ satisfies the second constraint of Problem \eqref{mark} for $\displaystyle k=\frac{1}{n}$.
Thus, $\displaystyle y^*=\frac{1}{n}\left(\frac{1}{\sigma_1},\dots,\frac{1}{\sigma_n}\right)$ is a solution of Problem \eqref{mark}, with $\displaystyle\lambda=-\frac{1+(n-1)c}{n}$. 
Thus, by strict convexity, $y^{*}$ is also the unique solution of Problem \eqref{mark0}.
The maximum diversification portfolio is then obtained by normalizing $y^{*}$, namely $\displaystyle x_{i}^{MD}=\frac{\sigma_{i}^{-1}}{\sum\limits_{j=1}^{n} \sigma_{j}^{-1}}$.
\end{proof}

\section{Linear and Quadratic Formulations for the Return-Diversification Model with Specific Risk Measures}\label{sec:DR_RiskMeasures}

In the case where diversification is measured with the diversification ratio \eqref{def:DR}, Problem \eqref{eq:biobjectiveGeneral} becomes
\begin{equation}\label{eq:biobjective}
		\begin{array}{lll}
            \max\limits_{x} & DR(x) & \\[5pt]
            \max\limits_{x} & \mu^{T} x & \\
            \mbox{s.t.} & & \\[5pt]
            & \displaystyle x\in\Delta
		\end{array}
\end{equation}
The Pareto-optimal solutions of Problem \eqref{eq:biobjective} can be obtained with the $\varepsilon-$constraint method 
as follows
\begin{equation}\label{eq:biobjective_b}
		\begin{array}{lll}
            \max\limits_{x} & \displaystyle DR(x)=\frac{\sum\limits_{i=1}^{n} x_i \rho(R_i)}{\rho(\sum\limits_{i=1}^{n} x_i R_i)} & \\[5pt]
            \mbox{s.t.} & & \\[5pt]
            & \displaystyle\sum_{i=1}^{n}\mu_{i}x_{i}\ge\eta &\\[5pt]
            & x\in\Delta \displaystyle 
		\end{array}
\end{equation}
where $\eta$ is the required level of the portfolio expected return.
Problem \eqref{eq:biobjective_b} is a Fractional Programming (FP) problem, and is
a nonconvex problem.
However, by condition (ii), Problem \eqref{eq:biobjective_b} belongs to the class of concave-convex FP problems, that can be tackled by means of the Schaible transform \citep{schaible1974} described below.

\noindent
For a long-only portfolio $x \geq 0$, we set  $y=tx$, where  $\displaystyle t=\frac{1}{\sum\limits_{i=1}^{n} x_{i}\rho_{i}}$.
Clearly, if $\boldsymbol{\rho}(\cdot)$ is a positive function, then $t>0$.
Thus, Problem \eqref{eq:biobjective_b} can be reformulated as follows
\begin{equation}\label{eq:MDP2}
\begin{array}{lll}
\min\limits_{x} & \displaystyle\frac{1}{DR(x)}= \frac{\boldsymbol{\rho}(x)}{\sum\limits_{i=1}^{n} x_{i}\rho_{i}} = t \boldsymbol{\rho}(x) = \boldsymbol{\rho}(tx) & \\[5pt]
s.t. &  & \\[5pt]
& \displaystyle\sum\limits_{i=1}^{n}y_{i} = t &  \\[5pt]
& \displaystyle\sum_{i=1}^{n}\mu_{i}\frac{y_{i}}{t}\ge\eta &\\[5pt]
& \displaystyle\sum\limits_{i=1}^{n}y_{i}\rho_{i}=1 &  \\[5pt]
& x_{i} =\displaystyle\frac{1}{t} y_{i} \geq 0 & i=1,\ldots ,n 
%
\end{array}
\end{equation}
%
By eliminating $x$ and $t$ from Problem \eqref{eq:MDP2}, we obtain
\begin{equation}\label{eq:MDR1}
\begin{array}{lll}
\displaystyle\min_{y} & \rho(y) &
\\[5pt]
s.t. &  &  \\[5pt]
& \displaystyle\sum\limits_{i=1}^{n}y_{i}\rho_{i}=1 &  \\[5pt]
& \displaystyle\sum\limits_{i=1}^{n}\mu_{i}y_i\ge (\displaystyle\sum_{i=1}^{n} y_i)\eta &  \\[5pt]
& y_{i}\geq 0 & i=1,\ldots ,n%
\end{array}
\end{equation}

\noindent
%
%
\noindent
Note that Problem \eqref{eq:MDR1} is a simple convex programming problem with linear constraints which can be efficiently solved even for large instances.
Since $\rho$ is a homogenous function of degree 1,
we find the portfolio weights that maximize the diversification ratio in Problem \eqref{eq:biobjective_b} by simply normalizing the solution $y^{*}$ of Problem \eqref{eq:MDR1}.
Thus, $x_{i}^{MD}=\displaystyle\frac{y_{i}^{*}}{\sum\limits_{k=1}^{n}y_{k}^{*}}$, for $i=1,\ldots ,n$.

\noindent
We are interested in computing the optimal values $h(\eta)$
of Problem \eqref{eq:biobjective_b} as a function of $\eta$ in the interval $[\eta_{min},\eta_{max}]$, where $\eta_{min}$ denotes the value of $\sum\limits_{i=1}^{n}\mu_{i}x_{i}$ at an optimal solution of the problem obtained by deleting the return constraint in \eqref{eq:biobjective_b}, and $\eta_{max}=\max\{\mu_{1},\dots,\mu_{n}\}$.
Indeed, the graph of $h(\eta)$ on the interval $[\eta_{min},\eta_{max}]$ contains the set of all non-dominated portfolios (efficient frontier), which is usually approximated by solving \eqref{eq:biobjective_b} for several (equally spaced) values of $\eta$ in $[\eta_{min},\eta_{max}]$.

We now describe the explicit formulation of Problem \eqref{eq:MDR1} in the case of some specific risk measures typically used for asset allocation.
%

\subsection{Volatility}

Let $\sigma_i$ be the volatility of asset $i$ and $c_{ij}$ be the (linear) correlation
between assets $i$ and $j$.
The portfolio volatility is defined as
\begin{equation}\label{eq:Voldef}
\sigma(R_P(x))=\sqrt{\sum\limits_{i=1}^{n}\sum\limits_{j=1}^{n}x_{i}x_{j} c_{ij}
  \sigma_i \sigma_j} 
\end{equation}
Note that volatility is a deviation risk measure satisfying Definition \ref{def:Vol}.
Hence, by condition (ii)
\begin{equation}\label{eq:VolJensen2}
\sigma(x) \leq \sum\limits_{i=1}^{n}x_{i} \sigma_i \, .
\end{equation}
When the asset returns are perfectly positively correlated, i.e., $c_{ij}=1$,
$\sigma(x)$ is additive.

\noindent
In the case of volatility, Problem \eqref{eq:MDR1} becomes

%
	%
	\begin{equation}
		\begin{array}{lll}
		\displaystyle\min_{y} & \sum\limits_{i=1}^{n}\sum\limits_{j=1}^{n}y_{i}y_{j} c_{ij}
  \sigma_i \sigma_j &
		\\
		s.t. &  &  \\
		& \sum\limits_{i=1}^{n}y_{i} \sigma_i =1 &  \\
        & \sum\limits_{i=1}^{n} \mu_i y_{i} \ge \eta \left( \sum\limits_{i=1}^{n}y_{i} \right) &  \\
		& y_{i}\geq 0 & i=1,\ldots ,n%
		\end{array}
		\label{eq:VolMDR}
	\end{equation}
	%

\subsection{MAD}

The Mean Absolute Deviation (MAD) is defined as the expected value of the absolute deviation of the portfolio return $R_{P}(x)$ from its mean $\mu_{P}(x)$, namely
\begin{equation}  \label{equ0}
MAD(x) = E[\,\lvert R_{P}(x) - \mu_{P}(x)\rvert \,] = E\left[\,\left\lvert\sum_{i=1}^{n} R_{i} x_{i} -\sum_{i=1}^{n} \mu_{i} x_{i}\right\rvert \,\right] \, ,
\end{equation}
where $R_{i}$ is the random return of asset $i$ and $\mu_{i}$ is its expected value.
Note that MAD is a deviation risk measure satisfying Definition \ref{def:Vol}.
Hence, by condition (ii)
\begin{equation}\label{eq:MADconexity}
  MAD(x) \leq \sum\limits_{i=1}^{n}x_{i} MAD(R_{i}) \, .
\end{equation}
Expression \eqref{eq:MADconexity} holds with equality, i.e., MAD is additive, under the conditions specified in the following proposition.
%
%
\begin{proposition}[\cite{Ararat2021}]
Let $R_{1}, R_{2}, \cdots, R_{n} \in L^{1}(\Omega, \mathcal{F}, \mathbb{P})$.
Then, MAD is additive if and only if for each $i \neq j$ the asset returns satisfy the following
conditions:
\begin{equation}\label{eq:MADadd}
 \left( R_{i} - E[R_{i}] \right) \left( R_{j} - E[R_{j}] \right) \geq 0 \qquad \mathbb{P} - a.s.
\end{equation}
\end{proposition}

\noindent
We assume that the asset returns are discrete random variables defined on a discrete state space with $T$ states of nature having probability $\pi_t$, with $t=1,...,T$.
The outcomes of the discrete random returns correspond to the historical scenarios.
More precisely, the portfolio choice is performed using an in-sample window of $T$ historical realizations having equal probability $\pi_t=1/T$, as is typically done in portfolio optimization \citep[see, e.g.,][and references therein]{carleo2017approximating,cesarone2020computational,bellini2021risk}.

\noindent
In the case of MAD, Problem \eqref{eq:MDR1} becomes
\begin{equation}
		\begin{array}{lll}
		\displaystyle\min_{y} & \displaystyle \frac{1}{T} \sum_{t=1}^{T} \left\lvert\sum_{i=1}^{n} r_{i,t} y_{i} -\sum_{i=1}^{n} \mu_{i} y_{i}\right\rvert & \\
		s.t. &  &  \\
		& \sum\limits_{i=1}^{n}y_{i} MAD_{i}=1 &  \\
        & \sum\limits_{i=1}^{n} \mu_i y_{i} \ge \eta \left( \sum\limits_{i=1}^{n}y_{i} \right) &  \\
		& y_{i}\geq 0 & i=1,\ldots ,n%
		\end{array}
		\label{eq:MADMDR}
\end{equation}
where  $r_{i,t}$ is the historical realization of the return of asset $i$ at time $t$, and $MAD_{i}=MAD(R_{i})= \frac{1}{T} \sum_{t=1}^{T} \left\lvert  r_{i,t} - \mu_{i}  \right\rvert$.

\noindent
As in \cite{Konno91},
we can linearize Problem \eqref{eq:MADMDR} by introducing $T$ auxiliary variables
$d_{t}$ defined as the absolute deviation of the portfolio return from its mean for each $t=1, \ldots, T$, i.e., $ d_{t} = \left\lvert\sum_{i=1}^{n} r_{i,t} y_{i} -\sum_{i=1}^{n} \mu_{i} y_{i}\right\rvert$,
and by replacing the absolute value with the following constraints: 
$d_{t} \geq \sum_{i=1}^{n} (r_{i,t} - \mu_i) y_{i}$ and
$d_{t} \geq -\sum_{i=1}^{n} (r_{i,t} - \mu_i) y_{i}$.
Hence, Problem \eqref{eq:MADMDR} can be rewritten as the following LP problem
%
\begin{equation}  \label{eq:MADMDR2}
\begin{array}{lll}
\displaystyle \min_{y,d} & \quad \displaystyle\frac{1}{T} \sum\limits_{t=1}^{T} d_{t} &\\
\mbox{s.t.} &  &  \\
& \displaystyle d_{t} \geq \sum_{i=1}^{n} (r_{i,t} - \mu_i) y_{i} &t=1,\ldots,T \\
& \displaystyle d_{t} \geq -\sum_{i=1}^{n} (r_{i,t} - \mu_i) y_{i} &t=1,\ldots,T \\
& \sum\limits_{i=1}^n y_{i}MAD_{i} = 1  &  \\
& \sum\limits_{i=1}^{n} \mu_i y_{i} \ge \eta \left( \sum\limits_{i=1}^{n}y_{i} \right) &  \\
& y_i\geq 0 & i=1,\ldots,n
\end{array}%
\end{equation}
%


\subsection{CVaR}

The formal definition of the Conditional Value-at-Risk (CVaR), at a specified confidence level $\varepsilon$, is
$$
CVaR_{\varepsilon}(x) = - \frac{1}{\varepsilon} \int \limits_{0}^{\varepsilon} Q_{R_{P}(x)}(\alpha) d \alpha \, ,
$$
where $Q_{R_{P}(x)}(\alpha)$ is the quantile function of the portfolio return $R_{P}(x)$.
Note that CVaR is a coherent risk measure satisfying Definition \ref{def:coh} and, in particular, the convexity condition (ii)
\begin{equation}\label{eq:CVaRconvexity}
  CVaR_{\varepsilon}(x) \leq \sum\limits_{i=1}^{n}x_{i} CVaR_{\varepsilon}(R_{i}) \, ,
\end{equation}
where  $CVaR_{\varepsilon}(R_{i})$ represents CVaR of asset $i$.
Expression \eqref{eq:CVaRconvexity} holds with equality, i.e., CVaR is additive in the case where the random variables $R_i$ are comonotonic  \citep{tasche2002expected}.
\noindent
Similar to the case of MAD, in a discrete world, Problem \eqref{eq:MDR1} with CVaR as a risk measure becomes
	\begin{equation}
		\begin{array}{lll}
		\displaystyle\min_{y} & \displaystyle CVaR_{\varepsilon}(y) & \\
		s.t. &  &  \\
		& \sum\limits_{i=1}^{n}y_{i} CVaR_{\varepsilon}(R_{i})=1 &  \\
        & \sum\limits_{i=1}^{n} \mu_i y_{i} \ge \eta \left( \sum\limits_{i=1}^{n}y_{i} \right) &  \\
		& y_{i}\geq 0 & i=1,\ldots ,n%
		\end{array}
		\label{eq:CVaRMDR}
	\end{equation}
where the Conditional Value-at-Risk of asset $i$ is given by $CVaR_{\varepsilon}(R_{i}) = - \frac{1}{j} \sum_{t=1}^{j}  r_{i,(t)}$, with $j=\round(\varepsilon T)$, while $r_{i,(1)} \leq \ldots \leq r_{i,(T)}$ represent the ordered realizations of the returns of asset $i$.
%
%
%
%
We recall the LP reformulation of CVaR by \cite{RockUrya00} for a given portfolio $y$
\begin{align}
CVaR_{\varepsilon}(y):= &   \min_{\zeta \in \mathbb{R}, d} \;  \zeta + \frac{1}{\varepsilon T} \sum_{t=1}^T d_t \nonumber \\
& \mbox{s.t.}  \hspace*{0.1cm} d_t \ge - \sum_{i=1}^{n} r_{i, t} y_{i} -\zeta  \quad t=1,\ldots,T, \nonumber \\
& \hspace*{.6cm} d_t\ge 0,\; \forall t \nonumber
\end{align}
Hence, Problem \eqref{eq:CVaRMDR} can be rewritten as the following LP problem
	\begin{equation}
		\begin{array}{lll}
		\displaystyle \min_{y,\zeta,d} & \displaystyle \zeta + \frac{1}{\varepsilon T} \sum_{t=1}^{T} d_{t} & \\
		s.t. &  &  \\
& d_{t} \geq -\sum\limits_{i=1}^{n} r_{i,t} y_{i} -\zeta \qquad & t=1,\ldots,T \\
& d_{t} \geq 0 & t=1,\ldots,T \\
		& \sum\limits_{i=1}^{n}y_{i} CVaR_{\varepsilon}(R_{i}) =1 &  \\
        & \sum\limits_{i=1}^{n} \mu_i y_{i} \ge \eta \left( \sum\limits_{i=1}^{n}y_{i} \right) &  \\
		& y_{i}\geq 0 & i=1,\ldots ,n  \\
        & \zeta \in \mathbb{R} &
		\end{array}
		\label{eq:CVaRMDR2}
	\end{equation}
	%
 

\subsection{Expectiles}
Expectiles have been defined by \cite{newey1987asymmetric} as the minimizers of the expected value of an asymmetric piecewise quadratic loss function
\begin{equation}  \label{Expectiles0}
e_{\alpha}(X)=  \argmin_{\kappa \in \mathbb{R}} \E [ \alpha (X - \kappa)_{+}^{2}
+ (1- \alpha)  (X - \kappa)_{-}^{2} ] \, ,
\end{equation}
where $X \in L^2$, $\alpha \in (0,1)$, and $(b)_{+}=\max(b,0)$, $(b)_{-}=-\min(b,0)$. 
The minimizer in \eqref{Expectiles0} is unique and determined by the first-order condition
\begin{equation} \label{Expectiles1}
\alpha \E[(X-e_\alpha(X))_+] = (1-\alpha) \E[(X-e_\alpha(X))_-],
\end{equation}
which is indeed an equivalent definition of $e_\alpha(X)$ valid for each $X \in L^1$ \citep[see, e.g.,][]{bellini2014generalized}.
We denote by $\zeta_{\alpha}(x)= \zeta_{\alpha}(R_{p}(x)) = e_\alpha(-R_p(x))$ the $\alpha$-expectile of the portfolio loss which, as usual, is defined as the opposite of the portfolio return.
For $\alpha \in [1/2,1)$, Expectiles satisfy Definition \ref{def:coh}, and hence condition (ii)
\begin{equation}\label{eq:Econvexity}
\zeta_{\alpha}(R_P(x)) \leq
 \sum\limits_{i=1}^{n} x_{i}\zeta_{\alpha}(R_{i})\,.
\end{equation}
%
%
Expression \eqref{eq:Econvexity} holds with equality, i.e., Expectile is additive, under the conditions specified in the following theorem.
\begin{theorem}[\cite{bellini2021risk}]
Let $R_{1}, R_{2}, \cdots, R_{n} \in L^{1}(\Omega, \mathcal{F}, \mathbb{P})$, and $\alpha\in(1/2,1)$.
Then, Expectile is additive if and only if for each $i \neq j$ the asset returns satisfy the following
conditions:
\begin{equation}\label{eq:}
 \mathbb{P} \left( (R_{i} - e_{\alpha}(R_{i})) (R_{j} - e_{\alpha}(R_{j})) \ge 0\right) = 1 \qquad \mathbb{P} - a.s.
\end{equation}
\end{theorem}
%

%

\noindent
Again, in a discrete world and in the case of Expectile, Problem \eqref{eq:MDR1} becomes
	\begin{equation}
		\begin{array}{lll}
		\displaystyle\min_{y} & \displaystyle \zeta_{\alpha}(y) & \\
		s.t. &  &  \\
		& \sum\limits_{i=1}^{n}y_{i} \zeta_{\alpha}(R_{i})=1 &  \\
        & \sum\limits_{i=1}^{n} \mu_i y_{i} \ge \eta \left( \sum\limits_{i=1}^{n}y_{i} \right) &  \\
		& y_{i}\geq 0 & i=1,\ldots ,n%
		\end{array}
		\label{eq:DRExpectile}
	\end{equation}
where  $\zeta_{\alpha}(R_{i})$ represents the $\alpha$-Expectile of asset $i$.
%
%

\noindent
As described in \cite{bellini2021risk}, for a fixed portfolio $y$, its $\alpha$-Expectile $\zeta_{\alpha}(y)=\zeta_{\alpha}$ can be obtained by solving the following equation
%
%
%
%
\begin{equation}
\alpha\sum_{t=1}^{T}\left (-\sum_{i=1}^{n} y_{i}r_{i,t} - \zeta_{\alpha} \right)_{+} \,=
(1- \alpha) \sum_{t=1}^{T}\left (-\sum_{i=1}^{n}  y_{i} r_{i,t} - \zeta_{\alpha} \right)_{-}\,.
\end{equation}
Therefore, Problem \eqref{eq:DRExpectile} can be rewritten as follows
	\begin{equation}
		\begin{array}{lll}
		\displaystyle\min_{y,\zeta_{\alpha}} & \displaystyle \zeta_{\alpha} & \\
		s.t. &  &  \\
  & \alpha\sum\limits_{t=1}^{T}\left (-\sum\limits_{i=1}^{n} y_{i}r_{i,t} - \zeta_{\alpha} \right)_{+} \,=
(1- \alpha) \sum\limits_{t=1}^{T}\left (-\sum\limits_{i=1}^{n}  y_{i} r_{i,t} - \zeta_{\alpha} \right)_{-} &  \\
		& \sum\limits_{i=1}^{n}y_{i} \zeta_{\alpha}(R_{i})=1 &  \\
        & \sum\limits_{i=1}^{n} \mu_i y_{i} \ge \eta \left( \sum\limits_{i=1}^{n}y_{i} \right) &  \\
		& y_{i}\geq 0 & i=1,\ldots ,n%
		\end{array}
		\label{eq:DRExpectile_y}
	\end{equation}
We can linearize Problem \eqref{eq:DRExpectile_y} by replacing the positive and negative parts of $\displaystyle-\sum\limits_{i=1}^{n} y_{i} r_{i,t} - \zeta_{\alpha}$ with $2T$ auxiliary variables $d_{t}^{+},d_{t}^{-}\geq 0$ satisfying $\displaystyle d_{t}^{+}-d_{t}^{-}=-\sum\limits_{i=1}^{n} y_{i} r_{i,t} - \zeta_{\alpha}$ and $d_{t}^{+}\cdot d_{t}^{-}=0$.
We note that, by using an argument similar to the one of Theorem 5 in \cite{bellini2021risk}, one can show that the complicating complementarity constraints $d_{t}^{+}\cdot d_{t}^{-}=0$ are always satisfied and can thus be dropped from the formulation.
%
Hence, Problem \eqref{eq:MDR1} for Expectiles can be reformulated as the following LP problem
\begin{equation}
\label{eq:ExpeMDR3}
\begin{array}{lll}
\displaystyle \min_{y,d^{+},d^{-},\zeta_{\alpha}} & \zeta_{\alpha}  &\\
s.t. &  &  \\
& \displaystyle\alpha \sum\limits_{t=1}^{T} d_{t}^{+} = (1-\alpha)\sum\limits_{t=1}^{T} d_{t}^{-} & \\
& \displaystyle-\sum\limits_{i=1}^{n} y_{i}r_{it} -\zeta_{\alpha} = d_{t}^{+} - d_{t}^{-}  & t=1,\ldots,T \\
& \displaystyle\sum\limits_{i=1}^{n}y_{i} (\zeta_{\alpha})_{i}=1 &  \\
& \displaystyle\sum\limits_{i=1}^{n} \mu_i y_{i} \ge \eta \left( \sum\limits_{i=1}^{n}y_{i} \right) &  \\
& y_{i}\geq 0 & i=1,\ldots ,n  \\
& d_{t}^{+},d_{t}^{-} \geq 0 & t=1,\ldots,T \\
& \zeta_{\alpha} \in \mathbb{R} &
\end{array}
\end{equation}
%


\section{Empirical Analysis}\label{sec:EmpAnalysis}

In this section, we present an extensive empirical analysis based on five real-world data sets.
More precisely, Section \ref{subsec:ExpSetup} describes the datasets, the setup of the experiments and of the portfolio strategies, and the performance measures used to evaluate the results.
In Section \ref{subsec:PerfEval}, we discuss the out-of-sample computational results obtained.
%

All experiments have been implemented on a workstation with Intel(R) Xeon(R) E5-2623 v4 CPU @ 2.60GHz processor and 64 GB of RAM, under MS Windows 10 Pro, using MATLAB R2022b and the GUROBI 9.5.1 optimization solver.


\subsection{Experimental Setup}\label{subsec:ExpSetup}

The empirical analysis is based on a rolling time window scheme of evaluation. We consider in-sample windows of 2 years (i.e., 500 observations), and we choose one financial month both as a rebalancing interval and as a holding period.
In Table \ref{tab:DailyDatasets}, we list the real-world data sets obtained from Refinitiv, which consist of daily prices adjusted for dividends and stock splits for some of the major market indexes, and which are publicly available at the website \url{https://www.francescocesarone.com/data-sets}.
\begin{table}[htbp!]
	\small
	\centering
	\scalebox{1.}{\begin{tabular}{l l l l l l l l l l l l l l l l l l}
			\toprule
			\textbf{Index} & & & & &  \textbf{$\#$Assets} & & & & \textbf{Country} & & & & \textbf{Time Interval}\\
			\midrule
			\text{DowJones} & & & & & 28 & & & & \text{USA} & & & & \text{Oct} 2006--\text{Feb} 2023\\
			\text{EuroStoxx50} & & & & & 46 & & & & \text{EU} & & & & \text{Oct} 2006--\text{Feb} 2023\\
			\text{FTSE100} & & & & & 82  & & & & \text{UK} & & & & \text{Oct} 2006--\text{Feb} 2023 \\
			\text{NASDAQ100} & & & & & 70 & & & & \text{USA} & & & & \text{Oct} 2006--\text{Feb} 2023 \\
			\text{SP500} & & & & & 420 & & & & \text{USA} & & & & \text{Oct} 2006--\text{Feb} 2023 \\
			\bottomrule
	\end{tabular}}
	\caption{List of data sets analyzed}
	\label{tab:DailyDatasets}
\end{table}
%

In our experiments, we test two classes of portfolio selection models: (\textit{i}) the return-diversification strategy  \eqref{eq:biobjective_b}, and (\textit{ii}) the return-risk approach reported below \citep[see, e.g.,][]{cesarone2020computational,bellini2021risk}.
\begin{equation}\label{eq:gain-risk}
		\begin{array}{lll}
            \min\limits_{x} & \rho(x) & \\[5pt]
            \mbox{s.t.} & & \\[5pt]
            & \displaystyle\sum_{i=1}^{n}\mu_{i}x_{i}\ge\eta &\\[5pt]
            & x\in\Delta \displaystyle 
		\end{array}
\end{equation}
These two classes of models are applied to the four risk measures described in Section \ref{sec:DR_RiskMeasures}.
%
%
For each model, we choose two target levels $\eta$ of portfolio expected return in the interval $[\eta_{min},\eta_{max}]$, where $\eta_{min}$ denotes the value of $\sum\limits_{i=1}^{n}\mu_{i}x_{i}$ at an optimal solution of the problem obtained by deleting the return constraint in \eqref{eq:biobjective_b} and \eqref{eq:gain-risk}, and $\eta_{max}=\max\{\mu_{1},\dots,\mu_{n}\}$.
More precisely, denoting the generic portfolio target return by $\eta_{\alpha}=\eta_{min}+\alpha(\eta_{max}-\eta_{min})$, we consider $\alpha=0, 1/3$.
%
%
Furthermore, for comparison purposes, we also examine the performance of the Risk Parity \citep[][]{Maillard2010} and of the Equally-Weighted \citep[][]{Demiguel2009} portfolios, and that of the market index.
\begin{table}[htbp]
	\centering
	\caption{List of portfolio strategies}
	\scalebox{0.70}{
	\begin{tabular}{lllllll}
		\toprule
		\#    & Approach & & & & & Abbreviation \\
		\midrule
		& \textit{\textbf{Minimum Risk}} \\
		1     & Min Variance & & & & & \textbf{MV0} \\
		2     & Min MAD & & & & & \textbf{MAD0} \\
		3     & Min CVaR$_{\varepsilon}$ & & & & & \textbf{CVaR0} \\
		4     & Min Expectiles & & & & & \textbf{Expe0} \\
		& \textit{\textbf{Minimum Risk with target return}} \\
		5     & Min Variance with target return & & & & & \textbf{MV1} \\
		6     & Min MAD with target return & & & & & \textbf{MAD1} \\
		7     & Min CVaR$_{\varepsilon}$ with target return & & & & & \textbf{CVaR1} \\
		8     & Min Expectiles with target return & & & & & \textbf{Expe1} \\
		& \textit{\textbf{Maximum Diversification Ratio}} \\
		9     & Max Volatility Ratio & & & & & \textbf{DRvol0} \\
		10    & Max MAD Ratio & & & & & \textbf{DRMAD0} \\
		11    & Max CVaR$_{\varepsilon}$ Ratio & & & & & \textbf{DRCVaR0} \\
		12    & Max Expectiles Ratio & & & & & \textbf{DRExpe0} \\
		& \textit{\textbf{Maximum Diversification Ratio with target return}} \\
		13    & Max Volatility Ratio with target return & & & & & \textbf{DRvol1} \\
		14    & Max MAD Ratio with target return & & & & & \textbf{DRMAD1} \\
		15    & Max CVaR$_{\varepsilon}$ Ratio with target return & & & & & \textbf{DRCVaR1} \\
		16    & Max Expectiles Ratio with target return & & & & & \textbf{DRExpe1} \\
              & & & & & & \\
        17    & Risk Parity & & & & & \textbf{RP} \\
        18    & Equally Weighted & & & & & \textbf{EW} \\
        19    & Market Index & & & & & \textbf{Index}
	\end{tabular}%
    }
	\label{tab:ListOfModels}%
\end{table}%
In Table \ref{tab:ListOfModels}, we list the nineteen portfolio strategies analyzed together with the respective abbreviations.
Note that abbreviations ending with \textbf{0} (e.g., \textbf{MV0} and \textbf{DRvol0}) indicate strategies obtained without the return constraint (or, equivalently, with target level $\eta_0$).  
Abbreviations ending with \textbf{1} (e.g., \textbf{MV1} and \textbf{DRvol1}) denote strategies obtained by solving Models \eqref{eq:biobjective_b} and \eqref{eq:gain-risk} with a common target return $\bar{\eta}_1=\max\{\eta_{1/3}^{MinV},\eta_{1/3}^{MinMAD},\eta_{1/3}^{MinCVaR},\eta_{1/3}^{MinExp},\eta_{1/3}^{DRVol},\eta_{1/3}^{DRMAD},\eta_{1/3}^{DRCVaR},\eta_{1/3}^{DRExp}\}$.
\noindent
Furthermore, in Models \eqref{eq:biobjective_b} and \eqref{eq:gain-risk}, we fix $\varepsilon=5\%$ for CVaR and $\alpha=0.9$ for Expectiles \citep[see, e.g.,][]{cesarone2017minimum,bellini2021risk}.
%

The out-of-sample performance results of each portfolio strategy are evaluated by considering several performance measures widely used in the literature \citep[see, e.g.,][]{cesarone2015linear,cesarone2016optimally,cesarone2020optimization,bruni2017exact,cesarone2017minimum}, and detailed below.
%
%
\begin{itemize}
\setlength\itemsep{1em}
\item[-] The average $\mathbf{\mu^{out}}$ and the standard deviation $\mathbf{\sigma^{out}}$ of the out-of-sample portfolio returns $R^{out}$.
  \item[-] The \textbf{Sharpe} ratio defined as $\mu^{out}/\sigma^{out}$ \citep{sharpe1966mutual,Sharpe1994}.
Clearly, higher Sharpe ratios are associated with better results.
%
%
%
\item[-] The Maximum DrawDown \citep[\textbf{MDD},][]{Chekhlov2005} is a measure of the maximum potential out-of-sample loss from the observed peak and is expressed as
\begin{equation*}
\mbox{MDD} = \min\limits_{T^{in}+1\le t\le T} DD_{t},
\end{equation*}
where $T^{in}$ and $T$ are the length of a single in-sample window and the entire time series, respectively.
%
%
%
The DrawDown is computed as
\begin{equation*}
DD_{t}=\frac{W_{t}-\max\limits_{T^{in}+1\le \tau \le t}W_{\tau}}{\max\limits_{T^{in}+1\le \tau \le t}W_{\tau}}, \qquad t \in \{T^{in}+1,\dots T\},
\end{equation*}
where $W_{t}=W_{t-1}(1+R^{out}_{t})$ denotes the portfolio wealth at time $t$, with $W_{0}=1$. 
Note that the MDD is always non-positive, hence higher values are preferable.
\item[-] The \textbf{Ulcer} index \citep{martin1989investor} measures the depth and duration of drawdowns in prices over the out‐of‐sample period, and is expressed as
\begin{equation*}
\mbox{Ulcer} = \sqrt{\frac{\sum\limits_{t=T^{in}+1}^{T}DD_{t}^2}{T-T^{in}}}\, .
\end{equation*}
%
Lower Ulcer values indicate better portfolio performance.
\item[-] The Rachev ratio \citep[\textbf{Rachev10},][]{RacBigOrtSto2004} is a measure of the relative gap between the mean of the best $\alpha\%$ values of $R^{out}-r_{f}$ and that of the worst $\beta\%$ ones, and it is defined as
\begin{equation*}
\mbox{Rachev10} = \frac{CVaR_{\alpha}(r_{f}-R^{out})}{CVaR_{\beta}(R^{out}-r_{f})},
\end{equation*}
with $r_{f} = 0$ and where $\alpha=\beta=10\%$. A higher Rachev ratio is clearly preferred.
\item[-] The Turnover \citep[\textbf{Turn},][]{Demiguel2009} evaluates the amount of trading required to perform in practice the portfolio strategy, and is defined as
\begin{equation*}
\mbox{Turn} = \frac{1}{S}\sum_{s=1}^{S}\sum_{k=1}^{n}\mid x_{s,k}-x_{s-1,k}\mid ,
\end{equation*}
where $S$ indicates the number of rebalances, $x_{s,k}$ is the portfolio weight of asset $k$ after rebalancing, and $x_{s-1,k}$ is the portfolio weight before rebalancing at time $s$. 
A lower turnover typically indicates better portfolio performance.
We emphasize that this definition of portfolio turnover is a proxy of the effective one, since it evaluates only the amount of trading generated by the models at each rebalance, without considering the trades due to changes in asset prices between one rebalance and the next. 
Hence, by definition, the turnover of the EW portfolio is zero.
\item[-] The Jensen's Alpha \citep[\textbf{AlphaJ},][]{jensen1968performance} is defined as the intercept of the line given by the linear regression of $R^{out}$ on $R_{I}^{out}$, namely
\begin{equation*}
\mbox{AlphaJ} = \mathbb{E}[R^{out}]-\beta\mathbb{E}[R_{I}^{out}]\,,
\end{equation*}
where $\displaystyle \beta=Cov(R^{out},R_{I}^{out}) / \sigma^{2}(R_{I}^{out})$ and $R_{I}^{out}$ is the out-of-sample portfolio return of the Market Index.
\item[-] The Information ratio \citep[\textbf{InfoR},][]{treynor1973use} is defined as the expected value of the difference between the out-of-sample portfolio return and that of the benchmark index divided by the standard deviation of such difference, namely
\begin{equation*}
\mbox{InfoR} = \frac{\mathbb{E}[R^{out}-R_{I}^{out}]}{\sigma[R^{out}-R_{I}^{out}]}.
\end{equation*}
Clearly, the larger its value, the better the portfolio performance.
\item[-] The Value-at-Risk of the out-of-sample portfolio returns $R^{out}$ with a confidence level equal to 5\% (\textbf{VaR5}).
%
\item[-] The Omega ratio \citep[\textbf{Omega},][]{harlow1989asset} is the ratio between the average of positive and negative out-of-sample portfolio returns, and is defined as
\begin{equation*}
\mbox{Omega} = \displaystyle \frac{\mathbb{E}[\max\{0,R^{out}-\eta\}]}{\mathbb{E}[\min\{0,R^{out}-\eta\}]}\,,
\end{equation*}
where $\eta$ is fixed to 0.
%
\item[-] The average number of selected assets (\textbf{ave$\#$}).
%
\item[-] The Return on Investment \citep[\textbf{ROI}, see, e.g.,][]{cesarone2022comparing,cesarone2022does}, namely the time-by-time return generated by each portfolio strategy over a specified time horizon $\Delta\nu$. ROI is defined as:
\begin{equation*}\label{eq:ROI}
ROI_{\tau }=\frac{W_{\tau }-W_{\tau - \Delta\nu}}{W_{\tau - \Delta\nu}} \qquad \tau = T^{in} + \Delta\nu + 1, \ldots, T
\end{equation*}
where $W_{\tau }=W_{\tau - \Delta\nu} \prod_{t=\tau - \Delta\nu +1}^{\tau} (1+R_{t}^{out})$ is the portfolio wealth generated by investing at the beginning of the time horizon
the amount of capital $W_{\tau - \Delta\nu}$.
%
In our experiment, we compute the annual ROI, namely we fix $\Delta\nu$ equal to 250 days.
\end{itemize}
%


\subsection{Out-of-Sample Performance Evaluation}\label{subsec:PerfEval}


We provide here the out-of-sample computational results of the portfolio selection strategies listed in Table \ref{tab:ListOfModels}, for each data set in Table \ref{tab:DailyDatasets}.
For reasons of space, the experimental results of the out-of-sample ROI are moved to Appendix \ref{sec:Appendix}.
The rank of performance results is shown in different colors, ranging from deep-red to deep-green, which represent the worst and the best results, respectively.



\noindent
In Table \ref{tab:CompRes_DJ} we report the computational results obtained with the DowJones data set.
First, we note that the return-diversification models generally tend to outperform the return-risk strategies in terms of $\mu^{out}$, Ulcer index, and Turnover.
Also when analyzing the Jensen's Alpha, and the Sharpe, Rachev, Information, and Omega ratios, which are generally considered highly desirable indicators by practitioners, the return-diversification approaches seem to obtain superior results.
Among these, the CVaR maximum diversification ratio with and without the return constraint (DRCVaR1 and DRCVaR0, respectively) achieves the highest Jensen's Alpha, Sharpe, Rachev, and Omega ratios.
Interestingly, the MAD maximum diversification ratio without the return constraint (i.e, DRMAD0) provides the lowest Ulcer index and relatively low Turnover.
These results indicate less severe and shorter-lasting drawdowns and lower trading costs.
As expected, the minimum-risk strategies exhibit the best results in terms of risk, as shown by the out-of-sample $\sigma^{out}$, MDD, and VaR5.
We also observe that when requiring a target constraint on the portfolio expected return, the performance results of the optimal portfolios typically improve. 
%
\begin{table}[htbp]
	\centering
	\caption{Out-of-sample performance results on the DowJones dataset}
	\resizebox{0.90\textwidth}{!}{\begin{tabular}{|l|c|c|c|c|c|c|c|c|c|c|c|c|}
		\toprule
		\multicolumn{1}{|c|}{\textbf{Approach}} & $\boldsymbol{\mu^{out}}$ & $\boldsymbol{\sigma^{out}}$ & \textbf{Sharpe} & \textbf{MDD} & \textbf{Ulcer} & \textbf{Rachev10} & \textbf{Turn} & \textbf{AlphaJ} & \textbf{InfoR} & \textbf{VaR5} & \textbf{Omega} & \textbf{ave \#} \\
		\midrule
		\textbf{MV0} & \cellcolor[rgb]{ .973,  .412,  .42}0,023\% & \cellcolor[rgb]{ .435,  .757,  .482}0,950\% & \cellcolor[rgb]{ .973,  .412,  .42}2,43\% & \cellcolor[rgb]{ .859,  .882,  .51}-0,329 & \cellcolor[rgb]{ .992,  .725,  .482}0,087 & \cellcolor[rgb]{ .976,  .541,  .443}0,948 & \cellcolor[rgb]{ .635,  .816,  .494}0,119 & \cellcolor[rgb]{ .973,  .412,  .42}-0,002\% & \cellcolor[rgb]{ .973,  .412,  .42}-0,022 & \cellcolor[rgb]{ .388,  .745,  .482}1,33\% & \cellcolor[rgb]{ .973,  .412,  .42}1,081 & 10 \\
		\midrule
		\textbf{MAD0} & \cellcolor[rgb]{ .976,  .522,  .439}0,028\% & \cellcolor[rgb]{ .494,  .773,  .486}0,968\% & \cellcolor[rgb]{ .98,  .573,  .447}2,91\% & \cellcolor[rgb]{ .957,  .91,  .518}-0,334 & \cellcolor[rgb]{ .816,  .867,  .506}0,077 & \cellcolor[rgb]{ .98,  .616,  .459}0,952 & \cellcolor[rgb]{ .894,  .89,  .51}0,224 & \cellcolor[rgb]{ .98,  .565,  .447}0,003\% & \cellcolor[rgb]{ .976,  .502,  .435}-0,015 & \cellcolor[rgb]{ .478,  .769,  .486}1,38\% & \cellcolor[rgb]{ .98,  .569,  .447}1,098 & 14 \\
		\midrule
		\textbf{CVaR0} & \cellcolor[rgb]{ .976,  .49,  .431}0,027\% & \cellcolor[rgb]{ .388,  .745,  .482}0,934\% & \cellcolor[rgb]{ .976,  .553,  .443}2,85\% & \cellcolor[rgb]{ .502,  .78,  .49}-0,310 & \cellcolor[rgb]{ 1,  .922,  .518}0,080 & \cellcolor[rgb]{ .996,  .871,  .506}0,966 & \cellcolor[rgb]{ .918,  .898,  .51}0,235 & \cellcolor[rgb]{ .98,  .604,  .455}0,004\% & \cellcolor[rgb]{ .976,  .506,  .435}-0,014 & \cellcolor[rgb]{ .525,  .784,  .49}1,40\% & \cellcolor[rgb]{ .976,  .514,  .439}1,092 & 8 \\
		\midrule
		\textbf{Expe0} & \cellcolor[rgb]{ .973,  .459,  .427}0,025\% & \cellcolor[rgb]{ .416,  .753,  .482}0,943\% & \cellcolor[rgb]{ .976,  .49,  .435}2,68\% & \cellcolor[rgb]{ .553,  .792,  .494}-0,313 & \cellcolor[rgb]{ .969,  .91,  .514}0,080 & \cellcolor[rgb]{ .984,  .659,  .467}0,955 & \cellcolor[rgb]{ .871,  .882,  .51}0,214 & \cellcolor[rgb]{ .976,  .51,  .435}0,001\% & \cellcolor[rgb]{ .973,  .463,  .427}-0,018 & \cellcolor[rgb]{ .439,  .757,  .482}1,35\% & \cellcolor[rgb]{ .973,  .475,  .431}1,088 & 9 \\
		\midrule
		\textbf{MV1} & \cellcolor[rgb]{ .984,  .627,  .459}0,033\% & \cellcolor[rgb]{ .827,  .871,  .506}1,077\% & \cellcolor[rgb]{ .98,  .62,  .459}3,07\% & \cellcolor[rgb]{ .992,  .843,  .502}-0,351 & \cellcolor[rgb]{ .973,  .412,  .42}0,097 & \cellcolor[rgb]{ .988,  .757,  .482}0,960 & \cellcolor[rgb]{ .984,  .592,  .455}0,406 & \cellcolor[rgb]{ .984,  .659,  .467}0,005\% & \cellcolor[rgb]{ .98,  .604,  .455}-0,006 & \cellcolor[rgb]{ .976,  .914,  .514}1,65\% & \cellcolor[rgb]{ .98,  .557,  .447}1,097 & 9 \\
		\midrule
		\textbf{MAD1} & \cellcolor[rgb]{ .984,  .69,  .471}0,036\% & \cellcolor[rgb]{ .875,  .886,  .51}1,093\% & \cellcolor[rgb]{ .984,  .69,  .471}3,27\% & \cellcolor[rgb]{ .992,  .847,  .502}-0,351 & \cellcolor[rgb]{ .98,  .51,  .439}0,094 & \cellcolor[rgb]{ .992,  .827,  .498}0,964 & \cellcolor[rgb]{ .98,  .545,  .447}0,426 & \cellcolor[rgb]{ .988,  .737,  .482}0,007\% & \cellcolor[rgb]{ .984,  .655,  .463}-0,002 & \cellcolor[rgb]{ .996,  .816,  .498}1,70\% & \cellcolor[rgb]{ .98,  .624,  .459}1,104 & 10 \\
		\midrule
		\textbf{CVaR1} & \cellcolor[rgb]{ .988,  .718,  .478}0,037\% & \cellcolor[rgb]{ .839,  .875,  .506}1,081\% & \cellcolor[rgb]{ .988,  .741,  .482}3,42\% & \cellcolor[rgb]{ .388,  .745,  .482}-0,304 & \cellcolor[rgb]{ .992,  .725,  .482}0,087 & \cellcolor[rgb]{ .733,  .847,  .506}0,982 & \cellcolor[rgb]{ .973,  .412,  .42}0,480 & \cellcolor[rgb]{ .996,  .859,  .502}0,011\% & \cellcolor[rgb]{ .984,  .678,  .471}0,000 & \cellcolor[rgb]{ 1,  .922,  .518}1,66\% & \cellcolor[rgb]{ .984,  .655,  .467}1,108 & 7 \\
		\midrule
		\textbf{Expe1} & \cellcolor[rgb]{ .984,  .639,  .463}0,034\% & \cellcolor[rgb]{ .749,  .847,  .502}1,052\% & \cellcolor[rgb]{ .984,  .663,  .467}3,19\% & \cellcolor[rgb]{ .953,  .91,  .518}-0,334 & \cellcolor[rgb]{ .976,  .443,  .427}0,096 & \cellcolor[rgb]{ .992,  .816,  .494}0,963 & \cellcolor[rgb]{ .98,  .553,  .447}0,422 & \cellcolor[rgb]{ .988,  .722,  .478}0,007\% & \cellcolor[rgb]{ .98,  .616,  .459}-0,005 & \cellcolor[rgb]{ .867,  .882,  .51}1,59\% & \cellcolor[rgb]{ .98,  .588,  .451}1,101 & 13 \\
		\midrule
		\textbf{EW} & \cellcolor[rgb]{ .643,  .82,  .498}0,051\% & \cellcolor[rgb]{ .973,  .412,  .42}1,239\% & \cellcolor[rgb]{ .831,  .875,  .51}4,09\% & \cellcolor[rgb]{ .98,  .62,  .459}-0,392 & \cellcolor[rgb]{ .773,  .855,  .502}0,076 & \cellcolor[rgb]{ 1,  .922,  .518}0,969 & -     & \cellcolor[rgb]{ .929,  .902,  .514}0,014\% & \cellcolor[rgb]{ .388,  .745,  .482}0,066 & \cellcolor[rgb]{ .98,  .522,  .443}1,79\% & \cellcolor[rgb]{ .639,  .82,  .498}1,143 & 28 \\
		\midrule
		\textbf{RP} & \cellcolor[rgb]{ 1,  .922,  .518}0,046\% & \cellcolor[rgb]{ .996,  .804,  .498}1,158\% & \cellcolor[rgb]{ .961,  .91,  .518}3,99\% & \cellcolor[rgb]{ .988,  .725,  .478}-0,373 & \cellcolor[rgb]{ .518,  .78,  .486}0,072 & \cellcolor[rgb]{ .992,  .82,  .498}0,964 & \cellcolor[rgb]{ .388,  .745,  .482}0,017 & \cellcolor[rgb]{ .996,  .89,  .51}0,012\% & \cellcolor[rgb]{ .706,  .839,  .502}0,043 & \cellcolor[rgb]{ 1,  .882,  .514}1,68\% & \cellcolor[rgb]{ .855,  .882,  .51}1,140 & 28 \\
		\midrule
		\textbf{Index} & \cellcolor[rgb]{ .988,  .714,  .475}0,037\% & \cellcolor[rgb]{ .98,  .49,  .435}1,223\% & \cellcolor[rgb]{ .98,  .608,  .455}3,02\% & \cellcolor[rgb]{ .973,  .412,  .42}-0,431 & \cellcolor[rgb]{ .98,  .494,  .435}0,095 & \cellcolor[rgb]{ .973,  .412,  .42}0,941 & -     & -     & -     & \cellcolor[rgb]{ .98,  .518,  .443}1,79\% & \cellcolor[rgb]{ .98,  .612,  .455}1,103 & - \\
		\midrule
		\textbf{DRvol0} & \cellcolor[rgb]{ .953,  .91,  .518}0,047\% & \cellcolor[rgb]{ .988,  .69,  .475}1,182\% & \cellcolor[rgb]{ 1,  .922,  .518}3,96\% & \cellcolor[rgb]{ .988,  .757,  .486}-0,367 & \cellcolor[rgb]{ 1,  .922,  .518}0,080 & \cellcolor[rgb]{ .784,  .863,  .506}0,979 & \cellcolor[rgb]{ .678,  .827,  .498}0,136 & \cellcolor[rgb]{ .945,  .906,  .518}0,014\% & \cellcolor[rgb]{ .976,  .918,  .518}0,022 & \cellcolor[rgb]{ 1,  .875,  .51}1,68\% & \cellcolor[rgb]{ 1,  .922,  .518}1,138 & 15 \\
		\midrule
		\textbf{DRMAD0} & \cellcolor[rgb]{ .867,  .882,  .51}0,048\% & \cellcolor[rgb]{ 1,  .859,  .506}1,147\% & \cellcolor[rgb]{ .71,  .839,  .502}4,18\% & \cellcolor[rgb]{ 1,  .922,  .518}-0,337 & \cellcolor[rgb]{ .388,  .745,  .482}0,070 & \cellcolor[rgb]{ .886,  .89,  .514}0,974 & \cellcolor[rgb]{ .776,  .855,  .502}0,177 & \cellcolor[rgb]{ .831,  .875,  .51}0,015\% & \cellcolor[rgb]{ .918,  .898,  .514}0,027 & \cellcolor[rgb]{ .89,  .89,  .51}1,60\% & \cellcolor[rgb]{ .427,  .757,  .486}1,147 & 17 \\
		\midrule
		\textbf{DRCVaR0} & \cellcolor[rgb]{ .698,  .835,  .502}0,050\% & \cellcolor[rgb]{ 1,  .922,  .518}1,133\% & \cellcolor[rgb]{ .388,  .745,  .482}4,41\% & \cellcolor[rgb]{ .988,  .741,  .482}-0,370 & \cellcolor[rgb]{ .706,  .835,  .498}0,075 & \cellcolor[rgb]{ .463,  .769,  .49}0,995 & \cellcolor[rgb]{ .992,  .769,  .49}0,332 & \cellcolor[rgb]{ .569,  .8,  .494}0,019\% & \cellcolor[rgb]{ .949,  .91,  .518}0,024 & \cellcolor[rgb]{ .898,  .89,  .51}1,61\% & \cellcolor[rgb]{ .388,  .745,  .482}1,147 & 11 \\
		\midrule
		\textbf{DRExpe0} & \cellcolor[rgb]{ .98,  .918,  .518}0,046\% & \cellcolor[rgb]{ .941,  .902,  .514}1,115\% & \cellcolor[rgb]{ .718,  .843,  .502}4,17\% & \cellcolor[rgb]{ .796,  .863,  .506}-0,326 & \cellcolor[rgb]{ .725,  .839,  .498}0,076 & \cellcolor[rgb]{ .725,  .843,  .502}0,982 & \cellcolor[rgb]{ .969,  .91,  .514}0,256 & \cellcolor[rgb]{ .796,  .863,  .506}0,016\% & \cellcolor[rgb]{ .996,  .902,  .514}0,019 & \cellcolor[rgb]{ .847,  .875,  .506}1,58\% & \cellcolor[rgb]{ .682,  .831,  .502}1,143 & 8 \\
		\midrule
		\textbf{DRvol1} & \cellcolor[rgb]{ .49,  .776,  .49}0,053\% & \cellcolor[rgb]{ .976,  .447,  .427}1,232\% & \cellcolor[rgb]{ .592,  .804,  .494}4,26\% & \cellcolor[rgb]{ .992,  .788,  .49}-0,362 & \cellcolor[rgb]{ .996,  .835,  .502}0,083 & \cellcolor[rgb]{ .675,  .827,  .502}0,985 & \cellcolor[rgb]{ 1,  .922,  .518}0,268 & \cellcolor[rgb]{ .569,  .8,  .494}0,019\% & \cellcolor[rgb]{ .906,  .894,  .514}0,028 & \cellcolor[rgb]{ .98,  .514,  .439}1,79\% & \cellcolor[rgb]{ .808,  .867,  .51}1,141 & 11 \\
		\midrule
		\textbf{DRMAD1} & \cellcolor[rgb]{ .388,  .745,  .482}0,054\% & \cellcolor[rgb]{ .98,  .502,  .439}1,221\% & \cellcolor[rgb]{ .4,  .749,  .486}4,40\% & \cellcolor[rgb]{ .741,  .847,  .506}-0,323 & \cellcolor[rgb]{ .647,  .82,  .494}0,074 & \cellcolor[rgb]{ .714,  .839,  .502}0,983 & \cellcolor[rgb]{ .996,  .827,  .502}0,307 & \cellcolor[rgb]{ .471,  .769,  .49}0,021\% & \cellcolor[rgb]{ .871,  .886,  .514}0,030 & \cellcolor[rgb]{ .973,  .412,  .42}1,82\% & \cellcolor[rgb]{ .514,  .784,  .49}1,145 & 13 \\
		\midrule
		\textbf{DRCVaR1} & \cellcolor[rgb]{ .408,  .753,  .486}0,054\% & \cellcolor[rgb]{ .98,  .529,  .443}1,215\% & \cellcolor[rgb]{ .392,  .749,  .486}4,41\% & \cellcolor[rgb]{ .984,  .678,  .471}-0,382 & \cellcolor[rgb]{ .992,  .761,  .486}0,086 & \cellcolor[rgb]{ .388,  .745,  .482}0,998 & \cellcolor[rgb]{ .984,  .612,  .459}0,397 & \cellcolor[rgb]{ .388,  .745,  .482}0,022\% & \cellcolor[rgb]{ .929,  .902,  .514}0,026 & \cellcolor[rgb]{ .976,  .427,  .424}1,82\% & \cellcolor[rgb]{ .702,  .835,  .502}1,142 & 9 \\
		\midrule
		\textbf{DRExpe1} & \cellcolor[rgb]{ .749,  .851,  .506}0,049\% & \cellcolor[rgb]{ .996,  .808,  .498}1,157\% & \cellcolor[rgb]{ .588,  .804,  .494}4,26\% & \cellcolor[rgb]{ .898,  .894,  .514}-0,331 & \cellcolor[rgb]{ .878,  .886,  .51}0,078 & \cellcolor[rgb]{ .706,  .839,  .502}0,983 & \cellcolor[rgb]{ .992,  .761,  .49}0,335 & \cellcolor[rgb]{ .624,  .816,  .498}0,018\% & \cellcolor[rgb]{ .984,  .918,  .518}0,022 & \cellcolor[rgb]{ .988,  .675,  .471}1,74\% & \cellcolor[rgb]{ .871,  .886,  .514}1,140 & 11 \\
		\bottomrule
	\end{tabular}}%
	\label{tab:CompRes_DJ}%
\end{table}%
%


%
%



In Table \ref{tab:CompRes_EUXX50} we report the computational results obtained with the EuroStoxx 50 data set.
We observe again that the return-diversification models typically outperform the other considered strategies in almost all performance measures.
In particular, DRVol0, DRMAD0, DRCVaR0, and DRExpe0 achieve significant performances in terms of $\mu^{out}$, $\sigma^{out}$,  Sharpe, MDD, Ulcer, and Omega, thus representing valid solutions to both risk-averse and risk-lover investors.
%
%
Furthermore, these strategies provide attractive performance relative to the benchmark index, as indicated by the Jensen's Alpha and the Information ratio.
\begin{table}[htbp]
	\centering
	\caption{Out-of-sample performance results on the EuroStoxx50 dataset}
	\resizebox{0.90\textwidth}{!}{\begin{tabular}{|l|c|c|c|c|c|c|c|c|c|c|c|c|}
		\toprule
		\multicolumn{1}{|c|}{\textbf{Approach}} & $\boldsymbol{\mu^{out}}$ & $\boldsymbol{\sigma^{out}}$ & \textbf{Sharpe} & \textbf{MDD} & \textbf{Ulcer} & \textbf{Rachev10} & \textbf{Turn} & \textbf{AlphaJ} & \textbf{InfoR} & \textbf{VaR5} & \textbf{Omega} & \textbf{ave \#} \\
		\midrule
		\textbf{MV0} & \cellcolor[rgb]{ .996,  .906,  .514}0,046\% & \cellcolor[rgb]{ .392,  .745,  .482}0,982\% & \cellcolor[rgb]{ .737,  .847,  .506}4,69\% & \cellcolor[rgb]{ .498,  .78,  .49}-0,311 & \cellcolor[rgb]{ .482,  .773,  .486}0,073 & \cellcolor[rgb]{ .969,  .914,  .518}0,966 & \cellcolor[rgb]{ .627,  .812,  .494}0,118 & \cellcolor[rgb]{ 1,  .922,  .518}0,036\% & \cellcolor[rgb]{ .996,  .882,  .51}0,033 & \cellcolor[rgb]{ .392,  .745,  .482}1,47\% & \cellcolor[rgb]{ .69,  .835,  .502}1,147 & 11 \\
		\midrule
		\textbf{MAD0} & \cellcolor[rgb]{ .992,  .831,  .498}0,042\% & \cellcolor[rgb]{ .541,  .788,  .49}1,024\% & \cellcolor[rgb]{ .996,  .914,  .514}4,10\% & \cellcolor[rgb]{ .992,  .922,  .518}-0,369 & \cellcolor[rgb]{ .784,  .859,  .502}0,088 & \cellcolor[rgb]{ 1,  .922,  .518}0,963 & \cellcolor[rgb]{ 1,  .859,  .506}0,294 & \cellcolor[rgb]{ .988,  .749,  .482}0,031\% & \cellcolor[rgb]{ .992,  .776,  .49}0,030 & \cellcolor[rgb]{ .459,  .765,  .486}1,51\% & \cellcolor[rgb]{ .957,  .91,  .518}1,128 & 16 \\
		\midrule
		\textbf{CVaR0} & \cellcolor[rgb]{ .992,  .922,  .518}0,047\% & \cellcolor[rgb]{ .482,  .769,  .486}1,007\% & \cellcolor[rgb]{ .737,  .847,  .506}4,69\% & \cellcolor[rgb]{ .988,  .765,  .486}-0,394 & \cellcolor[rgb]{ .761,  .851,  .502}0,087 & \cellcolor[rgb]{ .988,  .737,  .478}0,950 & \cellcolor[rgb]{ .961,  .91,  .514}0,257 & \cellcolor[rgb]{ .949,  .91,  .518}0,037\% & \cellcolor[rgb]{ .996,  .886,  .51}0,033 & \cellcolor[rgb]{ .522,  .78,  .486}1,54\% & \cellcolor[rgb]{ .714,  .839,  .502}1,145 & 9 \\
		\midrule
		\textbf{Expe0} & \cellcolor[rgb]{ .992,  .835,  .498}0,042\% & \cellcolor[rgb]{ .388,  .745,  .482}0,980\% & \cellcolor[rgb]{ .922,  .902,  .514}4,30\% & \cellcolor[rgb]{ .388,  .745,  .482}-0,298 & \cellcolor[rgb]{ .541,  .788,  .49}0,076 & \cellcolor[rgb]{ .957,  .91,  .518}0,967 & \cellcolor[rgb]{ .839,  .875,  .506}0,206 & \cellcolor[rgb]{ .992,  .784,  .49}0,032\% & \cellcolor[rgb]{ .988,  .706,  .475}0,027 & \cellcolor[rgb]{ .388,  .745,  .482}1,46\% & \cellcolor[rgb]{ .894,  .89,  .514}1,133 & 10 \\
		\midrule
		\textbf{MV1} & \cellcolor[rgb]{ .992,  .78,  .49}0,039\% & \cellcolor[rgb]{ .937,  .902,  .514}1,134\% & \cellcolor[rgb]{ .992,  .796,  .49}3,43\% & \cellcolor[rgb]{ .988,  .745,  .482}-0,397 & \cellcolor[rgb]{ .992,  .718,  .478}0,127 & \cellcolor[rgb]{ .973,  .443,  .424}0,930 & \cellcolor[rgb]{ .988,  .639,  .463}0,367 & \cellcolor[rgb]{ .98,  .624,  .459}0,028\% & \cellcolor[rgb]{ .98,  .58,  .451}0,023 & \cellcolor[rgb]{ 1,  .894,  .514}1,84\% & \cellcolor[rgb]{ .992,  .78,  .49}1,101 & 8 \\
		\midrule
		\textbf{MAD1} & \cellcolor[rgb]{ .984,  .682,  .471}0,033\% & \cellcolor[rgb]{ 1,  .922,  .518}1,152\% & \cellcolor[rgb]{ .988,  .702,  .475}2,91\% & \cellcolor[rgb]{ .976,  .494,  .435}-0,436 & \cellcolor[rgb]{ .98,  .561,  .451}0,148 & \cellcolor[rgb]{ .973,  .412,  .42}0,927 & \cellcolor[rgb]{ .98,  .537,  .447}0,400 & \cellcolor[rgb]{ .973,  .412,  .42}0,022\% & \cellcolor[rgb]{ .973,  .412,  .42}0,018 & \cellcolor[rgb]{ 1,  .906,  .518}1,84\% & \cellcolor[rgb]{ .984,  .69,  .471}1,086 & 9 \\
		\midrule
		\textbf{CVaR1} & \cellcolor[rgb]{ .98,  .918,  .518}0,048\% & \cellcolor[rgb]{ 1,  .922,  .518}1,152\% & \cellcolor[rgb]{ 1,  .922,  .518}4,13\% & \cellcolor[rgb]{ .992,  .824,  .498}-0,385 & \cellcolor[rgb]{ .992,  .757,  .486}0,122 & \cellcolor[rgb]{ .98,  .627,  .459}0,943 & \cellcolor[rgb]{ .973,  .412,  .42}0,440 & \cellcolor[rgb]{ .976,  .918,  .518}0,037\% & \cellcolor[rgb]{ .992,  .831,  .498}0,031 & \cellcolor[rgb]{ 1,  .922,  .518}1,82\% & \cellcolor[rgb]{ .996,  .906,  .514}1,123 & 7 \\
		\midrule
		\textbf{Expe1} & \cellcolor[rgb]{ .984,  .698,  .475}0,034\% & \cellcolor[rgb]{ .906,  .894,  .51}1,126\% & \cellcolor[rgb]{ .988,  .729,  .478}3,06\% & \cellcolor[rgb]{ .98,  .561,  .447}-0,425 & \cellcolor[rgb]{ .98,  .557,  .451}0,148 & \cellcolor[rgb]{ .98,  .569,  .447}0,938 & \cellcolor[rgb]{ .98,  .529,  .443}0,402 & \cellcolor[rgb]{ .973,  .467,  .427}0,024\% & \cellcolor[rgb]{ .973,  .416,  .42}0,018 & \cellcolor[rgb]{ .922,  .898,  .51}1,78\% & \cellcolor[rgb]{ .988,  .718,  .478}1,091 & 13 \\
		\midrule
		\textbf{EW} & \cellcolor[rgb]{ 1,  .922,  .518}0,047\% & \cellcolor[rgb]{ .976,  .451,  .427}1,405\% & \cellcolor[rgb]{ .992,  .78,  .49}3,34\% & \cellcolor[rgb]{ .976,  .525,  .439}-0,431 & \cellcolor[rgb]{ 1,  .878,  .51}0,105 & \cellcolor[rgb]{ .737,  .847,  .506}0,981 & -     & \cellcolor[rgb]{ .984,  .686,  .471}0,030\% & \cellcolor[rgb]{ .388,  .745,  .482}0,104 & \cellcolor[rgb]{ .98,  .498,  .439}2,13\% & \cellcolor[rgb]{ .992,  .808,  .494}1,106 & 46 \\
		\midrule
		\textbf{RP} & \cellcolor[rgb]{ .996,  .898,  .514}0,046\% & \cellcolor[rgb]{ .988,  .667,  .471}1,289\% & \cellcolor[rgb]{ .992,  .82,  .498}3,55\% & \cellcolor[rgb]{ .988,  .753,  .482}-0,396 & \cellcolor[rgb]{ .906,  .894,  .51}0,094 & \cellcolor[rgb]{ .882,  .89,  .514}0,971 & \cellcolor[rgb]{ .388,  .745,  .482}0,017 & \cellcolor[rgb]{ .984,  .698,  .475}0,030\% & \cellcolor[rgb]{ .541,  .792,  .494}0,086 & \cellcolor[rgb]{ .992,  .769,  .49}1,93\% & \cellcolor[rgb]{ .996,  .851,  .502}1,113 & 46 \\
		\midrule
		\textbf{Index} & \cellcolor[rgb]{ .973,  .412,  .42}0,018\% & \cellcolor[rgb]{ .973,  .412,  .42}1,425\% & \cellcolor[rgb]{ .973,  .412,  .42}1,25\% & \cellcolor[rgb]{ .973,  .412,  .42}-0,449 & \cellcolor[rgb]{ .973,  .412,  .42}0,167 & \cellcolor[rgb]{ .992,  .816,  .494}0,956 & -     & -     & -     & \cellcolor[rgb]{ .973,  .412,  .42}2,19\% & \cellcolor[rgb]{ .973,  .412,  .42}1,038 & - \\
		\midrule
		\textbf{DRvol0} & \cellcolor[rgb]{ .494,  .776,  .49}0,060\% & \cellcolor[rgb]{ .796,  .863,  .506}1,095\% & \cellcolor[rgb]{ .388,  .745,  .482}5,44\% & \cellcolor[rgb]{ .667,  .827,  .502}-0,331 & \cellcolor[rgb]{ .42,  .753,  .482}0,070 & \cellcolor[rgb]{ .627,  .816,  .498}0,989 & \cellcolor[rgb]{ .631,  .812,  .494}0,119 & \cellcolor[rgb]{ .49,  .776,  .49}0,048\% & \cellcolor[rgb]{ .835,  .875,  .51}0,053 & \cellcolor[rgb]{ .82,  .867,  .506}1,72\% & \cellcolor[rgb]{ .388,  .745,  .482}1,168 & 15 \\
		\midrule
		\textbf{DRMAD0} & \cellcolor[rgb]{ .639,  .82,  .498}0,056\% & \cellcolor[rgb]{ .812,  .867,  .506}1,099\% & \cellcolor[rgb]{ .549,  .792,  .494}5,10\% & \cellcolor[rgb]{ .871,  .886,  .514}-0,355 & \cellcolor[rgb]{ .561,  .792,  .49}0,077 & \cellcolor[rgb]{ .816,  .871,  .51}0,976 & \cellcolor[rgb]{ .808,  .863,  .506}0,193 & \cellcolor[rgb]{ .663,  .827,  .502}0,044\% & \cellcolor[rgb]{ .835,  .875,  .51}0,053 & \cellcolor[rgb]{ .765,  .851,  .502}1,68\% & \cellcolor[rgb]{ .514,  .784,  .49}1,160 & 18 \\
		\midrule
		\textbf{DRCVaR0} & \cellcolor[rgb]{ .388,  .745,  .482}0,062\% & \cellcolor[rgb]{ 1,  .898,  .514}1,164\% & \cellcolor[rgb]{ .435,  .761,  .486}5,34\% & \cellcolor[rgb]{ .694,  .835,  .502}-0,334 & \cellcolor[rgb]{ .427,  .757,  .482}0,070 & \cellcolor[rgb]{ .388,  .745,  .482}1,004 & \cellcolor[rgb]{ .988,  .918,  .514}0,268 & \cellcolor[rgb]{ .388,  .745,  .482}0,050\% & \cellcolor[rgb]{ .843,  .878,  .51}0,052 & \cellcolor[rgb]{ .937,  .902,  .514}1,79\% & \cellcolor[rgb]{ .443,  .761,  .486}1,164 & 10 \\
		\midrule
		\textbf{DRExpe0} & \cellcolor[rgb]{ .576,  .8,  .494}0,058\% & \cellcolor[rgb]{ .733,  .843,  .502}1,077\% & \cellcolor[rgb]{ .435,  .761,  .486}5,35\% & \cellcolor[rgb]{ .475,  .773,  .49}-0,308 & \cellcolor[rgb]{ .388,  .745,  .482}0,068 & \cellcolor[rgb]{ .471,  .769,  .49}0,999 & \cellcolor[rgb]{ .835,  .875,  .506}0,205 & \cellcolor[rgb]{ .561,  .796,  .494}0,046\% & \cellcolor[rgb]{ .871,  .886,  .514}0,049 & \cellcolor[rgb]{ .698,  .831,  .498}1,65\% & \cellcolor[rgb]{ .443,  .761,  .486}1,165 & 8 \\
		\midrule
		\textbf{DRvol1} & \cellcolor[rgb]{ .816,  .871,  .51}0,052\% & \cellcolor[rgb]{ .996,  .78,  .494}1,228\% & \cellcolor[rgb]{ .965,  .914,  .518}4,20\% & \cellcolor[rgb]{ .918,  .898,  .514}-0,360 & \cellcolor[rgb]{ 1,  .898,  .514}0,102 & \cellcolor[rgb]{ .988,  .753,  .482}0,952 & \cellcolor[rgb]{ 1,  .922,  .518}0,273 & \cellcolor[rgb]{ .867,  .886,  .514}0,039\% & \cellcolor[rgb]{ .953,  .91,  .518}0,040 & \cellcolor[rgb]{ .996,  .784,  .494}1,92\% & \cellcolor[rgb]{ 1,  .922,  .518}1,125 & 11 \\
		\midrule
		\textbf{DRMAD1} & \cellcolor[rgb]{ .996,  .902,  .514}0,046\% & \cellcolor[rgb]{ .992,  .757,  .486}1,240\% & \cellcolor[rgb]{ .992,  .843,  .502}3,70\% & \cellcolor[rgb]{ .992,  .792,  .49}-0,389 & \cellcolor[rgb]{ .996,  .804,  .498}0,115 & \cellcolor[rgb]{ .984,  .667,  .467}0,945 & \cellcolor[rgb]{ .992,  .714,  .478}0,342 & \cellcolor[rgb]{ .992,  .812,  .494}0,033\% & \cellcolor[rgb]{ .996,  .918,  .514}0,034 & \cellcolor[rgb]{ .992,  .737,  .482}1,96\% & \cellcolor[rgb]{ .992,  .827,  .498}1,110 & 13 \\
		\midrule
		\textbf{DRCVaR1} & \cellcolor[rgb]{ .494,  .776,  .49}0,060\% & \cellcolor[rgb]{ .992,  .718,  .478}1,262\% & \cellcolor[rgb]{ .725,  .843,  .502}4,72\% & \cellcolor[rgb]{ .651,  .824,  .498}-0,329 & \cellcolor[rgb]{ 1,  .922,  .518}0,099 & \cellcolor[rgb]{ .871,  .886,  .514}0,972 & \cellcolor[rgb]{ .988,  .686,  .475}0,351 & \cellcolor[rgb]{ .525,  .788,  .494}0,047\% & \cellcolor[rgb]{ .898,  .894,  .514}0,046 & \cellcolor[rgb]{ .992,  .749,  .486}1,95\% & \cellcolor[rgb]{ .78,  .859,  .506}1,141 & 8 \\
		\midrule
		\textbf{DRExpe1} & \cellcolor[rgb]{ .965,  .914,  .518}0,048\% & \cellcolor[rgb]{ .996,  .835,  .502}1,199\% & \cellcolor[rgb]{ .996,  .898,  .51}4,00\% & \cellcolor[rgb]{ 1,  .922,  .518}-0,370 & \cellcolor[rgb]{ 1,  .875,  .51}0,106 & \cellcolor[rgb]{ .988,  .745,  .482}0,951 & \cellcolor[rgb]{ .992,  .714,  .478}0,342 & \cellcolor[rgb]{ .996,  .918,  .514}0,036\% & \cellcolor[rgb]{ 1,  .922,  .518}0,034 & \cellcolor[rgb]{ 1,  .871,  .51}1,86\% & \cellcolor[rgb]{ .996,  .882,  .51}1,119 & 10 \\
		\bottomrule
	\end{tabular}}%
	\label{tab:CompRes_EUXX50}%
\end{table}%
%



Table \ref{tab:CompRes_FTSE100} provides the out-of-sample performance results on FTSE100.
As expected, the minimum risk strategies (i.e., MV0, MAD0, CVaR0, and Expe0) seem to be attractive to risk-averse investors, as they show lowest values of $\sigma^{out}$ and VaR5.
%
When focusing on gain-risk performances, the maximum diversification ratio portfolios with and without the return constraint attain the best results in terms of $\mu^{out}$, Jensen's Alpha, Sharpe, Information, and Omega ratios.
For the Turnover, RP, MV0, and DRvol0 show low values, which means lower transaction costs.
%
%
%
%
\noindent
Furthermore, we also observe that imposing the return constraint generally improves the performance of the optimal portfolios.
Indeed, DRvol1 and DRMAD1 show the highest $\mu^{out}$, Sharpe ratio, and Jensen's Alpha.
%
\begin{table}[htbp]
	\centering
	\caption{Out-of-sample performance results on the FTSE100 dataset}
	\resizebox{0.90\textwidth}{!}{\begin{tabular}{|l|c|c|c|c|c|c|c|c|c|c|c|c|}
		\toprule
		\multicolumn{1}{|c|}{\textbf{Approach}} & $\boldsymbol{\mu^{out}}$ & $\boldsymbol{\sigma^{out}}$ & \textbf{Sharpe} & \textbf{MDD} & \textbf{Ulcer} & \textbf{Rachev10} & \textbf{Turn} & \textbf{AlphaJ} & \textbf{InfoR} & \textbf{VaR5} & \textbf{Omega} & \textbf{ave \#} \\
		\midrule
		\textbf{MV0} & \cellcolor[rgb]{ .996,  .867,  .506}0,044\% & \cellcolor[rgb]{ .388,  .745,  .482}0,871\% & \cellcolor[rgb]{ 1,  .922,  .518}5,09\% & \cellcolor[rgb]{ .388,  .745,  .482}-0,305 & \cellcolor[rgb]{ .388,  .745,  .482}0,067 & \cellcolor[rgb]{ .937,  .906,  .518}0,998 & \cellcolor[rgb]{ .71,  .835,  .498}0,197 & \cellcolor[rgb]{ .992,  .8,  .494}0,034\% & \cellcolor[rgb]{ .992,  .78,  .49}0,039 & \cellcolor[rgb]{ .388,  .745,  .482}1,27\% & \cellcolor[rgb]{ .824,  .871,  .51}1,162 & 22 \\
		\midrule
		\textbf{MAD0} & \cellcolor[rgb]{ .988,  .733,  .478}0,037\% & \cellcolor[rgb]{ .463,  .765,  .486}0,886\% & \cellcolor[rgb]{ .992,  .78,  .49}4,12\% & \cellcolor[rgb]{ .788,  .863,  .506}-0,327 & \cellcolor[rgb]{ .604,  .808,  .494}0,076 & \cellcolor[rgb]{ .984,  .631,  .459}0,966 & \cellcolor[rgb]{ .824,  .871,  .506}0,259 & \cellcolor[rgb]{ .973,  .412,  .42}0,025\% & \cellcolor[rgb]{ .973,  .412,  .42}0,029 & \cellcolor[rgb]{ .412,  .749,  .482}1,28\% & \cellcolor[rgb]{ .992,  .804,  .494}1,131 & 25 \\
		\midrule
		\textbf{CVaR0} & \cellcolor[rgb]{ .992,  .824,  .498}0,042\% & \cellcolor[rgb]{ .49,  .773,  .486}0,892\% & \cellcolor[rgb]{ .996,  .863,  .506}4,70\% & \cellcolor[rgb]{ .463,  .769,  .49}-0,309 & \cellcolor[rgb]{ .988,  .643,  .467}0,101 & \cellcolor[rgb]{ .718,  .843,  .502}1,006 & \cellcolor[rgb]{ 1,  .89,  .514}0,364 & \cellcolor[rgb]{ .984,  .694,  .471}0,031\% & \cellcolor[rgb]{ .98,  .565,  .447}0,033 & \cellcolor[rgb]{ .412,  .749,  .482}1,28\% & \cellcolor[rgb]{ .996,  .878,  .506}1,146 & 16 \\
		\midrule
		\textbf{Expe0} & \cellcolor[rgb]{ .996,  .898,  .51}0,046\% & \cellcolor[rgb]{ .51,  .78,  .486}0,896\% & \cellcolor[rgb]{ .929,  .902,  .514}5,17\% & \cellcolor[rgb]{ 1,  .922,  .518}-0,338 & \cellcolor[rgb]{ .937,  .902,  .514}0,089 & \cellcolor[rgb]{ .847,  .878,  .51}1,001 & \cellcolor[rgb]{ .969,  .91,  .514}0,336 & \cellcolor[rgb]{ .996,  .91,  .514}0,036\% & \cellcolor[rgb]{ .988,  .749,  .482}0,038 & \cellcolor[rgb]{ .42,  .753,  .482}1,28\% & \cellcolor[rgb]{ .749,  .851,  .506}1,165 & 19 \\
		\midrule
		\textbf{MV1} & \cellcolor[rgb]{ .773,  .859,  .506}0,052\% & \cellcolor[rgb]{ .98,  .914,  .514}0,992\% & \cellcolor[rgb]{ .851,  .878,  .51}5,26\% & \cellcolor[rgb]{ .557,  .796,  .494}-0,314 & \cellcolor[rgb]{ 1,  .867,  .51}0,093 & \cellcolor[rgb]{ .992,  .804,  .494}0,984 & \cellcolor[rgb]{ .996,  .78,  .494}0,407 & \cellcolor[rgb]{ .745,  .851,  .506}0,041\% & \cellcolor[rgb]{ .996,  .914,  .514}0,043 & \cellcolor[rgb]{ 1,  .918,  .518}1,47\% & \cellcolor[rgb]{ .835,  .875,  .51}1,161 & 14 \\
		\midrule
		\textbf{MAD1} & \cellcolor[rgb]{ 1,  .922,  .518}0,048\% & \cellcolor[rgb]{ .984,  .914,  .514}0,992\% & \cellcolor[rgb]{ .996,  .878,  .506}4,79\% & \cellcolor[rgb]{ .635,  .816,  .498}-0,318 & \cellcolor[rgb]{ .996,  .788,  .494}0,096 & \cellcolor[rgb]{ .988,  .725,  .478}0,976 & \cellcolor[rgb]{ .988,  .686,  .475}0,444 & \cellcolor[rgb]{ .992,  .922,  .518}0,036\% & \cellcolor[rgb]{ .988,  .733,  .478}0,038 & \cellcolor[rgb]{ .969,  .91,  .514}1,46\% & \cellcolor[rgb]{ .996,  .886,  .51}1,148 & 15 \\
		\midrule
		\textbf{CVaR1} & \cellcolor[rgb]{ .996,  .847,  .502}0,043\% & \cellcolor[rgb]{ 1,  .922,  .518}0,995\% & \cellcolor[rgb]{ .992,  .816,  .494}4,36\% & \cellcolor[rgb]{ .976,  .529,  .439}-0,364 & \cellcolor[rgb]{ .973,  .412,  .42}0,110 & \cellcolor[rgb]{ 1,  .922,  .518}0,996 & \cellcolor[rgb]{ .973,  .412,  .42}0,551 & \cellcolor[rgb]{ .988,  .761,  .486}0,033\% & \cellcolor[rgb]{ .973,  .459,  .427}0,030 & \cellcolor[rgb]{ 1,  .914,  .518}1,47\% & \cellcolor[rgb]{ .992,  .804,  .494}1,131 & 12 \\
		\midrule
		\textbf{Expe1} & \cellcolor[rgb]{ .996,  .906,  .514}0,047\% & \cellcolor[rgb]{ .996,  .843,  .506}1,025\% & \cellcolor[rgb]{ .992,  .843,  .502}4,55\% & \cellcolor[rgb]{ .749,  .851,  .506}-0,325 & \cellcolor[rgb]{ .976,  .455,  .431}0,108 & \cellcolor[rgb]{ .996,  .906,  .514}0,995 & \cellcolor[rgb]{ .988,  .639,  .467}0,463 & \cellcolor[rgb]{ .996,  .898,  .514}0,036\% & \cellcolor[rgb]{ .98,  .592,  .451}0,034 & \cellcolor[rgb]{ 1,  .922,  .518}1,47\% & \cellcolor[rgb]{ .996,  .851,  .502}1,141 & 22 \\
		\midrule
		\textbf{EW} & \cellcolor[rgb]{ .992,  .922,  .518}0,048\% & \cellcolor[rgb]{ .973,  .412,  .42}1,181\% & \cellcolor[rgb]{ .988,  .769,  .486}4,04\% & \cellcolor[rgb]{ .973,  .443,  .424}-0,369 & \cellcolor[rgb]{ .816,  .867,  .506}0,084 & \cellcolor[rgb]{ .996,  .851,  .502}0,989 & -     & \cellcolor[rgb]{ .984,  .667,  .467}0,031\% & \cellcolor[rgb]{ .388,  .745,  .482}0,078 & \cellcolor[rgb]{ .973,  .412,  .42}1,76\% & \cellcolor[rgb]{ .992,  .796,  .49}1,129 & 82 \\
		\midrule
		\textbf{RP} & \cellcolor[rgb]{ .996,  .871,  .506}0,045\% & \cellcolor[rgb]{ .992,  .757,  .486}1,056\% & \cellcolor[rgb]{ .992,  .796,  .49}4,23\% & \cellcolor[rgb]{ .988,  .725,  .478}-0,351 & \cellcolor[rgb]{ .604,  .804,  .494}0,076 & \cellcolor[rgb]{ .988,  .761,  .486}0,979 & \cellcolor[rgb]{ .388,  .745,  .482}0,023 & \cellcolor[rgb]{ .98,  .604,  .455}0,029\% & \cellcolor[rgb]{ .482,  .773,  .49}0,072 & \cellcolor[rgb]{ .992,  .753,  .486}1,57\% & \cellcolor[rgb]{ .992,  .824,  .498}1,135 & 82 \\
		\midrule
		\textbf{Index} & \cellcolor[rgb]{ .973,  .412,  .42}0,018\% & \cellcolor[rgb]{ .976,  .471,  .431}1,160\% & \cellcolor[rgb]{ .973,  .412,  .42}1,52\% & \cellcolor[rgb]{ .976,  .494,  .435}-0,366 & \cellcolor[rgb]{ .98,  .518,  .443}0,106 & \cellcolor[rgb]{ .973,  .412,  .42}0,943 & -     & -     & -     & \cellcolor[rgb]{ .976,  .475,  .435}1,73\% & \cellcolor[rgb]{ .973,  .412,  .42}1,048 & - \\
		\midrule
		\textbf{DRvol0} & \cellcolor[rgb]{ .686,  .831,  .502}0,054\% & \cellcolor[rgb]{ .824,  .871,  .506}0,960\% & \cellcolor[rgb]{ .537,  .788,  .494}5,61\% & \cellcolor[rgb]{ .996,  .855,  .502}-0,342 & \cellcolor[rgb]{ .714,  .839,  .498}0,080 & \cellcolor[rgb]{ .976,  .918,  .518}0,997 & \cellcolor[rgb]{ .655,  .82,  .494}0,167 & \cellcolor[rgb]{ .69,  .831,  .502}0,042\% & \cellcolor[rgb]{ .863,  .882,  .51}0,051 & \cellcolor[rgb]{ .78,  .855,  .502}1,40\% & \cellcolor[rgb]{ .388,  .745,  .482}1,178 & 24 \\
		\midrule
		\textbf{DRMAD0} & \cellcolor[rgb]{ .996,  .863,  .506}0,044\% & \cellcolor[rgb]{ .792,  .859,  .502}0,954\% & \cellcolor[rgb]{ .996,  .855,  .502}4,63\% & \cellcolor[rgb]{ .973,  .412,  .42}-0,371 & \cellcolor[rgb]{ .808,  .867,  .506}0,084 & \cellcolor[rgb]{ .98,  .627,  .459}0,966 & \cellcolor[rgb]{ .867,  .882,  .51}0,281 & \cellcolor[rgb]{ .988,  .725,  .478}0,032\% & \cellcolor[rgb]{ .996,  .871,  .506}0,041 & \cellcolor[rgb]{ .718,  .839,  .498}1,38\% & \cellcolor[rgb]{ .996,  .871,  .506}1,145 & 27 \\
		\midrule
		\textbf{DRCVaR0} & \cellcolor[rgb]{ .635,  .82,  .498}0,055\% & \cellcolor[rgb]{ 1,  .922,  .518}0,995\% & \cellcolor[rgb]{ .624,  .816,  .498}5,52\% & \cellcolor[rgb]{ .996,  .91,  .514}-0,339 & \cellcolor[rgb]{ .847,  .875,  .506}0,085 & \cellcolor[rgb]{ .608,  .808,  .498}1,010 & \cellcolor[rgb]{ 1,  .922,  .518}0,352 & \cellcolor[rgb]{ .631,  .816,  .498}0,043\% & \cellcolor[rgb]{ .906,  .894,  .514}0,048 & \cellcolor[rgb]{ .996,  .918,  .514}1,47\% & \cellcolor[rgb]{ .51,  .78,  .49}1,173 & 17 \\
		\midrule
		\textbf{DRExpe0} & \cellcolor[rgb]{ .749,  .851,  .506}0,053\% & \cellcolor[rgb]{ .773,  .855,  .502}0,949\% & \cellcolor[rgb]{ .596,  .808,  .498}5,55\% & \cellcolor[rgb]{ .773,  .859,  .506}-0,326 & \cellcolor[rgb]{ .647,  .82,  .494}0,077 & \cellcolor[rgb]{ .427,  .757,  .486}1,016 & \cellcolor[rgb]{ .91,  .894,  .51}0,306 & \cellcolor[rgb]{ .729,  .843,  .502}0,042\% & \cellcolor[rgb]{ .929,  .902,  .514}0,047 & \cellcolor[rgb]{ .639,  .816,  .494}1,35\% & \cellcolor[rgb]{ .451,  .765,  .486}1,175 & 13 \\
		\midrule
		\textbf{DRvol1} & \cellcolor[rgb]{ .388,  .745,  .482}0,060\% & \cellcolor[rgb]{ .996,  .812,  .498}1,037\% & \cellcolor[rgb]{ .388,  .745,  .482}5,78\% & \cellcolor[rgb]{ .675,  .831,  .502}-0,321 & \cellcolor[rgb]{ 1,  .871,  .51}0,093 & \cellcolor[rgb]{ .776,  .859,  .506}1,004 & \cellcolor[rgb]{ .918,  .898,  .51}0,309 & \cellcolor[rgb]{ .388,  .745,  .482}0,048\% & \cellcolor[rgb]{ .827,  .875,  .51}0,053 & \cellcolor[rgb]{ 1,  .867,  .51}1,50\% & \cellcolor[rgb]{ .396,  .749,  .486}1,177 & 16 \\
		\midrule
		\textbf{DRMAD1} & \cellcolor[rgb]{ .467,  .769,  .49}0,058\% & \cellcolor[rgb]{ .996,  .812,  .498}1,036\% & \cellcolor[rgb]{ .525,  .784,  .49}5,63\% & \cellcolor[rgb]{ .992,  .824,  .498}-0,345 & \cellcolor[rgb]{ 1,  .922,  .518}0,091 & \cellcolor[rgb]{ .992,  .835,  .498}0,987 & \cellcolor[rgb]{ 1,  .875,  .51}0,371 & \cellcolor[rgb]{ .498,  .776,  .49}0,046\% & \cellcolor[rgb]{ .792,  .863,  .506}0,055 & \cellcolor[rgb]{ .996,  .8,  .494}1,54\% & \cellcolor[rgb]{ .475,  .773,  .49}1,175 & 19 \\
		\midrule
		\textbf{DRCVaR1} & \cellcolor[rgb]{ .655,  .824,  .498}0,055\% & \cellcolor[rgb]{ .992,  .718,  .478}1,071\% & \cellcolor[rgb]{ .996,  .922,  .518}5,09\% & \cellcolor[rgb]{ .984,  .694,  .471}-0,353 & \cellcolor[rgb]{ .984,  .596,  .455}0,103 & \cellcolor[rgb]{ .388,  .745,  .482}1,017 & \cellcolor[rgb]{ .988,  .643,  .467}0,461 & \cellcolor[rgb]{ .655,  .824,  .498}0,043\% & \cellcolor[rgb]{ 1,  .922,  .518}0,043 & \cellcolor[rgb]{ .992,  .769,  .49}1,56\% & \cellcolor[rgb]{ 1,  .922,  .518}1,155 & 13 \\
		\midrule
		\textbf{DRExpe1} & \cellcolor[rgb]{ .639,  .82,  .498}0,055\% & \cellcolor[rgb]{ .992,  .733,  .482}1,065\% & \cellcolor[rgb]{ .949,  .906,  .518}5,15\% & \cellcolor[rgb]{ .973,  .914,  .518}-0,337 & \cellcolor[rgb]{ .988,  .647,  .467}0,101 & \cellcolor[rgb]{ .725,  .843,  .502}1,006 & \cellcolor[rgb]{ .996,  .792,  .494}0,404 & \cellcolor[rgb]{ .643,  .82,  .498}0,043\% & \cellcolor[rgb]{ .976,  .918,  .518}0,044 & \cellcolor[rgb]{ .992,  .71,  .478}1,59\% & \cellcolor[rgb]{ .878,  .886,  .514}1,160 & 15 \\
		\bottomrule
	\end{tabular}}%
	\label{tab:CompRes_FTSE100}%
\end{table}%
%



In Table \ref{tab:CompRes_NASDAQ100} we report the computational results on NASDAQ100.
Again, we note that the return-diversification models generally outperform the other strategies.
Indeed, DRVol1, DRMAD1, DRCVaR1, and DRExpe1 typically obtain superior results in terms of $\mu^{out}$, Jensen's Alpha, Sharpe, Rachev, Information, and Omega.
%
%
However, for mildly risk-averse investors, DRVol0, DRMAD0, DRCVaR0, and DRExpe0 might be advisable, while for very risk-averse investors, strategies based on risk minimization are recommended.
%
\begin{table}[htbp]
	\centering
	\caption{Out-of-sample performance results on the NASDAQ100 dataset}
	\resizebox{0.90\textwidth}{!}{\begin{tabular}{|l|c|c|c|c|c|c|c|c|c|c|c|c|}
		\toprule
		\multicolumn{1}{|c|}{\textbf{Approach}} & $\boldsymbol{\mu^{out}}$ & $\boldsymbol{\sigma^{out}}$ & \textbf{Sharpe} & \textbf{MDD} & \textbf{Ulcer} & \textbf{Rachev10} & \textbf{Turn} & \textbf{AlphaJ} & \textbf{InfoR} & \textbf{VaR5} & \textbf{Omega} & \textbf{ave \#} \\
		\midrule
		\textbf{MV0} & \cellcolor[rgb]{ .973,  .42,  .42}0,043\% & \cellcolor[rgb]{ .388,  .745,  .482}0,989\% & \cellcolor[rgb]{ .973,  .439,  .424}4,33\% & \cellcolor[rgb]{ .478,  .773,  .49}-0,361 & \cellcolor[rgb]{ .624,  .812,  .494}0,085 & \cellcolor[rgb]{ .98,  .616,  .459}0,956 & \cellcolor[rgb]{ .627,  .812,  .494}0,141 & \cellcolor[rgb]{ .973,  .439,  .424}0,009\% & \cellcolor[rgb]{ .973,  .427,  .42}-0,022 & \cellcolor[rgb]{ .4,  .745,  .482}1,34\% & \cellcolor[rgb]{ .973,  .451,  .424}1,146 & 15 \\
		\midrule
		\textbf{MAD0} & \cellcolor[rgb]{ .973,  .412,  .42}0,042\% & \cellcolor[rgb]{ .388,  .745,  .482}0,990\% & \cellcolor[rgb]{ .973,  .412,  .42}4,24\% & \cellcolor[rgb]{ .451,  .765,  .486}-0,358 & \cellcolor[rgb]{ .408,  .749,  .482}0,079 & \cellcolor[rgb]{ .984,  .647,  .463}0,958 & \cellcolor[rgb]{ .706,  .835,  .498}0,181 & \cellcolor[rgb]{ .973,  .412,  .42}0,008\% & \cellcolor[rgb]{ .973,  .412,  .42}-0,023 & \cellcolor[rgb]{ .388,  .745,  .482}1,33\% & \cellcolor[rgb]{ .973,  .424,  .42}1,143 & 19 \\
		\midrule
		\textbf{CVaR0} & \cellcolor[rgb]{ .976,  .529,  .439}0,051\% & \cellcolor[rgb]{ .424,  .753,  .482}1,008\% & \cellcolor[rgb]{ .984,  .663,  .467}5,02\% & \cellcolor[rgb]{ .388,  .745,  .482}-0,350 & \cellcolor[rgb]{ .647,  .82,  .494}0,085 & \cellcolor[rgb]{ .973,  .914,  .518}0,975 & \cellcolor[rgb]{ .898,  .89,  .51}0,278 & \cellcolor[rgb]{ .98,  .608,  .455}0,018\% & \cellcolor[rgb]{ .976,  .533,  .443}-0,013 & \cellcolor[rgb]{ .435,  .757,  .482}1,38\% & \cellcolor[rgb]{ .984,  .663,  .467}1,167 & 11 \\
		\midrule
		\textbf{Expe0} & \cellcolor[rgb]{ .973,  .482,  .431}0,047\% & \cellcolor[rgb]{ .42,  .753,  .482}1,005\% & \cellcolor[rgb]{ .98,  .557,  .447}4,70\% & \cellcolor[rgb]{ .604,  .808,  .498}-0,377 & \cellcolor[rgb]{ 1,  .922,  .518}0,093 & \cellcolor[rgb]{ .976,  .553,  .443}0,953 & \cellcolor[rgb]{ .914,  .894,  .51}0,285 & \cellcolor[rgb]{ .976,  .522,  .439}0,014\% & \cellcolor[rgb]{ .976,  .486,  .431}-0,017 & \cellcolor[rgb]{ .447,  .761,  .482}1,39\% & \cellcolor[rgb]{ .976,  .549,  .443}1,155 & 14 \\
		\midrule
		\textbf{MV1} & \cellcolor[rgb]{ .941,  .906,  .518}0,082\% & \cellcolor[rgb]{ .988,  .678,  .475}1,456\% & \cellcolor[rgb]{ .996,  .871,  .506}5,66\% & \cellcolor[rgb]{ .961,  .914,  .518}-0,423 & \cellcolor[rgb]{ .976,  .447,  .427}0,143 & \cellcolor[rgb]{ .988,  .737,  .482}0,963 & \cellcolor[rgb]{ .984,  .592,  .455}0,490 & \cellcolor[rgb]{ .973,  .914,  .518}0,034\% & \cellcolor[rgb]{ .988,  .922,  .518}0,020 & \cellcolor[rgb]{ .988,  .639,  .467}2,29\% & \cellcolor[rgb]{ .988,  .745,  .482}1,175 & 10 \\
		\midrule
		\textbf{MAD1} & \cellcolor[rgb]{ .937,  .906,  .518}0,083\% & \cellcolor[rgb]{ .992,  .757,  .486}1,408\% & \cellcolor[rgb]{ .976,  .918,  .518}5,86\% & \cellcolor[rgb]{ .996,  .859,  .506}-0,442 & \cellcolor[rgb]{ .973,  .412,  .42}0,146 & \cellcolor[rgb]{ .984,  .918,  .518}0,974 & \cellcolor[rgb]{ .984,  .565,  .451}0,502 & \cellcolor[rgb]{ .961,  .914,  .518}0,034\% & \cellcolor[rgb]{ .965,  .914,  .518}0,021 & \cellcolor[rgb]{ .992,  .725,  .482}2,19\% & \cellcolor[rgb]{ .992,  .808,  .494}1,182 & 12 \\
		\midrule
		\textbf{CVaR1} & \cellcolor[rgb]{ 1,  .922,  .518}0,079\% & \cellcolor[rgb]{ .996,  .788,  .494}1,389\% & \cellcolor[rgb]{ .996,  .875,  .506}5,67\% & \cellcolor[rgb]{ .902,  .894,  .514}-0,415 & \cellcolor[rgb]{ .98,  .557,  .451}0,131 & \cellcolor[rgb]{ .922,  .902,  .514}0,978 & \cellcolor[rgb]{ .973,  .412,  .42}0,576 & \cellcolor[rgb]{ .969,  .914,  .518}0,034\% & \cellcolor[rgb]{ .996,  .878,  .506}0,016 & \cellcolor[rgb]{ .996,  .824,  .502}2,07\% & \cellcolor[rgb]{ .992,  .773,  .486}1,178 & 9 \\
		\midrule
		\textbf{Expe1} & \cellcolor[rgb]{ .996,  .878,  .51}0,076\% & \cellcolor[rgb]{ 1,  .922,  .518}1,308\% & \cellcolor[rgb]{ 1,  .922,  .518}5,81\% & \cellcolor[rgb]{ 1,  .922,  .518}-0,428 & \cellcolor[rgb]{ .98,  .553,  .447}0,132 & \cellcolor[rgb]{ .996,  .906,  .514}0,972 & \cellcolor[rgb]{ .984,  .561,  .451}0,504 & \cellcolor[rgb]{ .996,  .894,  .51}0,032\% & \cellcolor[rgb]{ .996,  .863,  .506}0,015 & \cellcolor[rgb]{ 1,  .922,  .518}1,95\% & \cellcolor[rgb]{ .992,  .82,  .498}1,183 & 20 \\
		\midrule
		\textbf{EW} & \cellcolor[rgb]{ 1,  .922,  .518}0,079\% & \cellcolor[rgb]{ .996,  .824,  .502}1,368\% & \cellcolor[rgb]{ .996,  .91,  .514}5,77\% & \cellcolor[rgb]{ .976,  .918,  .518}-0,425 & \cellcolor[rgb]{ .675,  .827,  .498}0,086 & \cellcolor[rgb]{ .992,  .843,  .502}0,969 & -     & \cellcolor[rgb]{ .984,  .698,  .475}0,022\% & \cellcolor[rgb]{ .612,  .812,  .498}0,041 & \cellcolor[rgb]{ .992,  .765,  .49}2,14\% & \cellcolor[rgb]{ 1,  .922,  .518}1,193 & 70 \\
		\midrule
		\textbf{RP} & \cellcolor[rgb]{ .992,  .831,  .498}0,072\% & \cellcolor[rgb]{ .918,  .898,  .51}1,266\% & \cellcolor[rgb]{ .996,  .89,  .51}5,72\% & \cellcolor[rgb]{ .878,  .886,  .514}-0,412 & \cellcolor[rgb]{ .388,  .745,  .482}0,079 & \cellcolor[rgb]{ .992,  .835,  .498}0,969 & \cellcolor[rgb]{ .388,  .745,  .482}0,020 & \cellcolor[rgb]{ .984,  .659,  .467}0,020\% & \cellcolor[rgb]{ .965,  .914,  .518}0,021 & \cellcolor[rgb]{ .996,  .918,  .514}1,95\% & \cellcolor[rgb]{ .976,  .918,  .518}1,194 & 70 \\
		\midrule
		\textbf{Index} & \cellcolor[rgb]{ .984,  .694,  .471}0,063\% & \cellcolor[rgb]{ .988,  .698,  .478}1,444\% & \cellcolor[rgb]{ .973,  .439,  .424}4,34\% & \cellcolor[rgb]{ .906,  .894,  .514}-0,416 & \cellcolor[rgb]{ .996,  .784,  .494}0,108 & \cellcolor[rgb]{ .973,  .412,  .42}0,945 & -     & -     & -     & \cellcolor[rgb]{ .984,  .608,  .459}2,33\% & \cellcolor[rgb]{ .973,  .412,  .42}1,141 & - \\
		\midrule
		\textbf{DRvol0} & \cellcolor[rgb]{ .941,  .906,  .518}0,082\% & \cellcolor[rgb]{ .953,  .906,  .514}1,284\% & \cellcolor[rgb]{ .71,  .839,  .502}6,41\% & \cellcolor[rgb]{ .98,  .604,  .455}-0,501 & \cellcolor[rgb]{ .565,  .796,  .49}0,083 & \cellcolor[rgb]{ .388,  .745,  .482}1,013 & \cellcolor[rgb]{ .659,  .824,  .498}0,158 & \cellcolor[rgb]{ .91,  .898,  .514}0,037\% & \cellcolor[rgb]{ .925,  .902,  .514}0,024 & \cellcolor[rgb]{ .82,  .867,  .506}1,77\% & \cellcolor[rgb]{ .459,  .765,  .486}1,220 & 21 \\
		\midrule
		\textbf{DRMAD0} & \cellcolor[rgb]{ .992,  .812,  .494}0,071\% & \cellcolor[rgb]{ .824,  .871,  .506}1,217\% & \cellcolor[rgb]{ .984,  .918,  .518}5,84\% & \cellcolor[rgb]{ .984,  .651,  .463}-0,490 & \cellcolor[rgb]{ .435,  .757,  .482}0,080 & \cellcolor[rgb]{ .882,  .89,  .514}0,981 & \cellcolor[rgb]{ .863,  .882,  .51}0,260 & \cellcolor[rgb]{ .992,  .788,  .49}0,026\% & \cellcolor[rgb]{ .992,  .82,  .498}0,011 & \cellcolor[rgb]{ .769,  .855,  .502}1,72\% & \cellcolor[rgb]{ .851,  .878,  .51}1,200 & 25 \\
		\midrule
		\textbf{DRCVaR0} & \cellcolor[rgb]{ .992,  .922,  .518}0,079\% & \cellcolor[rgb]{ .902,  .894,  .51}1,258\% & \cellcolor[rgb]{ .761,  .855,  .506}6,31\% & \cellcolor[rgb]{ .973,  .412,  .42}-0,546 & \cellcolor[rgb]{ .722,  .839,  .498}0,087 & \cellcolor[rgb]{ .643,  .82,  .498}0,996 & \cellcolor[rgb]{ .996,  .816,  .498}0,380 & \cellcolor[rgb]{ .914,  .898,  .514}0,036\% & \cellcolor[rgb]{ .996,  .91,  .514}0,019 & \cellcolor[rgb]{ .8,  .863,  .506}1,75\% & \cellcolor[rgb]{ .584,  .804,  .494}1,213 & 15 \\
		\midrule
		\textbf{DRExpe0} & \cellcolor[rgb]{ .992,  .843,  .502}0,073\% & \cellcolor[rgb]{ .812,  .867,  .506}1,210\% & \cellcolor[rgb]{ .886,  .89,  .514}6,05\% & \cellcolor[rgb]{ .98,  .624,  .459}-0,497 & \cellcolor[rgb]{ .537,  .788,  .49}0,082 & \cellcolor[rgb]{ .871,  .886,  .514}0,982 & \cellcolor[rgb]{ .98,  .914,  .514}0,318 & \cellcolor[rgb]{ .996,  .871,  .506}0,030\% & \cellcolor[rgb]{ .992,  .839,  .502}0,013 & \cellcolor[rgb]{ .741,  .847,  .502}1,69\% & \cellcolor[rgb]{ .82,  .871,  .51}1,202 & 11 \\
		\midrule
		\textbf{DRvol1} & \cellcolor[rgb]{ .388,  .745,  .482}0,114\% & \cellcolor[rgb]{ .973,  .412,  .42}1,617\% & \cellcolor[rgb]{ .388,  .745,  .482}7,07\% & \cellcolor[rgb]{ .988,  .757,  .486}-0,466 & \cellcolor[rgb]{ .988,  .655,  .467}0,121 & \cellcolor[rgb]{ .808,  .867,  .51}0,986 & \cellcolor[rgb]{ 1,  .922,  .518}0,328 & \cellcolor[rgb]{ .388,  .745,  .482}0,058\% & \cellcolor[rgb]{ .388,  .745,  .482}0,053 & \cellcolor[rgb]{ .973,  .412,  .42}2,56\% & \cellcolor[rgb]{ .388,  .745,  .482}1,223 & 13 \\
		\midrule
		\textbf{DRMAD1} & \cellcolor[rgb]{ .675,  .827,  .502}0,098\% & \cellcolor[rgb]{ .98,  .49,  .435}1,570\% & \cellcolor[rgb]{ .796,  .863,  .506}6,23\% & \cellcolor[rgb]{ .98,  .627,  .459}-0,496 & \cellcolor[rgb]{ .984,  .584,  .455}0,128 & \cellcolor[rgb]{ .996,  .894,  .51}0,972 & \cellcolor[rgb]{ .996,  .827,  .502}0,376 & \cellcolor[rgb]{ .784,  .863,  .506}0,042\% & \cellcolor[rgb]{ .643,  .82,  .498}0,039 & \cellcolor[rgb]{ .98,  .506,  .439}2,45\% & \cellcolor[rgb]{ .957,  .91,  .518}1,195 & 15 \\
		\midrule
		\textbf{DRCVaR1} & \cellcolor[rgb]{ .694,  .835,  .502}0,097\% & \cellcolor[rgb]{ .98,  .557,  .451}1,530\% & \cellcolor[rgb]{ .753,  .851,  .506}6,32\% & \cellcolor[rgb]{ .992,  .804,  .494}-0,455 & \cellcolor[rgb]{ .988,  .682,  .475}0,118 & \cellcolor[rgb]{ 1,  .922,  .518}0,973 & \cellcolor[rgb]{ .984,  .612,  .459}0,479 & \cellcolor[rgb]{ .706,  .839,  .502}0,045\% & \cellcolor[rgb]{ .729,  .847,  .506}0,034 & \cellcolor[rgb]{ .98,  .533,  .443}2,42\% & \cellcolor[rgb]{ .89,  .89,  .514}1,198 & 11 \\
		\midrule
		\textbf{DRExpe1} & \cellcolor[rgb]{ .741,  .847,  .506}0,094\% & \cellcolor[rgb]{ .988,  .694,  .475}1,448\% & \cellcolor[rgb]{ .675,  .827,  .502}6,49\% & \cellcolor[rgb]{ .98,  .608,  .455}-0,500 & \cellcolor[rgb]{ .992,  .729,  .482}0,113 & \cellcolor[rgb]{ .937,  .906,  .518}0,977 & \cellcolor[rgb]{ .992,  .718,  .478}0,428 & \cellcolor[rgb]{ .761,  .855,  .506}0,043\% & \cellcolor[rgb]{ .706,  .839,  .502}0,035 & \cellcolor[rgb]{ .992,  .725,  .482}2,19\% & \cellcolor[rgb]{ .718,  .839,  .502}1,207 & 15 \\
		\bottomrule
	\end{tabular}}%
	\label{tab:CompRes_NASDAQ100}%
\end{table}%
%



Finally, Table \ref{tab:CompRes_SP500} shows the out-of-sample performance results on S$\&$P500.
In this case, the return-diversification portfolios provide interesting results both in terms of gain and risk.
Indeed, DRvol0, DRMAD0, DRCVaR0, and DRExpe0 achieve good values of Sharpe, Rachev, and Omega while preserving control over risk, as highlighted by $\sigma^{out}$, Ulcer, and VaR5.
Furthermore, we observe that including the return constraint typically allows for achieving good performance in terms of $\mu^{out}$, Jensen's Alpha, and Information ratio, although with a worsening of risk.
%
\begin{table}[htbp]
	\centering
	\caption{Out-of-sample performance results on the SP500 dataset}
	\resizebox{0.90\textwidth}{!}{\begin{tabular}{|l|c|c|c|c|c|c|c|c|c|c|c|c|}
		\toprule
		\multicolumn{1}{|c|}{\textbf{Approach}} & $\boldsymbol{\mu^{out}}$ & $\boldsymbol{\sigma^{out}}$ & \textbf{Sharpe} & \textbf{MDD} & \textbf{Ulcer} & \textbf{Rachev10} & \textbf{Turn} & \textbf{AlphaJ} & \textbf{InfoR} & \textbf{VaR5} & \textbf{Omega} & \textbf{ave \#} \\
		\midrule
		\textbf{MV0} & \cellcolor[rgb]{ .973,  .443,  .424}0,030\% & \cellcolor[rgb]{ .431,  .757,  .482}0,903\% & \cellcolor[rgb]{ .976,  .533,  .443}3,30\% & \cellcolor[rgb]{ .506,  .78,  .49}-0,362 & \cellcolor[rgb]{ .682,  .827,  .498}0,082 & \cellcolor[rgb]{ .996,  .859,  .506}0,941 & \cellcolor[rgb]{ .694,  .831,  .498}0,252 & \cellcolor[rgb]{ .976,  .522,  .439}0,007\% & \cellcolor[rgb]{ .973,  .475,  .431}-0,014 & \cellcolor[rgb]{ .388,  .745,  .482}1,20\% & \cellcolor[rgb]{ .976,  .533,  .439}1,113 & 28 \\
		\midrule
		\textbf{MAD0} & \cellcolor[rgb]{ .973,  .412,  .42}0,028\% & \cellcolor[rgb]{ .616,  .808,  .494}1,014\% & \cellcolor[rgb]{ .973,  .412,  .42}2,73\% & \cellcolor[rgb]{ .537,  .788,  .494}-0,367 & \cellcolor[rgb]{ .871,  .882,  .51}0,087 & \cellcolor[rgb]{ .98,  .565,  .447}0,924 & \cellcolor[rgb]{ 1,  .867,  .51}0,501 & \cellcolor[rgb]{ .973,  .412,  .42}0,000\% & \cellcolor[rgb]{ .973,  .412,  .42}-0,020 & \cellcolor[rgb]{ .439,  .757,  .482}1,26\% & \cellcolor[rgb]{ .973,  .412,  .42}1,098 & 61 \\
		\midrule
		\textbf{CVaR0} & \cellcolor[rgb]{ .973,  .447,  .424}0,030\% & \cellcolor[rgb]{ .388,  .745,  .482}0,875\% & \cellcolor[rgb]{ .98,  .561,  .447}3,43\% & \cellcolor[rgb]{ .388,  .745,  .482}-0,341 & \cellcolor[rgb]{ .455,  .765,  .486}0,074 & \cellcolor[rgb]{ .992,  .835,  .498}0,940 & \cellcolor[rgb]{ .98,  .914,  .514}0,463 & \cellcolor[rgb]{ .976,  .557,  .447}0,009\% & \cellcolor[rgb]{ .976,  .486,  .431}-0,013 & \cellcolor[rgb]{ .471,  .769,  .486}1,29\% & \cellcolor[rgb]{ .976,  .506,  .435}1,110 & 19 \\
		\midrule
		\textbf{Expe0} & \cellcolor[rgb]{ .976,  .502,  .435}0,033\% & \cellcolor[rgb]{ .443,  .761,  .482}0,910\% & \cellcolor[rgb]{ .98,  .616,  .459}3,68\% & \cellcolor[rgb]{ .529,  .788,  .494}-0,366 & \cellcolor[rgb]{ .792,  .859,  .502}0,085 & \cellcolor[rgb]{ .984,  .671,  .467}0,930 & \cellcolor[rgb]{ 1,  .922,  .518}0,477 & \cellcolor[rgb]{ .98,  .588,  .451}0,011\% & \cellcolor[rgb]{ .976,  .529,  .439}-0,009 & \cellcolor[rgb]{ .463,  .765,  .486}1,28\% & \cellcolor[rgb]{ .98,  .608,  .455}1,122 & 24 \\
		\midrule
		\textbf{MV1} & \cellcolor[rgb]{ .855,  .882,  .51}0,065\% & \cellcolor[rgb]{ .992,  .706,  .478}1,345\% & \cellcolor[rgb]{ .996,  .875,  .506}4,87\% & \cellcolor[rgb]{ .996,  .851,  .502}-0,466 & \cellcolor[rgb]{ .98,  .533,  .443}0,145 & \cellcolor[rgb]{ 1,  .922,  .518}0,945 & \cellcolor[rgb]{ .988,  .647,  .467}0,593 & \cellcolor[rgb]{ .871,  .886,  .514}0,034\% & \cellcolor[rgb]{ .996,  .91,  .514}0,027 & \cellcolor[rgb]{ .992,  .753,  .486}1,99\% & \cellcolor[rgb]{ .996,  .886,  .51}1,156 & 16 \\
		\midrule
		\textbf{MAD1} & \cellcolor[rgb]{ .945,  .906,  .518}0,062\% & \cellcolor[rgb]{ .996,  .82,  .498}1,293\% & \cellcolor[rgb]{ .996,  .855,  .502}4,79\% & \cellcolor[rgb]{ .992,  .847,  .502}-0,467 & \cellcolor[rgb]{ .984,  .576,  .451}0,139 & \cellcolor[rgb]{ .992,  .792,  .49}0,938 & \cellcolor[rgb]{ .984,  .565,  .451}0,628 & \cellcolor[rgb]{ 1,  .922,  .518}0,031\% & \cellcolor[rgb]{ .996,  .886,  .51}0,025 & \cellcolor[rgb]{ .992,  .761,  .49}1,98\% & \cellcolor[rgb]{ .996,  .855,  .502}1,153 & 18 \\
		\midrule
		\textbf{CVaR1} & \cellcolor[rgb]{ .863,  .882,  .51}0,065\% & \cellcolor[rgb]{ .996,  .847,  .506}1,280\% & \cellcolor[rgb]{ .992,  .922,  .518}5,09\% & \cellcolor[rgb]{ .729,  .843,  .502}-0,401 & \cellcolor[rgb]{ .992,  .737,  .482}0,117 & \cellcolor[rgb]{ .969,  .914,  .518}0,947 & \cellcolor[rgb]{ .973,  .412,  .42}0,691 & \cellcolor[rgb]{ .804,  .867,  .51}0,036\% & \cellcolor[rgb]{ .996,  .906,  .514}0,027 & \cellcolor[rgb]{ .996,  .796,  .494}1,95\% & \cellcolor[rgb]{ .941,  .906,  .518}1,165 & 13 \\
		\midrule
		\textbf{Expe1} & \cellcolor[rgb]{ .827,  .875,  .51}0,067\% & \cellcolor[rgb]{ .965,  .91,  .514}1,224\% & \cellcolor[rgb]{ .753,  .851,  .506}5,44\% & \cellcolor[rgb]{ .757,  .851,  .506}-0,407 & \cellcolor[rgb]{ .992,  .741,  .486}0,116 & \cellcolor[rgb]{ .992,  .922,  .518}0,945 & \cellcolor[rgb]{ .98,  .541,  .447}0,637 & \cellcolor[rgb]{ .749,  .851,  .506}0,037\% & \cellcolor[rgb]{ .953,  .91,  .518}0,031 & \cellcolor[rgb]{ 1,  .922,  .518}1,83\% & \cellcolor[rgb]{ .796,  .863,  .506}1,174 & 30 \\
		\midrule
		\textbf{EW} & \cellcolor[rgb]{ 1,  .922,  .518}0,060\% & \cellcolor[rgb]{ .984,  .631,  .463}1,379\% & \cellcolor[rgb]{ .988,  .757,  .482}4,32\% & \cellcolor[rgb]{ .992,  .839,  .502}-0,468 & \cellcolor[rgb]{ 1,  .922,  .518}0,091 & \cellcolor[rgb]{ .682,  .831,  .502}0,962 & -     & \cellcolor[rgb]{ .988,  .702,  .475}0,018\% & \cellcolor[rgb]{ .388,  .745,  .482}0,062 & \cellcolor[rgb]{ .988,  .69,  .475}2,04\% & \cellcolor[rgb]{ .992,  .831,  .498}1,150 & 420 \\
		\midrule
		\textbf{RP} & \cellcolor[rgb]{ .992,  .824,  .498}0,054\% & \cellcolor[rgb]{ 1,  .922,  .518}1,245\% & \cellcolor[rgb]{ .988,  .753,  .482}4,30\% & \cellcolor[rgb]{ .929,  .902,  .514}-0,437 & \cellcolor[rgb]{ .796,  .863,  .506}0,085 & \cellcolor[rgb]{ .788,  .863,  .506}0,956 & \cellcolor[rgb]{ .388,  .745,  .482}0,022 & \cellcolor[rgb]{ .984,  .671,  .467}0,016\% & \cellcolor[rgb]{ .655,  .824,  .498}0,047 & \cellcolor[rgb]{ .976,  .914,  .514}1,80\% & \cellcolor[rgb]{ .992,  .831,  .498}1,150 & 418 \\
		\midrule
		\textbf{Index} & \cellcolor[rgb]{ .98,  .612,  .455}0,040\% & \cellcolor[rgb]{ .996,  .824,  .502}1,291\% & \cellcolor[rgb]{ .976,  .498,  .435}3,13\% & \cellcolor[rgb]{ .996,  .847,  .502}-0,466 & \cellcolor[rgb]{ .996,  .816,  .498}0,106 & \cellcolor[rgb]{ .98,  .62,  .459}0,927 & -     & -     & -     & \cellcolor[rgb]{ .996,  .808,  .498}1,93\% & \cellcolor[rgb]{ .973,  .478,  .431}1,106 & - \\
		\midrule
		\textbf{DRvol0} & \cellcolor[rgb]{ .816,  .871,  .51}0,067\% & \cellcolor[rgb]{ .8,  .863,  .506}1,124\% & \cellcolor[rgb]{ .388,  .745,  .482}5,97\% & \cellcolor[rgb]{ .89,  .89,  .514}-0,430 & \cellcolor[rgb]{ .451,  .761,  .482}0,074 & \cellcolor[rgb]{ .651,  .824,  .498}0,963 & \cellcolor[rgb]{ .682,  .827,  .498}0,242 & \cellcolor[rgb]{ .765,  .855,  .506}0,037\% & \cellcolor[rgb]{ .757,  .855,  .506}0,041 & \cellcolor[rgb]{ .8,  .863,  .506}1,62\% & \cellcolor[rgb]{ .388,  .745,  .482}1,201 & 34 \\
		\midrule
		\textbf{DRMAD0} & \cellcolor[rgb]{ .996,  .875,  .506}0,057\% & \cellcolor[rgb]{ .714,  .839,  .498}1,073\% & \cellcolor[rgb]{ .859,  .882,  .51}5,29\% & \cellcolor[rgb]{ .839,  .878,  .51}-0,422 & \cellcolor[rgb]{ .388,  .745,  .482}0,072 & \cellcolor[rgb]{ .769,  .855,  .506}0,957 & \cellcolor[rgb]{ .894,  .89,  .51}0,398 & \cellcolor[rgb]{ .996,  .863,  .506}0,027\% & \cellcolor[rgb]{ .996,  .914,  .514}0,027 & \cellcolor[rgb]{ .678,  .827,  .498}1,50\% & \cellcolor[rgb]{ .741,  .847,  .506}1,178 & 40 \\
		\midrule
		\textbf{DRCVaR0} & \cellcolor[rgb]{ .996,  .878,  .51}0,057\% & \cellcolor[rgb]{ .753,  .851,  .502}1,097\% & \cellcolor[rgb]{ .918,  .898,  .514}5,20\% & \cellcolor[rgb]{ .992,  .808,  .494}-0,475 & \cellcolor[rgb]{ .867,  .882,  .51}0,087 & \cellcolor[rgb]{ .388,  .745,  .482}0,977 & \cellcolor[rgb]{ .996,  .918,  .514}0,476 & \cellcolor[rgb]{ .996,  .886,  .51}0,029\% & \cellcolor[rgb]{ .996,  .871,  .506}0,023 & \cellcolor[rgb]{ .745,  .847,  .502}1,57\% & \cellcolor[rgb]{ .867,  .886,  .514}1,170 & 23 \\
		\midrule
		\textbf{DRExpe0} & \cellcolor[rgb]{ .996,  .914,  .514}0,059\% & \cellcolor[rgb]{ .722,  .839,  .498}1,078\% & \cellcolor[rgb]{ .71,  .839,  .502}5,50\% & \cellcolor[rgb]{ 1,  .922,  .518}-0,450 & \cellcolor[rgb]{ .682,  .827,  .498}0,081 & \cellcolor[rgb]{ .631,  .816,  .498}0,964 & \cellcolor[rgb]{ .953,  .906,  .514}0,443 & \cellcolor[rgb]{ .996,  .918,  .514}0,030\% & \cellcolor[rgb]{ .992,  .922,  .518}0,028 & \cellcolor[rgb]{ .71,  .835,  .498}1,53\% & \cellcolor[rgb]{ .667,  .827,  .502}1,183 & 16 \\
		\midrule
		\textbf{DRvol1} & \cellcolor[rgb]{ .388,  .745,  .482}0,084\% & \cellcolor[rgb]{ .973,  .412,  .42}1,479\% & \cellcolor[rgb]{ .58,  .8,  .494}5,69\% & \cellcolor[rgb]{ .976,  .525,  .439}-0,538 & \cellcolor[rgb]{ .984,  .62,  .463}0,133 & \cellcolor[rgb]{ .984,  .643,  .463}0,929 & \cellcolor[rgb]{ .914,  .894,  .51}0,414 & \cellcolor[rgb]{ .388,  .745,  .482}0,046\% & \cellcolor[rgb]{ .584,  .804,  .494}0,051 & \cellcolor[rgb]{ .973,  .412,  .42}2,31\% & \cellcolor[rgb]{ .694,  .835,  .502}1,181 & 21 \\
		\midrule
		\textbf{DRMAD1} & \cellcolor[rgb]{ .639,  .82,  .498}0,074\% & \cellcolor[rgb]{ .98,  .506,  .439}1,436\% & \cellcolor[rgb]{ .941,  .906,  .518}5,16\% & \cellcolor[rgb]{ .973,  .412,  .42}-0,565 & \cellcolor[rgb]{ .973,  .412,  .42}0,162 & \cellcolor[rgb]{ .973,  .412,  .42}0,915 & \cellcolor[rgb]{ 1,  .898,  .514}0,488 & \cellcolor[rgb]{ .749,  .851,  .506}0,037\% & \cellcolor[rgb]{ .761,  .855,  .506}0,041 & \cellcolor[rgb]{ .98,  .533,  .443}2,19\% & \cellcolor[rgb]{ .953,  .91,  .518}1,164 & 24 \\
		\midrule
		\textbf{DRCVaR1} & \cellcolor[rgb]{ .686,  .831,  .502}0,072\% & \cellcolor[rgb]{ .98,  .533,  .443}1,424\% & \cellcolor[rgb]{ 1,  .922,  .518}5,08\% & \cellcolor[rgb]{ .984,  .647,  .463}-0,512 & \cellcolor[rgb]{ .984,  .631,  .463}0,132 & \cellcolor[rgb]{ .792,  .863,  .506}0,956 & \cellcolor[rgb]{ .984,  .588,  .455}0,618 & \cellcolor[rgb]{ .745,  .851,  .506}0,037\% & \cellcolor[rgb]{ .863,  .882,  .51}0,036 & \cellcolor[rgb]{ .98,  .541,  .447}2,18\% & \cellcolor[rgb]{ 1,  .922,  .518}1,161 & 16 \\
		\midrule
		\textbf{DRExpe1} & \cellcolor[rgb]{ .722,  .843,  .502}0,071\% & \cellcolor[rgb]{ .988,  .694,  .475}1,351\% & \cellcolor[rgb]{ .886,  .89,  .514}5,24\% & \cellcolor[rgb]{ .976,  .549,  .443}-0,534 & \cellcolor[rgb]{ .984,  .612,  .459}0,134 & \cellcolor[rgb]{ .98,  .557,  .447}0,924 & \cellcolor[rgb]{ .996,  .792,  .494}0,532 & \cellcolor[rgb]{ .78,  .859,  .506}0,036\% & \cellcolor[rgb]{ .82,  .871,  .51}0,038 & \cellcolor[rgb]{ .992,  .733,  .482}2,00\% & \cellcolor[rgb]{ .898,  .894,  .514}1,168 & 20 \\
		\bottomrule
	\end{tabular}}%
	\label{tab:CompRes_SP500}%
\end{table}%
%


\section{Conclusions}\label{sec:Conclusions}

We have proposed a new return-diversification approach to portfolio selection and we have focused on the case of the diversification ratio for general and specific risk measures.
An extensive empirical analysis based on several real-world data sets has shown promising out-of-sample performance results of our approach.

\noindent
The theoretical and empirical analysis of the models based on our bi-objective approach using different diversification measures is left for future research.
It is also worth investigating further the connections between the Risk Parity and the Maximum Diversification ratio approaches both in the case of volatility and in the case of general risk measures, possibly trying to establish their equivalence under appropriate assumptions.


{\footnotesize
\bibliographystyle{apa}
\bibliography{BibbaseBCG_MaxDiv_20230703}
}

\appendix

\section{Additional Out-of-Sample Performance Results}\label{sec:Appendix}


For the sake of completeness, we report here additional tables containing some statistics of the out-of-sample annual ROI obtained by all the portfolio selection strategies listed in Table \ref{tab:ListOfModels} on DowJones (Table \ref{tab:ROI250_DJ}), EuroStoxx 50 (Table \ref{tab:ROI250_EUXX50}), NASDAQ100 (Table \ref{tab:ROI250_NASDAQ100}), FTSE100 (Table \ref{tab:ROI250_FTSE100}), and S$\&$P500 (Table \ref{tab:ROI250_SP500}).
%
\begin{table}[htbp!]
	\centering
	\caption{Annual ROI on the DowJones dataset}
	\resizebox{0.7\textwidth}{!}{\begin{tabular}{|l|c|c|c|c|c|c|c|}
		\toprule
		\multicolumn{1}{|c|}{\textbf{Approach}} & \textbf{Mean} & \textbf{Vol} & \textbf{5\%-perc} & \textbf{25\%-perc} & \textbf{50\%-perc} & \textbf{75\%-perc} & \textbf{95\%-perc} \\
		\midrule
		\textbf{MV0} & \cellcolor[rgb]{ .973,  .412,  .42}7,09\% & \cellcolor[rgb]{ .518,  .78,  .486}8,58\% & \cellcolor[rgb]{ .988,  .749,  .482}-7,04\% & \cellcolor[rgb]{ .973,  .412,  .42}1,77\% & \cellcolor[rgb]{ .973,  .412,  .42}7,05\% & \cellcolor[rgb]{ .973,  .412,  .42}12,55\% & \cellcolor[rgb]{ .973,  .431,  .424}21,09\% \\
		\midrule
		\textbf{MAD0} & \cellcolor[rgb]{ .976,  .529,  .439}8,47\% & \cellcolor[rgb]{ .549,  .788,  .49}8,75\% & \cellcolor[rgb]{ .992,  .835,  .498}-6,69\% & \cellcolor[rgb]{ .98,  .6,  .455}2,98\% & \cellcolor[rgb]{ .98,  .584,  .451}9,15\% & \cellcolor[rgb]{ .98,  .565,  .447}14,30\% & \cellcolor[rgb]{ .976,  .506,  .435}22,50\% \\
		\midrule
		\textbf{CVaR0} & \cellcolor[rgb]{ .973,  .478,  .431}7,88\% & \cellcolor[rgb]{ .388,  .745,  .482}7,92\% & \cellcolor[rgb]{ .949,  .91,  .518}-6,16\% & \cellcolor[rgb]{ .984,  .639,  .463}3,23\% & \cellcolor[rgb]{ .973,  .475,  .431}7,84\% & \cellcolor[rgb]{ .973,  .471,  .427}13,23\% & \cellcolor[rgb]{ .973,  .427,  .42}21,01\% \\
		\midrule
		\textbf{Expe0} & \cellcolor[rgb]{ .973,  .459,  .427}7,64\% & \cellcolor[rgb]{ .435,  .757,  .482}8,17\% & \cellcolor[rgb]{ .992,  .835,  .498}-6,68\% & \cellcolor[rgb]{ .98,  .588,  .451}2,90\% & \cellcolor[rgb]{ .976,  .518,  .439}8,33\% & \cellcolor[rgb]{ .973,  .42,  .42}12,65\% & \cellcolor[rgb]{ .973,  .412,  .42}20,68\% \\
		\midrule
		\textbf{MV1} & \cellcolor[rgb]{ .984,  .671,  .467}10,11\% & \cellcolor[rgb]{ .945,  .906,  .514}10,75\% & \cellcolor[rgb]{ .992,  .788,  .49}-6,89\% & \cellcolor[rgb]{ .98,  .596,  .455}2,96\% & \cellcolor[rgb]{ .984,  .635,  .463}9,77\% & \cellcolor[rgb]{ .992,  .8,  .494}16,99\% & \cellcolor[rgb]{ .992,  .812,  .494}28,22\% \\
		\midrule
		\textbf{MAD1} & \cellcolor[rgb]{ .988,  .722,  .478}10,68\% & \cellcolor[rgb]{ .867,  .882,  .51}10,35\% & \cellcolor[rgb]{ 1,  .922,  .518}-6,33\% & \cellcolor[rgb]{ .992,  .78,  .49}4,14\% & \cellcolor[rgb]{ .988,  .729,  .478}10,88\% & \cellcolor[rgb]{ .992,  .792,  .49}16,91\% & \cellcolor[rgb]{ .992,  .816,  .494}28,23\% \\
		\midrule
		\textbf{CVaR1} & \cellcolor[rgb]{ .988,  .741,  .482}10,90\% & \cellcolor[rgb]{ .796,  .863,  .506}10,00\% & \cellcolor[rgb]{ .459,  .769,  .49}-4,60\% & \cellcolor[rgb]{ 1,  .922,  .518}5,02\% & \cellcolor[rgb]{ .984,  .686,  .471}10,38\% & \cellcolor[rgb]{ .988,  .749,  .482}16,40\% & \cellcolor[rgb]{ .992,  .796,  .49}27,86\% \\
		\midrule
		\textbf{Expe1} & \cellcolor[rgb]{ .988,  .702,  .475}10,45\% & \cellcolor[rgb]{ .875,  .882,  .51}10,39\% & \cellcolor[rgb]{ .996,  .898,  .51}-6,42\% & \cellcolor[rgb]{ .992,  .776,  .486}4,11\% & \cellcolor[rgb]{ .984,  .627,  .459}9,69\% & \cellcolor[rgb]{ .992,  .788,  .49}16,85\% & \cellcolor[rgb]{ .996,  .894,  .51}29,71\% \\
		\midrule
		\textbf{EW} & \cellcolor[rgb]{ .769,  .855,  .506}14,12\% & \cellcolor[rgb]{ .984,  .624,  .463}12,67\% & \cellcolor[rgb]{ .855,  .882,  .51}-5,86\% & \cellcolor[rgb]{ .667,  .827,  .502}6,39\% & \cellcolor[rgb]{ .737,  .847,  .506}14,19\% & \cellcolor[rgb]{ .867,  .886,  .514}19,96\% & \cellcolor[rgb]{ .533,  .788,  .494}38,41\% \\
		\midrule
		\textbf{RP} & \cellcolor[rgb]{ 1,  .922,  .518}12,97\% & \cellcolor[rgb]{ 1,  .906,  .518}11,11\% & \cellcolor[rgb]{ .608,  .812,  .498}-5,07\% & \cellcolor[rgb]{ .631,  .816,  .498}6,54\% & \cellcolor[rgb]{ 1,  .922,  .518}13,17\% & \cellcolor[rgb]{ 1,  .922,  .518}18,35\% & \cellcolor[rgb]{ .871,  .886,  .514}32,49\% \\
		\midrule
		\textbf{Index} & \cellcolor[rgb]{ .988,  .714,  .475}10,62\% & \cellcolor[rgb]{ .996,  .804,  .498}11,69\% & \cellcolor[rgb]{ .973,  .412,  .42}-8,48\% & \cellcolor[rgb]{ .984,  .671,  .467}3,44\% & \cellcolor[rgb]{ .984,  .635,  .463}9,77\% & \cellcolor[rgb]{ .996,  .875,  .506}17,84\% & \cellcolor[rgb]{ .992,  .922,  .518}30,38\% \\
		\midrule
		\textbf{DRvol0} & \cellcolor[rgb]{ .98,  .918,  .518}13,07\% & \cellcolor[rgb]{ .992,  .725,  .482}12,11\% & \cellcolor[rgb]{ .98,  .596,  .455}-7,69\% & \cellcolor[rgb]{ .996,  .886,  .51}4,82\% & \cellcolor[rgb]{ .62,  .812,  .498}14,64\% & \cellcolor[rgb]{ .882,  .89,  .514}19,78\% & \cellcolor[rgb]{ .965,  .914,  .518}30,87\% \\
		\midrule
		\textbf{DRMAD0} & \cellcolor[rgb]{ .957,  .91,  .518}13,19\% & \cellcolor[rgb]{ 1,  .922,  .518}11,02\% & \cellcolor[rgb]{ .804,  .867,  .51}-5,70\% & \cellcolor[rgb]{ .686,  .831,  .502}6,31\% & \cellcolor[rgb]{ .757,  .855,  .506}14,11\% & \cellcolor[rgb]{ .937,  .906,  .518}19,12\% & \cellcolor[rgb]{ 1,  .922,  .518}30,20\% \\
		\midrule
		\textbf{DRCVaR0} & \cellcolor[rgb]{ .737,  .847,  .506}14,29\% & \cellcolor[rgb]{ .984,  .576,  .451}12,95\% & \cellcolor[rgb]{ .996,  .867,  .506}-6,55\% & \cellcolor[rgb]{ .863,  .882,  .51}5,58\% & \cellcolor[rgb]{ .573,  .8,  .494}14,82\% & \cellcolor[rgb]{ .776,  .859,  .506}21,08\% & \cellcolor[rgb]{ .667,  .827,  .502}36,07\% \\
		\midrule
		\textbf{DRExpe0} & \cellcolor[rgb]{ .976,  .918,  .518}13,10\% & \cellcolor[rgb]{ .945,  .906,  .514}10,75\% & \cellcolor[rgb]{ .992,  .82,  .498}-6,75\% & \cellcolor[rgb]{ .443,  .761,  .486}7,31\% & \cellcolor[rgb]{ .871,  .886,  .514}13,68\% & \cellcolor[rgb]{ .867,  .882,  .51}19,99\% & \cellcolor[rgb]{ .996,  .91,  .514}30,00\% \\
		\midrule
		\textbf{DRvol1} & \cellcolor[rgb]{ .557,  .796,  .494}15,17\% & \cellcolor[rgb]{ .988,  .663,  .471}12,47\% & \cellcolor[rgb]{ .765,  .855,  .506}-5,57\% & \cellcolor[rgb]{ .388,  .745,  .482}7,52\% & \cellcolor[rgb]{ .388,  .745,  .482}15,52\% & \cellcolor[rgb]{ .58,  .804,  .494}23,40\% & \cellcolor[rgb]{ .663,  .824,  .498}36,20\% \\
		\midrule
		\textbf{DRMAD1} & \cellcolor[rgb]{ .533,  .788,  .494}15,28\% & \cellcolor[rgb]{ .988,  .667,  .471}12,43\% & \cellcolor[rgb]{ .82,  .871,  .51}-5,75\% & \cellcolor[rgb]{ .455,  .765,  .486}7,26\% & \cellcolor[rgb]{ .459,  .769,  .49}15,26\% & \cellcolor[rgb]{ .565,  .796,  .494}23,60\% & \cellcolor[rgb]{ .671,  .827,  .502}36,06\% \\
		\midrule
		\textbf{DRCVaR1} & \cellcolor[rgb]{ .388,  .745,  .482}16,00\% & \cellcolor[rgb]{ .973,  .412,  .42}13,84\% & \cellcolor[rgb]{ .388,  .745,  .482}-4,38\% & \cellcolor[rgb]{ .757,  .855,  .506}6,02\% & \cellcolor[rgb]{ .463,  .769,  .49}15,24\% & \cellcolor[rgb]{ .388,  .745,  .482}25,71\% & \cellcolor[rgb]{ .388,  .745,  .482}40,96\% \\
		\midrule
		\textbf{DRExpe1} & \cellcolor[rgb]{ .71,  .839,  .502}14,41\% & \cellcolor[rgb]{ .996,  .812,  .498}11,64\% & \cellcolor[rgb]{ .929,  .902,  .514}-6,10\% & \cellcolor[rgb]{ .475,  .773,  .49}7,17\% & \cellcolor[rgb]{ .588,  .804,  .494}14,76\% & \cellcolor[rgb]{ .706,  .839,  .502}21,92\% & \cellcolor[rgb]{ .796,  .863,  .506}33,85\% \\
		\bottomrule
	\end{tabular}}%
	\label{tab:ROI250_DJ}%
\end{table}%
%
\begin{table}[htbp!]
	\centering
	\caption{Annual ROI on the EuroStoxx50 dataset}
	\resizebox{0.7\textwidth}{!}{\begin{tabular}{|l|c|c|c|c|c|c|c|c|c|}
		\toprule
		\multicolumn{1}{|c|}{\textbf{Approach}} & \textbf{Mean} & \textbf{Vol} & \textbf{5\%-perc} & \textbf{25\%-perc} & \textbf{50\%-perc} & \textbf{75\%-perc} & \textbf{95\%-perc} \\
		\midrule
		\textbf{MV0} & \cellcolor[rgb]{ .996,  .906,  .514}13,40\% & \cellcolor[rgb]{ .729,  .843,  .502}12,57\% & \cellcolor[rgb]{ .525,  .784,  .49}-4,39\% & \cellcolor[rgb]{ .996,  .898,  .51}3,37\% & \cellcolor[rgb]{ .992,  .835,  .498}11,79\% & \cellcolor[rgb]{ .996,  .855,  .502}22,00\% & \cellcolor[rgb]{ 1,  .922,  .518}36,05\% \\
		\midrule
		\textbf{MAD0} & \cellcolor[rgb]{ .992,  .839,  .502}12,20\% & \cellcolor[rgb]{ .82,  .867,  .506}12,94\% & \cellcolor[rgb]{ .906,  .894,  .514}-8,63\% & \cellcolor[rgb]{ .996,  .878,  .506}2,96\% & \cellcolor[rgb]{ .992,  .835,  .498}11,83\% & \cellcolor[rgb]{ .992,  .824,  .498}21,49\% & \cellcolor[rgb]{ .988,  .741,  .482}33,23\% \\
		\midrule
		\textbf{CVaR0} & \cellcolor[rgb]{ .792,  .863,  .506}14,75\% & \cellcolor[rgb]{ .855,  .878,  .506}13,07\% & \cellcolor[rgb]{ .431,  .761,  .486}-3,34\% & \cellcolor[rgb]{ .918,  .898,  .514}4,35\% & \cellcolor[rgb]{ .996,  .894,  .51}13,00\% & \cellcolor[rgb]{ .569,  .8,  .494}25,18\% & \cellcolor[rgb]{ .906,  .894,  .514}37,81\% \\
		\midrule
		\textbf{Expe0} & \cellcolor[rgb]{ .996,  .855,  .502}12,45\% & \cellcolor[rgb]{ .388,  .745,  .482}11,20\% & \cellcolor[rgb]{ .51,  .78,  .49}-4,21\% & \cellcolor[rgb]{ 1,  .922,  .518}3,85\% & \cellcolor[rgb]{ .992,  .816,  .494}11,42\% & \cellcolor[rgb]{ .988,  .722,  .478}20,02\% & \cellcolor[rgb]{ .984,  .702,  .475}32,60\% \\
		\midrule
		\textbf{MV1} & \cellcolor[rgb]{ .992,  .831,  .498}12,01\% & \cellcolor[rgb]{ .694,  .831,  .498}12,43\% & \cellcolor[rgb]{ .996,  .882,  .51}-10,44\% & \cellcolor[rgb]{ .996,  .914,  .514}3,70\% & \cellcolor[rgb]{ .996,  .89,  .51}12,92\% & \cellcolor[rgb]{ .992,  .8,  .494}21,18\% & \cellcolor[rgb]{ .976,  .525,  .439}29,86\% \\
		\midrule
		\textbf{MAD1} & \cellcolor[rgb]{ .988,  .773,  .486}10,93\% & \cellcolor[rgb]{ .714,  .839,  .498}12,52\% & \cellcolor[rgb]{ .988,  .761,  .486}-12,86\% & \cellcolor[rgb]{ .996,  .878,  .51}3,02\% & \cellcolor[rgb]{ .992,  .839,  .502}11,94\% & \cellcolor[rgb]{ .988,  .722,  .478}20,05\% & \cellcolor[rgb]{ .973,  .478,  .431}29,09\% \\
		\midrule
		\textbf{CVaR1} & \cellcolor[rgb]{ .906,  .894,  .514}14,17\% & \cellcolor[rgb]{ .722,  .839,  .498}12,54\% & \cellcolor[rgb]{ .882,  .89,  .514}-8,34\% & \cellcolor[rgb]{ .522,  .784,  .49}6,63\% & \cellcolor[rgb]{ .816,  .871,  .51}14,67\% & \cellcolor[rgb]{ 1,  .922,  .518}22,92\% & \cellcolor[rgb]{ .988,  .761,  .486}33,51\% \\
		\midrule
		\textbf{Expe1} & \cellcolor[rgb]{ .992,  .78,  .49}11,09\% & \cellcolor[rgb]{ .443,  .761,  .482}11,43\% & \cellcolor[rgb]{ 1,  .922,  .518}-9,69\% & \cellcolor[rgb]{ .996,  .914,  .514}3,74\% & \cellcolor[rgb]{ .992,  .843,  .502}11,97\% & \cellcolor[rgb]{ .984,  .686,  .471}19,52\% & \cellcolor[rgb]{ .973,  .412,  .42}28,01\% \\
		\midrule
		\textbf{EW} & \cellcolor[rgb]{ .992,  .843,  .502}12,28\% & \cellcolor[rgb]{ .984,  .588,  .455}16,95\% & \cellcolor[rgb]{ .992,  .796,  .49}-12,17\% & \cellcolor[rgb]{ .984,  .682,  .471}-1,09\% & \cellcolor[rgb]{ 1,  .922,  .518}13,50\% & \cellcolor[rgb]{ .996,  .914,  .514}22,83\% & \cellcolor[rgb]{ .729,  .843,  .502}41,07\% \\
		\midrule
		\textbf{RP} & \cellcolor[rgb]{ .992,  .843,  .502}12,21\% & \cellcolor[rgb]{ .992,  .776,  .49}15,09\% & \cellcolor[rgb]{ .996,  .914,  .514}-9,84\% & \cellcolor[rgb]{ .988,  .729,  .478}-0,10\% & \cellcolor[rgb]{ 1,  .922,  .518}13,49\% & \cellcolor[rgb]{ .996,  .859,  .506}22,05\% & \cellcolor[rgb]{ .906,  .894,  .514}37,83\% \\
		\midrule
		\textbf{Index} & \cellcolor[rgb]{ .973,  .412,  .42}4,21\% & \cellcolor[rgb]{ .996,  .78,  .494}15,05\% & \cellcolor[rgb]{ .973,  .412,  .42}-19,97\% & \cellcolor[rgb]{ .973,  .412,  .42}-6,82\% & \cellcolor[rgb]{ .973,  .412,  .42}3,44\% & \cellcolor[rgb]{ .973,  .412,  .42}15,47\% & \cellcolor[rgb]{ .973,  .431,  .424}28,35\% \\
		\midrule
		\textbf{DRvol0} & \cellcolor[rgb]{ .388,  .745,  .482}16,84\% & \cellcolor[rgb]{ .992,  .749,  .486}15,37\% & \cellcolor[rgb]{ .388,  .745,  .482}-2,88\% & \cellcolor[rgb]{ .576,  .8,  .494}6,31\% & \cellcolor[rgb]{ .937,  .906,  .518}13,91\% & \cellcolor[rgb]{ .824,  .871,  .51}23,86\% & \cellcolor[rgb]{ .51,  .78,  .49}45,14\% \\
		\midrule
		\textbf{DRMAD0} & \cellcolor[rgb]{ .541,  .792,  .494}16,05\% & \cellcolor[rgb]{ .996,  .82,  .498}14,67\% & \cellcolor[rgb]{ .51,  .78,  .49}-4,19\% & \cellcolor[rgb]{ .878,  .886,  .514}4,57\% & \cellcolor[rgb]{ .812,  .867,  .51}14,71\% & \cellcolor[rgb]{ .624,  .816,  .498}24,90\% & \cellcolor[rgb]{ .631,  .816,  .498}42,88\% \\
		\midrule
		\textbf{DRCVaR0} & \cellcolor[rgb]{ .392,  .749,  .486}16,83\% & \cellcolor[rgb]{ .973,  .412,  .42}18,66\% & \cellcolor[rgb]{ .769,  .855,  .506}-7,10\% & \cellcolor[rgb]{ .824,  .871,  .51}4,87\% & \cellcolor[rgb]{ .996,  .898,  .51}13,07\% & \cellcolor[rgb]{ .6,  .808,  .498}25,01\% & \cellcolor[rgb]{ .388,  .745,  .482}47,33\% \\
		\midrule
		\textbf{DRExpe0} & \cellcolor[rgb]{ .522,  .784,  .49}16,16\% & \cellcolor[rgb]{ 1,  .922,  .518}13,65\% & \cellcolor[rgb]{ .522,  .784,  .49}-4,32\% & \cellcolor[rgb]{ .557,  .796,  .494}6,43\% & \cellcolor[rgb]{ .871,  .886,  .514}14,33\% & \cellcolor[rgb]{ .643,  .82,  .498}24,79\% & \cellcolor[rgb]{ .671,  .827,  .502}42,20\% \\
		\midrule
		\textbf{DRvol1} & \cellcolor[rgb]{ .749,  .851,  .506}14,97\% & \cellcolor[rgb]{ .996,  .835,  .502}14,51\% & \cellcolor[rgb]{ .996,  .91,  .514}-9,89\% & \cellcolor[rgb]{ .859,  .882,  .51}4,68\% & \cellcolor[rgb]{ .537,  .788,  .494}16,46\% & \cellcolor[rgb]{ .6,  .808,  .498}25,02\% & \cellcolor[rgb]{ .922,  .902,  .514}37,52\% \\
		\midrule
		\textbf{DRMAD1} & \cellcolor[rgb]{ 1,  .922,  .518}13,67\% & \cellcolor[rgb]{ .996,  .78,  .49}15,06\% & \cellcolor[rgb]{ .988,  .725,  .478}-13,59\% & \cellcolor[rgb]{ .996,  .878,  .506}3,00\% & \cellcolor[rgb]{ .631,  .816,  .498}15,84\% & \cellcolor[rgb]{ .596,  .808,  .498}25,05\% & \cellcolor[rgb]{ .996,  .878,  .51}35,42\% \\
		\midrule
		\textbf{DRCVaR1} & \cellcolor[rgb]{ .4,  .749,  .486}16,78\% & \cellcolor[rgb]{ .996,  .804,  .498}14,83\% & \cellcolor[rgb]{ .996,  .894,  .51}-10,20\% & \cellcolor[rgb]{ .388,  .745,  .482}7,39\% & \cellcolor[rgb]{ .388,  .745,  .482}17,39\% & \cellcolor[rgb]{ .388,  .745,  .482}26,11\% & \cellcolor[rgb]{ .682,  .831,  .502}41,93\% \\
		\midrule
		\textbf{DRExpe1} & \cellcolor[rgb]{ .973,  .914,  .518}13,81\% & \cellcolor[rgb]{ .957,  .91,  .514}13,49\% & \cellcolor[rgb]{ .996,  .859,  .506}-10,92\% & \cellcolor[rgb]{ .863,  .882,  .51}4,67\% & \cellcolor[rgb]{ .749,  .851,  .506}15,11\% & \cellcolor[rgb]{ .792,  .863,  .506}24,01\% & \cellcolor[rgb]{ .988,  .765,  .486}33,60\% \\
		\bottomrule
	\end{tabular}}%
	\label{tab:ROI250_EUXX50}%
\end{table}%
%
\begin{table}[htbp!]
	\centering
	\caption{Annual ROI on the FTSE100 dataset}
	\resizebox{0.7\textwidth}{!}{\begin{tabular}{|l|c|c|c|c|c|c|c|c|c|}
		\toprule
		\multicolumn{1}{|c|}{\textbf{Approach}} & \textbf{Mean} & \textbf{Vol} & \textbf{5\%-perc} & \textbf{25\%-perc} & \textbf{50\%-perc} & \textbf{75\%-perc} & \textbf{95\%-perc} \\
		\midrule
		\textbf{MV0} & \cellcolor[rgb]{ .992,  .804,  .494}12,03\% & \cellcolor[rgb]{ .718,  .839,  .498}13,95\% & \cellcolor[rgb]{ .851,  .878,  .51}-10,95\% & \cellcolor[rgb]{ .635,  .82,  .498}4,86\% & \cellcolor[rgb]{ .992,  .827,  .498}11,16\% & \cellcolor[rgb]{ .988,  .722,  .478}19,57\% & \cellcolor[rgb]{ .992,  .824,  .498}37,53\% \\
		\midrule
		\textbf{MAD0} & \cellcolor[rgb]{ .984,  .69,  .471}9,88\% & \cellcolor[rgb]{ .553,  .792,  .49}12,84\% & \cellcolor[rgb]{ .996,  .875,  .506}-12,07\% & \cellcolor[rgb]{ .98,  .918,  .518}3,39\% & \cellcolor[rgb]{ .988,  .706,  .475}9,18\% & \cellcolor[rgb]{ .984,  .627,  .459}17,21\% & \cellcolor[rgb]{ .988,  .718,  .478}33,86\% \\
		\midrule
		\textbf{CVaR0} & \cellcolor[rgb]{ .992,  .804,  .494}11,98\% & \cellcolor[rgb]{ .898,  .89,  .51}15,18\% & \cellcolor[rgb]{ .973,  .431,  .424}-15,11\% & \cellcolor[rgb]{ .996,  .878,  .506}2,81\% & \cellcolor[rgb]{ 1,  .922,  .518}12,66\% & \cellcolor[rgb]{ .992,  .808,  .494}21,77\% & \cellcolor[rgb]{ .992,  .835,  .498}37,87\% \\
		\midrule
		\textbf{Expe0} & \cellcolor[rgb]{ .996,  .863,  .506}13,08\% & \cellcolor[rgb]{ .808,  .867,  .506}14,57\% & \cellcolor[rgb]{ .98,  .573,  .447}-14,16\% & \cellcolor[rgb]{ .525,  .784,  .49}5,35\% & \cellcolor[rgb]{ .996,  .918,  .514}12,63\% & \cellcolor[rgb]{ .992,  .812,  .494}21,85\% & \cellcolor[rgb]{ .992,  .843,  .502}38,17\% \\
		\midrule
		\textbf{MV1} & \cellcolor[rgb]{ .8,  .867,  .51}15,42\% & \cellcolor[rgb]{ .89,  .89,  .51}15,13\% & \cellcolor[rgb]{ .925,  .902,  .514}-11,34\% & \cellcolor[rgb]{ .388,  .745,  .482}5,92\% & \cellcolor[rgb]{ .471,  .769,  .49}17,72\% & \cellcolor[rgb]{ .894,  .894,  .514}25,94\% & \cellcolor[rgb]{ .992,  .827,  .498}37,68\% \\
		\midrule
		\textbf{MAD1} & \cellcolor[rgb]{ .988,  .918,  .518}14,25\% & \cellcolor[rgb]{ .831,  .871,  .506}14,73\% & \cellcolor[rgb]{ .969,  .914,  .518}-11,58\% & \cellcolor[rgb]{ .941,  .906,  .518}3,56\% & \cellcolor[rgb]{ .6,  .808,  .498}16,48\% & \cellcolor[rgb]{ .89,  .89,  .514}25,98\% & \cellcolor[rgb]{ .988,  .753,  .482}35,11\% \\
		\midrule
		\textbf{CVaR1} & \cellcolor[rgb]{ .996,  .863,  .506}13,08\% & \cellcolor[rgb]{ 1,  .922,  .518}15,85\% & \cellcolor[rgb]{ .973,  .42,  .42}-15,20\% & \cellcolor[rgb]{ .992,  .796,  .49}1,82\% & \cellcolor[rgb]{ .643,  .82,  .498}16,09\% & \cellcolor[rgb]{ .996,  .918,  .514}24,56\% & \cellcolor[rgb]{ .992,  .788,  .49}36,23\% \\
		\midrule
		\textbf{Expe1} & \cellcolor[rgb]{ 1,  .922,  .518}14,17\% & \cellcolor[rgb]{ .996,  .812,  .498}16,55\% & \cellcolor[rgb]{ .984,  .675,  .467}-13,46\% & \cellcolor[rgb]{ .992,  .835,  .498}2,27\% & \cellcolor[rgb]{ .737,  .847,  .506}15,16\% & \cellcolor[rgb]{ .961,  .914,  .518}25,08\% & \cellcolor[rgb]{ .969,  .914,  .518}41,25\% \\
		\midrule
		\textbf{EW} & \cellcolor[rgb]{ .996,  .851,  .502}12,86\% & \cellcolor[rgb]{ 1,  .894,  .514}16,04\% & \cellcolor[rgb]{ .455,  .765,  .486}-8,86\% & \cellcolor[rgb]{ .988,  .714,  .475}0,83\% & \cellcolor[rgb]{ .996,  .867,  .506}11,79\% & \cellcolor[rgb]{ .992,  .831,  .498}22,32\% & \cellcolor[rgb]{ .906,  .894,  .514}42,05\% \\
		\midrule
		\textbf{RP} & \cellcolor[rgb]{ .992,  .812,  .494}12,17\% & \cellcolor[rgb]{ .78,  .859,  .502}14,38\% & \cellcolor[rgb]{ .388,  .745,  .482}-8,50\% & \cellcolor[rgb]{ .992,  .78,  .49}1,61\% & \cellcolor[rgb]{ .992,  .827,  .498}11,15\% & \cellcolor[rgb]{ .992,  .788,  .49}21,24\% & \cellcolor[rgb]{ .992,  .843,  .502}38,22\% \\
		\midrule
		\textbf{Index} & \cellcolor[rgb]{ .973,  .412,  .42}4,59\% & \cellcolor[rgb]{ .388,  .745,  .482}11,72\% & \cellcolor[rgb]{ .973,  .459,  .427}-14,94\% & \cellcolor[rgb]{ .973,  .412,  .42}-2,82\% & \cellcolor[rgb]{ .973,  .412,  .42}4,33\% & \cellcolor[rgb]{ .973,  .412,  .42}11,62\% & \cellcolor[rgb]{ .973,  .412,  .42}23,15\% \\
		\midrule
		\textbf{DRvol0} & \cellcolor[rgb]{ .835,  .875,  .51}15,20\% & \cellcolor[rgb]{ .973,  .412,  .42}19,08\% & \cellcolor[rgb]{ .922,  .902,  .514}-11,33\% & \cellcolor[rgb]{ .906,  .898,  .514}3,71\% & \cellcolor[rgb]{ .996,  .906,  .514}12,43\% & \cellcolor[rgb]{ 1,  .922,  .518}24,59\% & \cellcolor[rgb]{ .714,  .839,  .502}44,53\% \\
		\midrule
		\textbf{DRMAD0} & \cellcolor[rgb]{ .992,  .824,  .498}12,39\% & \cellcolor[rgb]{ .949,  .906,  .514}15,51\% & \cellcolor[rgb]{ 1,  .922,  .518}-11,76\% & \cellcolor[rgb]{ 1,  .922,  .518}3,30\% & \cellcolor[rgb]{ .992,  .784,  .49}10,48\% & \cellcolor[rgb]{ .992,  .788,  .49}21,26\% & \cellcolor[rgb]{ 1,  .922,  .518}40,82\% \\
		\midrule
		\textbf{DRCVaR0} & \cellcolor[rgb]{ .804,  .867,  .51}15,39\% & \cellcolor[rgb]{ .976,  .424,  .424}19,02\% & \cellcolor[rgb]{ .714,  .839,  .502}-10,23\% & \cellcolor[rgb]{ .988,  .765,  .486}1,45\% & \cellcolor[rgb]{ .831,  .875,  .51}14,30\% & \cellcolor[rgb]{ .937,  .906,  .518}25,38\% & \cellcolor[rgb]{ .596,  .808,  .498}46,02\% \\
		\midrule
		\textbf{DRExpe0} & \cellcolor[rgb]{ .859,  .882,  .51}15,05\% & \cellcolor[rgb]{ .996,  .808,  .498}16,58\% & \cellcolor[rgb]{ .4,  .749,  .486}-8,56\% & \cellcolor[rgb]{ .984,  .918,  .518}3,37\% & \cellcolor[rgb]{ .996,  .902,  .514}12,34\% & \cellcolor[rgb]{ .839,  .875,  .51}26,64\% & \cellcolor[rgb]{ .725,  .843,  .502}44,35\% \\
		\midrule
		\textbf{DRvol1} & \cellcolor[rgb]{ .388,  .745,  .482}17,96\% & \cellcolor[rgb]{ .976,  .463,  .431}18,77\% & \cellcolor[rgb]{ .98,  .561,  .447}-14,24\% & \cellcolor[rgb]{ .714,  .839,  .502}4,55\% & \cellcolor[rgb]{ .573,  .8,  .494}16,74\% & \cellcolor[rgb]{ .388,  .745,  .482}32,27\% & \cellcolor[rgb]{ .388,  .745,  .482}48,67\% \\
		\midrule
		\textbf{DRMAD1} & \cellcolor[rgb]{ .467,  .769,  .49}17,48\% & \cellcolor[rgb]{ 1,  .851,  .506}16,29\% & \cellcolor[rgb]{ .792,  .863,  .506}-10,63\% & \cellcolor[rgb]{ .647,  .824,  .498}4,82\% & \cellcolor[rgb]{ .388,  .745,  .482}18,49\% & \cellcolor[rgb]{ .533,  .788,  .494}30,45\% & \cellcolor[rgb]{ .984,  .918,  .518}41,04\% \\
		\midrule
		\textbf{DRCVaR1} & \cellcolor[rgb]{ .757,  .851,  .506}15,68\% & \cellcolor[rgb]{ .984,  .62,  .463}17,77\% & \cellcolor[rgb]{ .973,  .412,  .42}-15,27\% & \cellcolor[rgb]{ .992,  .835,  .498}2,30\% & \cellcolor[rgb]{ .584,  .804,  .494}16,62\% & \cellcolor[rgb]{ .651,  .824,  .498}28,99\% & \cellcolor[rgb]{ .886,  .89,  .514}42,30\% \\
		\midrule
		\textbf{DRExpe1} & \cellcolor[rgb]{ .627,  .816,  .498}16,49\% & \cellcolor[rgb]{ .98,  .549,  .447}18,22\% & \cellcolor[rgb]{ .973,  .431,  .42}-15,14\% & \cellcolor[rgb]{ .992,  .827,  .498}2,17\% & \cellcolor[rgb]{ .663,  .824,  .498}15,90\% & \cellcolor[rgb]{ .584,  .804,  .494}29,83\% & \cellcolor[rgb]{ .694,  .835,  .502}44,76\% \\
		\bottomrule
	\end{tabular}}%
	\label{tab:ROI250_FTSE100}%
\end{table}%
%
\begin{table}[htbp!]
	\centering
	\caption{Annual ROI on the NASDAQ100 dataset}
	\resizebox{0.7\textwidth}{!}{\begin{tabular}{|l|c|c|c|c|c|c|c|c|c|}
		\toprule
		\multicolumn{1}{|c|}{\textbf{Approach}} & \textbf{Mean} & \textbf{Vol} & \textbf{5\%-perc} & \textbf{25\%-perc} & \textbf{50\%-perc} & \textbf{75\%-perc} & \textbf{95\%-perc} \\
		\midrule
		\textbf{MV0} & \cellcolor[rgb]{ .973,  .431,  .42}12,90\% & \cellcolor[rgb]{ .388,  .745,  .482}8,07\% & \cellcolor[rgb]{ .388,  .745,  .482}-0,85\% & \cellcolor[rgb]{ .976,  .529,  .439}8,05\% & \cellcolor[rgb]{ .973,  .463,  .427}13,60\% & \cellcolor[rgb]{ .973,  .435,  .424}18,64\% & \cellcolor[rgb]{ .973,  .412,  .42}25,11\% \\
		\midrule
		\textbf{MAD0} & \cellcolor[rgb]{ .973,  .412,  .42}12,41\% & \cellcolor[rgb]{ .388,  .745,  .482}8,01\% & \cellcolor[rgb]{ .459,  .769,  .49}-1,45\% & \cellcolor[rgb]{ .973,  .412,  .42}7,29\% & \cellcolor[rgb]{ .973,  .412,  .42}12,62\% & \cellcolor[rgb]{ .973,  .412,  .42}17,94\% & \cellcolor[rgb]{ .973,  .412,  .42}25,20\% \\
		\midrule
		\textbf{CVaR0} & \cellcolor[rgb]{ .976,  .529,  .439}15,31\% & \cellcolor[rgb]{ .439,  .757,  .482}9,07\% & \cellcolor[rgb]{ .416,  .753,  .486}-1,08\% & \cellcolor[rgb]{ .992,  .816,  .494}9,90\% & \cellcolor[rgb]{ .984,  .631,  .459}16,81\% & \cellcolor[rgb]{ .976,  .541,  .443}21,61\% & \cellcolor[rgb]{ .973,  .467,  .427}28,68\% \\
		\midrule
		\textbf{Expe0} & \cellcolor[rgb]{ .976,  .486,  .431}14,26\% & \cellcolor[rgb]{ .424,  .753,  .482}8,80\% & \cellcolor[rgb]{ .396,  .749,  .486}-0,90\% & \cellcolor[rgb]{ .984,  .635,  .463}8,73\% & \cellcolor[rgb]{ .976,  .541,  .443}15,13\% & \cellcolor[rgb]{ .976,  .506,  .435}20,61\% & \cellcolor[rgb]{ .973,  .439,  .424}26,95\% \\
		\midrule
		\textbf{MV1} & \cellcolor[rgb]{ .851,  .878,  .51}27,58\% & \cellcolor[rgb]{ .98,  .514,  .439}27,26\% & \cellcolor[rgb]{ .976,  .498,  .435}-12,65\% & \cellcolor[rgb]{ .984,  .631,  .459}8,72\% & \cellcolor[rgb]{ .933,  .902,  .514}24,02\% & \cellcolor[rgb]{ .659,  .824,  .498}44,76\% & \cellcolor[rgb]{ .541,  .792,  .494}79,51\% \\
		\midrule
		\textbf{MAD1} & \cellcolor[rgb]{ .835,  .875,  .51}27,88\% & \cellcolor[rgb]{ .984,  .565,  .451}26,41\% & \cellcolor[rgb]{ .98,  .569,  .447}-11,56\% & \cellcolor[rgb]{ .988,  .761,  .486}9,55\% & \cellcolor[rgb]{ .98,  .918,  .518}22,86\% & \cellcolor[rgb]{ .675,  .827,  .502}44,22\% & \cellcolor[rgb]{ .522,  .784,  .49}80,37\% \\
		\midrule
		\textbf{CVaR1} & \cellcolor[rgb]{ .957,  .91,  .518}25,59\% & \cellcolor[rgb]{ .992,  .71,  .478}24,01\% & \cellcolor[rgb]{ .98,  .624,  .459}-10,69\% & \cellcolor[rgb]{ .976,  .541,  .443}8,14\% & \cellcolor[rgb]{ .992,  .922,  .518}22,52\% & \cellcolor[rgb]{ .792,  .863,  .506}39,93\% & \cellcolor[rgb]{ .745,  .851,  .506}70,00\% \\
		\midrule
		\textbf{Expe1} & \cellcolor[rgb]{ .992,  .922,  .518}24,92\% & \cellcolor[rgb]{ 1,  .902,  .514}20,89\% & \cellcolor[rgb]{ .988,  .753,  .482}-8,70\% & \cellcolor[rgb]{ .996,  .875,  .506}10,29\% & \cellcolor[rgb]{ .949,  .91,  .518}23,61\% & \cellcolor[rgb]{ .831,  .875,  .51}38,51\% & \cellcolor[rgb]{ .91,  .898,  .514}62,29\% \\
		\midrule
		\textbf{EW} & \cellcolor[rgb]{ .992,  .835,  .498}22,68\% & \cellcolor[rgb]{ .788,  .859,  .502}16,20\% & \cellcolor[rgb]{ 1,  .922,  .518}-6,09\% & \cellcolor[rgb]{ .78,  .859,  .506}12,24\% & \cellcolor[rgb]{ .961,  .914,  .518}23,27\% & \cellcolor[rgb]{ 1,  .922,  .518}32,25\% & \cellcolor[rgb]{ .992,  .804,  .494}50,62\% \\
		\midrule
		\textbf{RP} & \cellcolor[rgb]{ .988,  .761,  .486}20,89\% & \cellcolor[rgb]{ .675,  .827,  .498}13,87\% & \cellcolor[rgb]{ .824,  .871,  .51}-4,56\% & \cellcolor[rgb]{ .737,  .847,  .506}12,56\% & \cellcolor[rgb]{ .996,  .886,  .51}21,68\% & \cellcolor[rgb]{ .992,  .788,  .49}28,54\% & \cellcolor[rgb]{ .988,  .718,  .478}44,86\% \\
		\midrule
		\textbf{Index} & \cellcolor[rgb]{ .984,  .659,  .467}18,44\% & \cellcolor[rgb]{ .839,  .875,  .506}17,28\% & \cellcolor[rgb]{ .976,  .533,  .443}-12,10\% & \cellcolor[rgb]{ .976,  .541,  .443}8,14\% & \cellcolor[rgb]{ .988,  .765,  .486}19,34\% & \cellcolor[rgb]{ .992,  .78,  .49}28,29\% & \cellcolor[rgb]{ .988,  .753,  .482}47,22\% \\
		\midrule
		\textbf{DRvol0} & \cellcolor[rgb]{ 1,  .922,  .518}24,74\% & \cellcolor[rgb]{ 1,  .922,  .518}20,51\% & \cellcolor[rgb]{ .871,  .886,  .514}-4,97\% & \cellcolor[rgb]{ .812,  .867,  .51}12,02\% & \cellcolor[rgb]{ 1,  .922,  .518}22,29\% & \cellcolor[rgb]{ .953,  .91,  .518}33,98\% & \cellcolor[rgb]{ 1,  .922,  .518}58,01\% \\
		\midrule
		\textbf{DRMAD0} & \cellcolor[rgb]{ .992,  .776,  .49}21,31\% & \cellcolor[rgb]{ .733,  .843,  .502}15,14\% & \cellcolor[rgb]{ .459,  .765,  .486}-1,43\% & \cellcolor[rgb]{ .953,  .91,  .518}10,93\% & \cellcolor[rgb]{ .992,  .776,  .486}19,57\% & \cellcolor[rgb]{ .996,  .863,  .506}30,70\% & \cellcolor[rgb]{ .988,  .733,  .478}46,05\% \\
		\midrule
		\textbf{DRCVaR0} & \cellcolor[rgb]{ .996,  .922,  .518}24,82\% & \cellcolor[rgb]{ 1,  .918,  .518}20,63\% & \cellcolor[rgb]{ .761,  .855,  .506}-4,01\% & \cellcolor[rgb]{ .788,  .863,  .506}12,20\% & \cellcolor[rgb]{ .996,  .89,  .51}21,70\% & \cellcolor[rgb]{ .996,  .906,  .514}31,88\% & \cellcolor[rgb]{ .957,  .91,  .518}60,05\% \\
		\midrule
		\textbf{DRExpe0} & \cellcolor[rgb]{ .992,  .8,  .494}21,81\% & \cellcolor[rgb]{ .839,  .875,  .506}17,31\% & \cellcolor[rgb]{ .659,  .824,  .498}-3,15\% & \cellcolor[rgb]{ 1,  .922,  .518}10,57\% & \cellcolor[rgb]{ .992,  .812,  .494}20,24\% & \cellcolor[rgb]{ .996,  .875,  .506}30,96\% & \cellcolor[rgb]{ .992,  .827,  .498}52,16\% \\
		\midrule
		\textbf{DRvol1} & \cellcolor[rgb]{ .388,  .745,  .482}36,20\% & \cellcolor[rgb]{ .973,  .412,  .42}28,89\% & \cellcolor[rgb]{ .973,  .412,  .42}-14,01\% & \cellcolor[rgb]{ .388,  .745,  .482}15,18\% & \cellcolor[rgb]{ .388,  .745,  .482}37,32\% & \cellcolor[rgb]{ .388,  .745,  .482}54,65\% & \cellcolor[rgb]{ .388,  .745,  .482}86,52\% \\
		\midrule
		\textbf{DRMAD1} & \cellcolor[rgb]{ .655,  .824,  .498}31,28\% & \cellcolor[rgb]{ .984,  .576,  .455}26,20\% & \cellcolor[rgb]{ .976,  .533,  .443}-12,07\% & \cellcolor[rgb]{ .722,  .843,  .502}12,69\% & \cellcolor[rgb]{ .631,  .816,  .498}31,42\% & \cellcolor[rgb]{ .561,  .796,  .494}48,36\% & \cellcolor[rgb]{ .58,  .8,  .494}77,71\% \\
		\midrule
		\textbf{DRCVaR1} & \cellcolor[rgb]{ .694,  .835,  .502}30,50\% & \cellcolor[rgb]{ .98,  .549,  .447}26,64\% & \cellcolor[rgb]{ .976,  .549,  .443}-11,86\% & \cellcolor[rgb]{ .914,  .898,  .514}11,22\% & \cellcolor[rgb]{ .733,  .847,  .506}28,88\% & \cellcolor[rgb]{ .6,  .808,  .498}46,97\% & \cellcolor[rgb]{ .592,  .804,  .494}77,05\% \\
		\midrule
		\textbf{DRExpe1} & \cellcolor[rgb]{ .745,  .851,  .506}29,57\% & \cellcolor[rgb]{ .988,  .686,  .475}24,41\% & \cellcolor[rgb]{ .996,  .851,  .502}-7,18\% & \cellcolor[rgb]{ .855,  .882,  .51}11,67\% & \cellcolor[rgb]{ .725,  .843,  .502}29,07\% & \cellcolor[rgb]{ .694,  .835,  .502}43,59\% & \cellcolor[rgb]{ .655,  .824,  .498}74,10\% \\
		\bottomrule
	\end{tabular}}%
	\label{tab:ROI250_NASDAQ100}%
\end{table}%
%
\begin{table}[htbp!]
	\centering
	\caption{Annual ROI on the SP500 dataset}
	\resizebox{0.7\textwidth}{!}{\begin{tabular}{|l|c|c|c|c|c|c|c|c|c|}
		\toprule
		\multicolumn{1}{|c|}{\textbf{Approach}} & \textbf{Mean} & \textbf{Vol} & \textbf{5\%-perc} & \textbf{25\%-perc} & \textbf{50\%-perc} & \textbf{75\%-perc} & \textbf{95\%-perc} \\
		\midrule
		\textbf{MV0} & \cellcolor[rgb]{ .973,  .439,  .424}9,27\% & \cellcolor[rgb]{ .388,  .745,  .482}7,38\% & \cellcolor[rgb]{ .608,  .808,  .498}-3,94\% & \cellcolor[rgb]{ .98,  .624,  .459}5,64\% & \cellcolor[rgb]{ .976,  .514,  .439}10,15\% & \cellcolor[rgb]{ .973,  .412,  .42}13,75\% & \cellcolor[rgb]{ .973,  .412,  .42}18,78\% \\
		\midrule
		\textbf{MAD0} & \cellcolor[rgb]{ .973,  .412,  .42}8,72\% & \cellcolor[rgb]{ .42,  .753,  .482}7,79\% & \cellcolor[rgb]{ .537,  .788,  .494}-3,49\% & \cellcolor[rgb]{ .973,  .412,  .42}3,28\% & \cellcolor[rgb]{ .973,  .435,  .424}8,75\% & \cellcolor[rgb]{ .973,  .416,  .42}13,87\% & \cellcolor[rgb]{ .973,  .447,  .424}20,61\% \\
		\midrule
		\textbf{CVaR0} & \cellcolor[rgb]{ .973,  .439,  .424}9,24\% & \cellcolor[rgb]{ .424,  .753,  .482}7,83\% & \cellcolor[rgb]{ .388,  .745,  .482}-2,50\% & \cellcolor[rgb]{ .976,  .486,  .431}4,12\% & \cellcolor[rgb]{ .973,  .412,  .42}8,32\% & \cellcolor[rgb]{ .973,  .431,  .42}14,22\% & \cellcolor[rgb]{ .976,  .486,  .431}22,52\% \\
		\midrule
		\textbf{Expe0} & \cellcolor[rgb]{ .976,  .518,  .439}10,64\% & \cellcolor[rgb]{ .475,  .769,  .486}8,44\% & \cellcolor[rgb]{ 1,  .922,  .518}-6,56\% & \cellcolor[rgb]{ .988,  .722,  .478}6,76\% & \cellcolor[rgb]{ .98,  .565,  .447}11,03\% & \cellcolor[rgb]{ .973,  .467,  .427}15,06\% & \cellcolor[rgb]{ .976,  .514,  .439}23,84\% \\
		\midrule
		\textbf{MV1} & \cellcolor[rgb]{ .694,  .835,  .502}23,14\% & \cellcolor[rgb]{ .973,  .412,  .42}24,85\% & \cellcolor[rgb]{ .973,  .412,  .42}-15,06\% & \cellcolor[rgb]{ 1,  .922,  .518}8,94\% & \cellcolor[rgb]{ .769,  .855,  .506}21,25\% & \cellcolor[rgb]{ .729,  .847,  .506}32,59\% & \cellcolor[rgb]{ .455,  .765,  .486}68,30\% \\
		\midrule
		\textbf{MAD1} & \cellcolor[rgb]{ .773,  .859,  .506}21,79\% & \cellcolor[rgb]{ .98,  .529,  .443}22,51\% & \cellcolor[rgb]{ .976,  .49,  .431}-13,71\% & \cellcolor[rgb]{ .996,  .91,  .514}8,82\% & \cellcolor[rgb]{ .839,  .875,  .51}20,07\% & \cellcolor[rgb]{ .816,  .871,  .51}30,44\% & \cellcolor[rgb]{ .502,  .78,  .49}66,15\% \\
		\midrule
		\textbf{CVaR1} & \cellcolor[rgb]{ .773,  .859,  .506}21,81\% & \cellcolor[rgb]{ .98,  .49,  .435}23,36\% & \cellcolor[rgb]{ .98,  .612,  .455}-11,72\% & \cellcolor[rgb]{ .894,  .894,  .514}9,97\% & \cellcolor[rgb]{ .973,  .914,  .518}17,69\% & \cellcolor[rgb]{ .937,  .906,  .518}27,48\% & \cellcolor[rgb]{ .431,  .761,  .486}69,45\% \\
		\midrule
		\textbf{Expe1} & \cellcolor[rgb]{ .749,  .851,  .506}22,20\% & \cellcolor[rgb]{ .984,  .62,  .463}20,74\% & \cellcolor[rgb]{ .984,  .647,  .463}-11,09\% & \cellcolor[rgb]{ .878,  .886,  .514}10,11\% & \cellcolor[rgb]{ .808,  .867,  .51}20,60\% & \cellcolor[rgb]{ .82,  .871,  .51}30,38\% & \cellcolor[rgb]{ .561,  .796,  .494}63,56\% \\
		\midrule
		\textbf{EW} & \cellcolor[rgb]{ .996,  .878,  .506}17,10\% & \cellcolor[rgb]{ 1,  .851,  .506}16,08\% & \cellcolor[rgb]{ .91,  .898,  .514}-5,95\% & \cellcolor[rgb]{ .988,  .702,  .475}6,54\% & \cellcolor[rgb]{ .996,  .867,  .506}16,25\% & \cellcolor[rgb]{ .996,  .863,  .506}24,48\% & \cellcolor[rgb]{ .89,  .89,  .514}48,56\% \\
		\midrule
		\textbf{RP} & \cellcolor[rgb]{ .992,  .792,  .49}15,55\% & \cellcolor[rgb]{ .894,  .89,  .51}13,40\% & \cellcolor[rgb]{ .812,  .867,  .51}-5,29\% & \cellcolor[rgb]{ .992,  .776,  .49}7,36\% & \cellcolor[rgb]{ .992,  .792,  .49}14,94\% & \cellcolor[rgb]{ .988,  .753,  .482}21,86\% & \cellcolor[rgb]{ .996,  .867,  .506}40,94\% \\
		\midrule
		\textbf{Index} & \cellcolor[rgb]{ .98,  .588,  .451}11,94\% & \cellcolor[rgb]{ .827,  .871,  .506}12,61\% & \cellcolor[rgb]{ .988,  .718,  .478}-9,95\% & \cellcolor[rgb]{ .973,  .467,  .427}3,92\% & \cellcolor[rgb]{ .984,  .639,  .463}12,28\% & \cellcolor[rgb]{ .98,  .612,  .455}18,58\% & \cellcolor[rgb]{ .988,  .722,  .478}33,82\% \\
		\midrule
		\textbf{DRvol0} & \cellcolor[rgb]{ .875,  .886,  .514}20,05\% & \cellcolor[rgb]{ 1,  .922,  .518}14,66\% & \cellcolor[rgb]{ .8,  .867,  .51}-5,22\% & \cellcolor[rgb]{ .769,  .855,  .506}11,17\% & \cellcolor[rgb]{ .867,  .882,  .51}19,58\% & \cellcolor[rgb]{ .867,  .882,  .51}29,26\% & \cellcolor[rgb]{ 1,  .922,  .518}43,42\% \\
		\midrule
		\textbf{DRMAD0} & \cellcolor[rgb]{ .996,  .859,  .506}16,80\% & \cellcolor[rgb]{ .816,  .867,  .506}12,51\% & \cellcolor[rgb]{ .929,  .902,  .514}-6,09\% & \cellcolor[rgb]{ 1,  .922,  .518}8,98\% & \cellcolor[rgb]{ 1,  .922,  .518}17,18\% & \cellcolor[rgb]{ .996,  .878,  .506}24,83\% & \cellcolor[rgb]{ .992,  .784,  .49}36,95\% \\
		\midrule
		\textbf{DRCVaR0} & \cellcolor[rgb]{ .996,  .898,  .514}17,49\% & \cellcolor[rgb]{ .969,  .91,  .514}14,32\% & \cellcolor[rgb]{ .647,  .82,  .498}-4,21\% & \cellcolor[rgb]{ .992,  .8,  .494}7,62\% & \cellcolor[rgb]{ .992,  .827,  .498}15,59\% & \cellcolor[rgb]{ 1,  .922,  .518}25,85\% & \cellcolor[rgb]{ .996,  .918,  .514}43,40\% \\
		\midrule
		\textbf{DRExpe0} & \cellcolor[rgb]{ 1,  .922,  .518}17,85\% & \cellcolor[rgb]{ .847,  .875,  .506}12,86\% & \cellcolor[rgb]{ .616,  .812,  .498}-4,00\% & \cellcolor[rgb]{ .859,  .882,  .51}10,30\% & \cellcolor[rgb]{ .996,  .878,  .51}16,47\% & \cellcolor[rgb]{ .996,  .867,  .506}24,57\% & \cellcolor[rgb]{ .996,  .894,  .51}42,14\% \\
		\midrule
		\textbf{DRvol1} & \cellcolor[rgb]{ .388,  .745,  .482}28,40\% & \cellcolor[rgb]{ .98,  .557,  .447}22,02\% & \cellcolor[rgb]{ .988,  .733,  .478}-9,67\% & \cellcolor[rgb]{ .388,  .745,  .482}14,79\% & \cellcolor[rgb]{ .388,  .745,  .482}27,93\% & \cellcolor[rgb]{ .388,  .745,  .482}41,04\% & \cellcolor[rgb]{ .388,  .745,  .482}71,32\% \\
		\midrule
		\textbf{DRMAD1} & \cellcolor[rgb]{ .522,  .784,  .49}26,15\% & \cellcolor[rgb]{ .984,  .608,  .459}20,96\% & \cellcolor[rgb]{ .98,  .612,  .455}-11,69\% & \cellcolor[rgb]{ .553,  .796,  .494}13,22\% & \cellcolor[rgb]{ .416,  .753,  .486}27,48\% & \cellcolor[rgb]{ .467,  .769,  .49}39,13\% & \cellcolor[rgb]{ .529,  .788,  .494}65,03\% \\
		\midrule
		\textbf{DRCVaR1} & \cellcolor[rgb]{ .659,  .824,  .498}23,76\% & \cellcolor[rgb]{ .988,  .643,  .467}20,25\% & \cellcolor[rgb]{ .988,  .741,  .482}-9,52\% & \cellcolor[rgb]{ .82,  .871,  .51}10,70\% & \cellcolor[rgb]{ .651,  .824,  .498}23,32\% & \cellcolor[rgb]{ .663,  .827,  .502}34,27\% & \cellcolor[rgb]{ .592,  .804,  .494}62,16\% \\
		\midrule
		\textbf{DRExpe1} & \cellcolor[rgb]{ .635,  .82,  .498}24,16\% & \cellcolor[rgb]{ .988,  .647,  .467}20,21\% & \cellcolor[rgb]{ .992,  .808,  .494}-8,44\% & \cellcolor[rgb]{ .898,  .894,  .514}9,93\% & \cellcolor[rgb]{ .608,  .808,  .498}24,12\% & \cellcolor[rgb]{ .588,  .804,  .494}36,10\% & \cellcolor[rgb]{ .643,  .82,  .498}59,85\% \\
		\bottomrule
	\end{tabular}}%
	\label{tab:ROI250_SP500}%
\end{table}%
%


\end{document}